\documentclass{amsart}
\usepackage{fullpage,graphicx,subfigure,mathpazo,color}
\usepackage{amsmath,amscd,tikz,mathrsfs}
\usepackage[normalem]{ulem}
\usepackage{amsmath}
\usepackage{setspace}
\usepackage{multirow}
\usepackage{extpfeil,extarrows}
\usepackage{graphbox}
\usepackage{xcolor}
\definecolor{Darkgreen}{RGB}{34,139,34}

\newcommand{\ii}{\mathrm{i}}
\newcommand{\ee}{\mathrm{e}}
\newcommand{\dd}{\mathrm{d}}
\newcommand{\cn}{\mathrm{cn}}
\newcommand{\dn}{\mathrm{dn}}
\newcommand{\sn}{\mathrm{sn}}
\newcommand{\tn}{\mathrm{tn}}
\newcommand{\nc}{\mathrm{nc}}
\newcommand{\dc}{\mathrm{dc}}
\newcommand{\cs}{\mathrm{cs}}
\newcommand{\ns}{\mathrm{ns}}
\newcommand{\ds}{\mathrm{ds}}
\newcommand{\cd}{\mathrm{cd}}
\newcommand{\sd}{\mathrm{sd}}
\newcommand{\nd}{\mathrm{nd}}

\newtheorem{theorem}{Theorem}
\newtheorem{lemma}{Lemma}
\newtheorem{define}{Definition}
\newtheorem{prop}{Proposition}

\newtheorem{definition}{Definition}

\newtheorem{remark}{Remark}

\usepackage{graphicx}      
\usepackage{titlesec}

\titleformat{\title}{\centering\LARGE\bfseries}{\thesection}{1em}{}
\titleformat{\section}{\centering\LARGE\bfseries}{\thesection}{1em}{}
\titleformat{\subsection}{\Large\bfseries}{\thesubsection}{1em}{}
\begin{document}

\title{The multi elliptic-localized solutions and their asymptotic behaviors for the mKdV equation}

\author{Liming Ling}

\address{School of Mathematics, South China University of Technology, Guangzhou, China, 510641}
\email{linglm@scut.edu.cn}

\author{Xuan Sun}

\address{School of Mathematics, South China University of Technology, Guangzhou, China, 510641}

\email{masunx@mail.scut.edu.cn}

\begin{abstract}
We mainly construct and analyze the multi elliptic-localized solutions under the background of elliptic function solutions for the focusing modified Korteweg-de Vries (mKdV) equation. Based on the Darboux-B\"{a}cklund transformation, we provide a uniform expression for these solutions by the Jacobi theta functions. The asymptotic behaviors of multi elliptic-localized solutions are provided directly in two categories. By the consistent asymptotic expression of those solutions, we obtain that the collisions between the elliptic-breathers/solitons are elastic. Moreover, a sufficient condition of the strictly elastic collision between the solitons and breathers has been given by the symmetric analysis. In addition, as $k\rightarrow0^{+}$, the multi elliptic-localized solutions degenerate into solitons, breathers or soliton-breather solutions, which implies that we extend the solutions from the constant and vanishing backgrounds to the periodic solutions backgrounds. Moreover, we illustrate figures of the multi elliptic-localized solutions to visualize the above analysis.

{\bf Keywords:} modified Korteweg-de Vries equation, multi elliptic-localized solution, Darboux-B\"{a}cklund transformation, Jacobi theta function, asymptotic analysis.

\end{abstract}

\date{\today}

\maketitle
\section{Introduction}
In this work, we mainly study the multi elliptic-localized solutions of the focusing modified Korteweg-de Vries (mKdV) equation
\begin{equation}\label{eq:mKdV-equation}
	u_t+6u^2u_{x}+u_{xxx}=0,
\end{equation} 
under the elliptic functions backgrounds, where $u(x,t)$ is a real-valued function with $(x,t)\in \mathbb{R}^2$. According to the dynamic behavior of the multi elliptic-localized solutions, we call them the multi elliptic-soliton solutions, multi elliptic-breather solutions and multi elliptic-soliton-breather solutions. Furthermore, we study their symmetric properties, asymptotic behavior and degenerate form.

The mKdV equation is a well-known completely integrable model \cite{AblowitzC-91,AblowitzS-81}, which admits the Lax-pair \cite{LambG-80,Lax-68}, infinitely conserved
quantities \cite{MiuraGK-68}, bi-Hamiltonian structure \cite{MaF-96}, and so on. Over the past few decades, many kinds of solutions in the background of the vanishing solutions and plane waves, such as rational solutions \cite{ChowduryAN-16}, solitons \cite{ChanL-1994,YanL-20}, breather solutions \cite{ChvartatskyiF-17,KevrekidisKS-03,MunozP-19,TajiriW-98} and rogue waves \cite{ChvartatskyiF-17,KedzioraAN-14}, have been obtained extensively, through the inverse scattering transformation \cite{Pelinovsky-97}, the Darboux transformation \cite{Slyunyaev-01}, and so on.

Recently, the solutions on the background of elliptic functions have also been derived. Based on the squared wave function method, Shin gave the cnoidal waves \cite{Shin-03} and solitons on a cnoidal wave background \cite{Shin-04} of the coupled NLS equation. Combining with the Darboux transformation, Shin provided dark soliton on a cnoidal wave background \cite{Shin-05} of the  defocusing NLS equation. Hu, Lou and Chen \cite{HuLC-12explicit} constructed the explicit solutions on a cnoidal wave background of the KdV equation by utilizing the localization procedure of nonlocal symmetries. With the aid of the Darboux transformation method, Kedziora, Ankiewicz and Akhmediev \cite{KedzioraAN-14} constructed the rogue waves of the nonlinear Schr\"{o}dinger (NLS) equation under the cnoidal wave background. Based on the formal algebraic method from Cao et al. \cite{Cao-1990classical,Cao-1999-relation,GengC-01}, Chen and Pelinovsky developed an algebraic approach to gain the rogue wave solutions, the periodic traveling waves and the algebraic soliton solutions on the periodic waves background in many integrable nonlinear systems, such as the mKdV equation \cite{ChenP-18-mKdV,ChenPD-19-mKdV}, the NLS equation \cite{ChenP-18-NLS,ChenPW-19}, the sine-Gordon equation \cite{PelinovskyW-20}, and the derivative NLS equation \cite{ChenPD-21,ChenPU-21}. In the past two years, an increasing number of scholars have focused on solutions under the background of periodic functions. Then a large number of articles have appeared on equations with solutions expressed as elliptic functions, such as the Hirota equation \cite{GaoZ-2020-rogue-Hirota,PengTWZ-20-hirota-roguewave}, the sine-Gordon equation \cite{LiG-20}, and the derivative NLS equation \cite{ZhaiJG-21}.

In the past 40 years, a large number of scholars focused on utilizing the theta functions to represent the solutions under the background of elliptic functions. In the 1980s, Its, Rybin and Sail \cite{It-88} used the algebro-geometric approach to construct the finite-gap solutions of the NLS equation in terms of the theta functions. In 1994, Belokolos et al.\cite{BelokolosBEIM-94} investigated the finite-gap theta function formulas, which provide a theoretical basis for constructing the solutions in the theta function form. Fortunately, in recent years, expressing the solutions by the theta functions has progressed well. Shin \cite{Shin-12} gained the solutions in the theta function form to study the soliton dynamics moving in phase-modulated lattices. Utilizing the theta functions, Feng et al. \cite{FengLT-19} constructed the multi elliptic-breather solutions and rogue waves of the NLS equation successfully. Ling and Sun \cite{LinglmS-21} provided breather solutions in the theta function forms to exhibit the stable or unstable dynamic behaviors vividly after studying the spectral and orbital stability of the elliptic function solutions for the mKdV equation. However, to the best of our knowledge, the multi elliptic-solitons solutions and multi elliptic-breather solutions in theta function forms are still lacking. Therefore, our primary goal is to construct them in terms of the theta functions directly and analyze them from the aspects of symmetry properties, asymptotic behavior, and degenerate form.

In this paper, with the aid of the Darboux-B\"{a}cklund transformation, we provide a uniform expression to represent the multi elliptic-localized solutions under the elliptic function background. Unlike the solutions in the constant or vanishing background, the above solutions are expressed in terms of the theta functions with the uniform parameter $z$. Using the parameter $z$ instead of spectral parameter $\lambda$ could avoid the complex analysis on the genus-1 Riemann surface since the Abel map between $z$ and $\lambda$ establishes a conformal mapping between the genus-1 Riemann surface and the rectangular region, which have studied in our previous article \cite{FengLT-19,LinglmS-21}. Based on the multi elliptic-localized solutions, we study their properties, such as their asymptotic behavior, symmetric property, and degenerate form. The innovations of this paper mainly include the following aspects:

\begin{itemize}
	\item We construct the multi elliptic-localized solutions of the mKdV equation and provide a uniform expression of them in terms of theta functions either in $\cn$-type or $\dn$-type solutions background. Based on the asymptotic expressions of the multi elliptic-localized solutions, we obtain that the collisions between the breathers and solitons of all multi elliptic-localized solutions are elastic. Moreover, a sufficient condition of the strictly elastic collision between the solitons and breathers is given by the symmetric property.
	\item As $t\rightarrow \pm \infty$, the asymptotic expressions of the solutions are revealed by the exact expressions directly in two categories: along the line $L_i^{\pm}$ with velocity $v_i$ and in the area $R_{i}^{\pm}$ between two elliptic-breathers/solitons. The asymptotic expressions in region $R_{i}^{\pm}$ can be regarded as a shift in the background solution. As $k\rightarrow 0^+$, the multi elliptic-localized solutions could degenerate into the multi solitons/breathers/solitons-breathers under constant or vanishing backgrounds, which implies that we extend the multi solitons/breathers/solitons-breathers solutions under the constant and vanishing backgrounds into the backgrounds of the elliptic functions.
	\item The relations between velocity $v$ and the spectral parameter $\lambda$ for the multi elliptic-localized solutions are obtained, which has never been reported before to the best of our knowledge. Analyzing the zeros and poles of meromorphic functions $\Re(I(z))$ and $\Re(\Omega(z))$, as well as studying the conformal map between the spectral parameter $\lambda$ and the uniform parameter $z$ are two indispensable steps in proving the relations between the velocity $v$ and the spectral parameter $\lambda$. 
\end{itemize}

\subsection{Main results}
As we all know, the following two elliptic functions are the periodic solutions of the mKdV equation \eqref{eq:mKdV-equation}:
\begin{equation}\label{eq:mKdV-solution-cn-dn}
	u(x,t)=k\alpha \cn(\alpha(x-st),k),
	\qquad \qquad \qquad  
	u(x,t)=\alpha \dn(\alpha(x-st),k),
\end{equation} 
where $\cn(\cdot,k)$, $\dn(\cdot,k)$ denote the Jacobi elliptic functions with elliptic modulus $k$, and $s$ is the velocity. If the solution of mKdV equation is $\cn$-type, the value of $s$ is $\alpha^2(2k^2-1)$; if it is $\dn$-type solutions, $s=\alpha^2(2-k^2)$. The exact calculation process of solutions \eqref{eq:mKdV-solution-cn-dn} was given in \cite{LinglmS-21}.

The mKdV equation \eqref{eq:mKdV-equation} admits the Lax pair:
\begin{equation}\label{eq:Lax-pair}
	\Phi_{x}(x,t;\lambda)
	=\mathbf{U}(\lambda;u)\Phi(x,t;\lambda), \qquad
	\Phi_{t}(x,t;\lambda)
	=\mathbf{V}(\lambda;u)\Phi(x,t;\lambda),
\end{equation}
where $\lambda\in \mathbb{C}\cup \{\infty\}$ is called the spectral parameter and matrices $\mathbf{U}(\lambda;u)$ and $\mathbf{V}(\lambda;u)$ in equation \eqref{eq:Lax-pair} are defined as
\begin{equation}\label{eq:Lax-U-V-sigma3}
	\begin{split}
		\mathbf{U}(\lambda;u)
		:=&-\ii \lambda \sigma_3 +\mathbf{Q}, \qquad \sigma_3:= \mathrm{diag}(1,-1), \qquad 
		\mathbf{Q}:= \begin{bmatrix}
			0 & u \\ -u & 0
		\end{bmatrix},\\
		\mathbf{V}(\lambda;u)
		:=&	-4\ii \lambda^3 \sigma_3 +4\lambda^2\mathbf{Q}+\ii \lambda\sigma_3\left(2\mathbf{Q}_{x}-2\mathbf{Q}^2\right)+2 \mathbf{Q}^3-\mathbf{Q}_{xx},
	\end{split}
\end{equation}
satisfying the symmetric properties
\begin{equation}\label{eq:UV-sym}
	\begin{split}
		\mathbf{U}^{\dagger}(\lambda^*;u)
		=&-\mathbf{U}(\lambda;u), \qquad 
		 \mathbf{U}^{\top}(-\lambda;u)
		=-\mathbf{U}(\lambda;u),
		\\
		\mathbf{V}^{\dagger}(\lambda^*;u)
		=&-\mathbf{V}(\lambda;u), \qquad 
		\mathbf{V}^{\top}(-\lambda;u)
		=-\mathbf{V}(\lambda;u).
	\end{split}
\end{equation} 

For the elliptic solutions \eqref{eq:mKdV-solution-cn-dn} of mKdV equation \eqref{eq:mKdV-equation}, the corresponding fundamental solution $\Phi(x,t;\lambda)$ of the Lax pair \eqref{eq:Lax-pair} could be written as
\begin{equation}\label{eq:Lax-solution-Phi}
	\Phi(x,t;\lambda)
	=\frac{\alpha \vartheta_2\vartheta_4}{\vartheta_3\vartheta_4(\frac{\alpha \xi}{2K})}
	\begin{bmatrix}
		\frac{\vartheta_1(\frac{\ii(z-l)-\alpha\xi}{2K})}{\vartheta_4(\frac{\ii(z-l)}{2K})}E_1(x,t;z)
		& \frac{\vartheta_3(\frac{\ii(z+l)+\alpha\xi}{2K})}{\vartheta_2(\frac{\ii(z+l)}{2K})}E_2(x,t;z) \\ 
		-\frac{\vartheta_3(\frac{\ii(z+l)-\alpha\xi}{2K})}{\vartheta_2(\frac{\ii(z+l)}{2K})}\ee^{\alpha \xi Z(2\ii l+K)}E_1(x,t;z)
		& -\frac{\vartheta_1(\frac{\ii(z-l)+\alpha\xi}{2K})}{\vartheta_4(\frac{\ii(z-l)}{2K})}\ee^{\alpha \xi Z(2\ii l+K)}E_2(x,t;z)
	\end{bmatrix},
\end{equation} 
where $\xi=x-st$, $l=0 \text{ or }\frac{K'}{2}$, 
\begin{equation}\label{eq:E1-E2}
	\begin{split}
		E_1(z)\equiv E_1(x,t;z)&=\exp \left((\alpha Z(\ii (z-l))+\ii \lambda)(x-st)   +4\ii  y \lambda t \right),\\
		E_2(z)\equiv E_2(x,t;z)&=\exp \left(\left(-\ii\frac{\alpha\pi}{2K}-\alpha Z(\ii (z+l)+K+\ii K')+\ii \lambda\right)(x-st)   -4\ii  y \lambda t \right),
	\end{split}
\end{equation}
$y(z)=\frac{\alpha^2k^2}{4}\left(\mathrm{sn}^2(\ii(z-l))-\mathrm{sn}^2(\ii(z+l)+K+\ii K')\right)$ and 
\begin{equation}\label{eq:Omega}
		\Omega(z)=8\ii\lambda y=-\alpha^3 \left( k^2 \dn(\ii(z-l))\sn(\ii(z-l))\cn(\ii(z-l))+k'^2\frac{\dn(\ii(z+l))\sn(\ii(z+l))}{\cn^3(\ii(z+l))}\right),
\end{equation}
with the parameter $\lambda$ defined in \eqref{eq:lambda-elliptic}, the Jacobi theta functions $\vartheta_i(z), i=1,2,3,4$ defined in \eqref{eq:theta1234}, the Jacobi Zeta functions $Z(z)$ defined in \eqref{eq:zeta} and the complete elliptic integrals $K=K(k), K'=K(k')$ defined in \eqref{eq:J-K-E-int}. The selection of parameter $l$ is closely related to the elliptic function solutions \eqref{eq:mKdV-solution-cn-dn}. If the solution is $\cn$-type, the corresponding fundamental solution $\Phi(x,t;\lambda)$ of Lax pair \eqref{eq:Lax-pair} is obtained by the parameter $l=0$. If the solution is $\dn$-type, the fundamental solution is obtained by the parameter $l=\frac{K'}{2}$. The detailed analysis and calculation process of solution $\Phi(x,t;\lambda)$ is given in \cite{LinglmS-21}. Furthermore, we obtain the expressions of the spectral parameter $\lambda$ by the uniform parameter $z$ with different values of $l$ in the following:
\begin{subequations}\label{eq:lambda-elliptic}
	\begin{align}
		\lambda(z)=&\frac{\ii \alpha}{2}\frac{\dn(\ii(z-l))\sn(\ii(z-l))}{\cn(\ii(z-l))},\qquad l=0, \label{eq:lambda-elliptic-0}\\
		\lambda(z)=&\frac{\ii \alpha k^2}{2}\frac{\sn(\ii(z-l))\cn(\ii(z-l))}{\dn(\ii(z-l))},\qquad l=\frac{K'}{2} \label{eq:lambda-elliptic-K}.
	\end{align}
\end{subequations}
The conformal map $\lambda(z)$ maps the rectangular area $S$ onto the whole complex plane with two cuts, where the region $S$ is defined as
\begin{equation}\label{eq:set-S}
	S:=\left\{ z\in \mathbb{C}\left| -K'+l\le \Re(z)\le K'+l, -\frac{K}{2}\le\Im(z)\le \frac{K}{2},\quad l=0 \,\, \text{  or  } \,\, \frac{K'}{2}\right.\right\}.
\end{equation} 
 Therefore, the studies of the spectral parameter $\lambda$ over the whole complex plane with two cuts (the Riemann surface with genus-1) are converted to analyze the uniform parameter $z$ in the rectangular region $S$, which is proved in \cite{LinglmS-21}.

Based on the Darboux-B\"{a}cklund transformation, the Darboux matrix $\mathbf{T}^{[N]}(x,t;\lambda)$, Theorem \ref{theorem:T}, and Theorem \ref{theorem:u-N-breather}, the exact expression of the solution $u^{[N]}(x,t)$ in \eqref{eq:u-N-breather} could be rewritten in terms of theta functions.

\begin{theorem}\label{theorem:u-N}
	The solution $u^{[N]}(x,t)$ of equation \eqref{eq:mKdV-equation} could be written as 
	\begin{equation}\label{eq:mKdV-solution-xi-n-1}
		u^{[N]}(x,t)=\frac{\alpha \vartheta_2\vartheta_4}{\vartheta_3\vartheta_3(\frac{2\ii l}{2K})}\left(\frac{\vartheta_4(\frac{\alpha\xi}{2K})}{\vartheta_2(\frac{\alpha\xi+2\ii l}{2K})}\right)^{m-1}\frac{\det(\mathcal{M})}{\det(\mathcal{D})}\ee^{-\alpha\xi Z(2\ii l+K)},
	\end{equation}
	where $\xi=x-st$, the matrices $\mathcal{M}$ and $\mathcal{D}$ are $m\times m$, whose elements are given by 
	\begin{equation}
		\begin{split}
			(\mathcal{M})_{i,j}=&
			\mathbf{E}_i^{\dagger}
			\mathbf{r}_i^{*}
			\begin{bmatrix}
				-\frac{\vartheta_2(\frac{\ii(z_i^*-z_j+2l)+\alpha\xi}{2K})}{\vartheta_1(\frac{\ii(z_i^*-z_j)}{2K})} & 
				-\frac{\vartheta_4(\frac{\ii(z_i^*+z_j+2l)+\alpha\xi}{2K})}{\vartheta_3(\frac{\ii(z_i^*+z_j)}{2K})} \\
				\frac{\vartheta_4(\frac{\ii(-z_i^*-z_j+2l)+\alpha\xi}{2K})}{\vartheta_3(\frac{\ii(-z_i^*-z_j)}{2K})} & 
				\frac{\vartheta_2(\frac{\ii(z_j-z_i^*+2l)+\alpha\xi}{2K})}{\vartheta_1(\frac{\ii(z_j-z_i^*)}{2K})}
			\end{bmatrix}\mathbf{r}_j^{-1}\mathbf{E}_j,\\
			(\mathcal{D})_{i,j}=&\mathbf{E}_i^{\dagger}
			\begin{bmatrix}
				-\frac{\vartheta_4(\frac{\ii(z_i^*-z_j)+\alpha\xi}{2K})}{\vartheta_1(\frac{\ii(z_i^*-z_j)}{2K})} & 
				\frac{\vartheta_2(\frac{\ii(z_i^*+z_j)+\alpha\xi}{2K})}{\vartheta_3(\frac{\ii(z_i^*+z_j)}{2K})} \\
				\frac{\vartheta_2(\frac{-\ii(z_i^*+z_j)+\alpha\xi}{2K})}{\vartheta_3(\frac{-\ii(z_i^*+z_j)}{2K})} & 
				\frac{\vartheta_4(\frac{\ii(z_j-z_i^*)+\alpha\xi}{2K})}{\vartheta_1(\frac{\ii(z_j-z_i^*)}{2K})}
			\end{bmatrix}\mathbf{E}_j,	 
		\end{split}
	\end{equation}
	\begin{equation}\label{eq:r_i}
	\mathbf{r}_i=\mathrm{diag}\left( r_i^{-1}, r_i \right),
	\qquad r_i\equiv r(z_i):=\frac{\vartheta_2(\frac{\ii(z_i+l)}{2K})}{\vartheta_4(\frac{\ii(z_i-l)}{2K})},\qquad
	\mathbf{E}_i=\begin{bmatrix}
		E_1(z_i) &  c_iE_2(z_i)
	\end{bmatrix}^{\top},
\end{equation}	
$E_1(z)$ and $E_2(z)$ are defined in equation \eqref{eq:E1-E2}, and $m=n_1+2n_2$, the parameters $n_1$ and $n_2$ will be determined in equation \eqref{eq:multi-T}.

\end{theorem}
  By the different dynamic behaviors of those elliptic-localized solutions and corresponding relationships between numbers $N$ and $m$, solutions $u^{[N]}(x,t)$ are named as follows.
	\begin{itemize}
		\item  If $n_2=0$ ($m=N$), which implies the spectral parameters satisfy $\lambda_i\in \ii \mathbb{R}$, $i=1,2,\cdots,N$ and the multi-fold Darboux matrix is $\mathbf{T}^{[N]}(x,t;\lambda)=\mathbf{T}^{\mathrm{P}}_N(x,t;\lambda)\mathbf{T}^{\mathrm{P}}_{N-1}(x,t;\lambda)\cdots \mathbf{T}^{\mathrm{P}}_1(x,t;\lambda)$ defined in \eqref{eq:multi-T}, the solution $u^{[N]}(x,t)$ is called the {\bf multi elliptic-soliton solution}.
		
		\item If $n_1=0$ ($m=2N$), which implies the spectral parameters satisfy $\lambda_i\in\mathbb{C}\backslash(\mathbb{R}\cup \ii \mathbb{R})$, $i=1,2,\cdots,N$, and the multi-fold Darboux matrix is $\mathbf{T}^{[N]}(x,t;\lambda)=\mathbf{T}^{\mathrm{C}}_N(x,t;\lambda)\mathbf{T}^{\mathrm{C}}_{N-1}(x,t;\lambda)\cdots \mathbf{T}^{\mathrm{C}}_1(x,t;\lambda)$ in \eqref{eq:multi-T}, the solution $u^{[N]}(x,t)$ is called the {\bf multi elliptic-breather solution}. 
		
		\item  If $n_1,n_2\neq 0$ ($N<m<2N$), which implies the spectral parameters contain $\lambda_i\in \ii \mathbb{R}$ and $\lambda_i\in \mathbb{C}\backslash(\mathbb{R}\cup\ii \mathbb{R})$, $i=1,2,\cdots,N$, and the multi-fold Darboux matrix is $\mathbf{T}^{[N]}(x,t;\lambda)=\mathbf{T}^{\mathrm{J}}_N(x,t;\lambda)\mathbf{T}^{\mathrm{J}}_{N-1}(x,t;\lambda)\cdots \mathbf{T}^{\mathrm{J}}_1(x,t;\lambda)$,  $J=\mathrm{P}, \mathrm{C}$, in \eqref{eq:multi-T}, the solution $u^{[N]}(x,t)$ is called the {\bf multi elliptic-soliton-breather solution}.
	\end{itemize}

Based on the above exact solutions, we conduct a series of analyses on them and then get the following conclusions. For the convenience of study, we will make a rotation $\xi=x-st$, $\tau=t$ on the solution $u^{[N]}(x,t)$ and define it as
\begin{equation}\label{eq:xi-x-solution}
	\hat{u}^{[N]}(\xi,\tau) \,
	\xlongequal{\xi=x-st,\tau=t} \,
	u^{[N]}(x,t).
\end{equation}
Then, we mainly consider asymptotic expressions of function $\hat{u}^{[N]}(\xi,\tau)$ under the $\xi$ and $\tau$ axis. Define $L_i^{\pm}$, $i=1,2,\cdots,N$, as the evolution direction of breathers or solitons for the solution $\hat{u}^{[N]}(\xi,\tau)$ in \eqref{eq:xi-x-solution}. Symbols `$\pm$' indicate positive/negative directions with respect to time. The $R_i^{\pm}$ denote the region between the lines $L_{i-1}^{\pm}$ and $L_{i}^{\pm}$, $i=1,2,\cdots,N$ respectively. Figure \ref{fig:li} provides a sketch map of $L_i^{\pm}$ and $R_i^{\pm}$, $i=1,2,\cdots,N$. As $\tau\rightarrow \pm \infty$, we get the following asymptotic expressions. Along line $L_q^{+}$, we divided lines $L_i^{+}$, $i=1,2,\cdots,N$ into two categories by line $L_q^{+}$. On the left side of line $L_i^{+}$, $i=1,2,\cdots,q-1$, we define the number of related spectral parameters $\lambda_i\in \ii \mathbb{R}$ and $\lambda_i \in \mathbb{C}\backslash(\ii \mathbb{R}\cup \mathbb{R})$, $i=1,2,\cdots,q$ as $q_1$ and $q_2$ respectively. Then, we get
	\begin{equation}\label{eq:defie-q}
		q\equiv q(q_1,q_2):=q_1+q_2+1, \qquad h\equiv h(q_1,q_2):=q_1+2q_2.
	\end{equation}

\begin{theorem}\label{theorem:exact-N-solution}
	The asymptotic expression of solution $\hat{u}^{[N]}(\xi,\tau)$ along the line $L_q^{\pm}$ as $\tau\rightarrow \pm \infty$ could be rewritten in the following two different forms:
	\begin{itemize}
		\item[(i)] Along line $L_q^{\pm}$ with $\lambda_q\in \ii \mathbb{R}$, i.e., there exists only one parameter $z_h$ satisfies $\eta_h=\text{const}$ defined in \eqref{eq:E-2-hat}, as $\tau\rightarrow \pm \infty$, the asymptotic expression of solution $\hat{u}^{[N]}(\xi,\tau)$ is 
	\begin{equation}\label{eq:mKdV-solution-xi-n-k-vec}
		\begin{split}
			\hat{u}^{[N]}(\xi,\tau;L_q^{\pm})= &\frac{\alpha \vartheta_2\vartheta_4}{\vartheta_3\vartheta_3(\frac{2\ii l}{2K})}\left(\frac{\vartheta_4(\frac{\alpha\xi}{2K})}{\vartheta_2(\frac{\alpha\xi+2\ii l}{2K})}\right)^{m-1} \frac{\det\left(\sum_{i,j=1}^2(\ii)^{i+j}\mathbf{X}_i^{[h,\pm]\dagger}\mathbf{Y}^{[2]}_i \mathcal{M}^{[i,j]}\mathbf{Y}^{[1]}_j\mathbf{X}_j^{[h,\pm]} \right)}{\det\left(\sum_{i,j=1}^2(-1)^{j}\mathbf{X}_i^{[h,\pm]\dagger}\mathcal{D}^{[i,j]}\mathbf{X}_j^{[h,\pm]}\right)\ee^{\alpha\xi Z(2\ii l+K)}} \\
			&+\mathcal{O}\left(\exp\left(-\min_{j\neq h}
			\Re(I_j)|v_h-v_j||t|\right)\right), \qquad \qquad \tau \rightarrow \pm \infty,
		\end{split}	
	\end{equation}
	 where the relation between $q$ and $h$ is defined in equation \eqref{eq:defie-q}, matrices $\mathbf{Y}^{[1]}_i,\mathbf{Y}^{[2]}_i$ are defined in \eqref{eq:Y}, $\mathcal{M}^{[i,j]},\mathcal{D}^{[i,j]}$ are defined in \eqref{eq:M-H-zi-zj} and
	\begin{equation}\label{eq:X1-X2-hat}
		\begin{split}
			&\mathbf{X}_1^{[h,+]}:=\mathrm{diag} \left(
			0, \cdots , 0, 1,\overbrace{1,\cdots ,1}^{m-h}
			\right), \qquad 
			\mathbf{X}_2^{[h,+]}:=\mathrm{diag} \left( \overbrace{1,\cdots ,1}^{h-1},
			\ee^{\eta_h},0, \cdots , 0
			\right),\\
			&\mathbf{X}_1^{[h,-]}:=\mathrm{diag} \left(
			\overbrace{1,\cdots ,1}^{h-1},
			\ee^{-\eta_h},0, \cdots , 0
			\right), \quad 
			\mathbf{X}_2^{[h,-]}:=\mathrm{diag} \left( 0, \cdots , 0,1, \overbrace{1,\cdots ,1}^{m-h}
			\right).
		\end{split}
	\end{equation}
\item[(ii)] Along the line $L_q^{\pm}$ with $\lambda_q \in \mathbb{C}\backslash(\ii \mathbb{R}\cup \mathbb{R})$, i.e., there exist two parameters $z_h,z_{h+1}$ satisfy $\eta_h=\eta_{h+1}=\text{const}$, as $\tau\rightarrow \pm \infty$, 
the asymptotic expression of solution $\hat{u}^{[N]}(\xi,\tau)$ also could be written as \eqref{eq:mKdV-solution-xi-n-k-vec}, with the matrices
\begin{equation}\label{eq:X1-X2-hat-2}
	\begin{split}
		&\mathbf{X}_1^{[h,+]}:=\mathrm{diag} \left(
		0, \cdots , 0, 1, 1, \overbrace{1,\cdots ,1}^{m-h-1}
		\right), \qquad
		\mathbf{X}_2^{[h,+]}:=\mathrm{diag} \left( \overbrace{1,\cdots ,1}^{h-1},
		\ee^{\eta_h},\ee^{\eta_{h+1}},0, \cdots , 0
		\right),\\
		&\mathbf{X}_1^{[h,-]}:=\mathrm{diag} \left(
		\overbrace{1,\cdots ,1}^{h-1},
		\ee^{-\eta_h},\ee^{-\eta_{h+1}},0, \cdots , 0
		\right), \quad 
		\mathbf{X}_2^{[h,-]}:=\mathrm{diag} \left( 0, \cdots , 0, 1, 1, \overbrace{1,\cdots ,1}^{m-h-1}
		\right).
	\end{split}
	\end{equation}
The relation between $q$ and $h$ is defined in equation \eqref{eq:defie-q}.
	\end{itemize}
\end{theorem}
Utilizing the theta ratio determinant in \cite{Takahashi-16}, we could simplify the first case of the solution $\hat{u}^{[N]}(\xi,\tau;L_q^{\pm})$ \eqref{eq:mKdV-solution-xi-n-k-vec} in Remark \ref{remark:mKdV-solution-asy-v-simp-i}. In addition, the asymptotic expressions on the region $R_{k}^{\pm},k=1,2,\cdots,N$ are obtained. 
\begin{theorem}\label{theorem:asy-R-elliptic}
	The asymptotic expressions on the region $R_{k}^{\pm},k=1,2,\cdots,N,$ could be divided in the following two types: 
	\begin{itemize}
		\item[(i)] Along line $L_{q-1}$ with $\lambda_{q-1}\in \ii \mathbb{R}$, 
		as $\tau\rightarrow \pm \infty$, 
		the asymptotic expression on the region $R_{q}^{\pm}$ could be written as
		\begin{equation}\label{eq:solution-u-n-infy-R-cn}
			\begin{split}
					\hat{u}^{[N]}(\xi,\tau;R_{q}^{\pm}) \rightarrow&(-1)^{p}\alpha k \cn(\alpha \xi +s_{h,h}^{\pm}), \qquad \tau\rightarrow \pm \infty, \qquad l=0,
			\end{split}
		\end{equation}
	or 
		\begin{equation}\label{eq:solution-u-n-infy-R-dn}
		\hat{u}^{[N]}(\xi,\tau;R_{q}^{\pm}) \rightarrow(-1)^{p}\alpha \dn(\alpha \xi +s_{h,h}^{\pm}),  
	\qquad \tau\rightarrow \pm \infty, \qquad l=\frac{K'}{2},
	\end{equation}
 in which the relations between $q$ and $h$ are defined in equation \eqref{eq:defie-q};  $p$ is the number of spectral parameter $\lambda_i\in \ii \mathbb{R}$, $i=1,2,\cdots,N$, satisfying $|\Im(\lambda_i)|>\frac{\alpha(1+k')}{2}$ when $l=\frac{K'}{2}$ or $|\Im(\lambda_i)|>\frac{\alpha}{2}$ when $l=0$; and
	\begin{equation}\label{eq:s-i-j}
		\begin{split}
			s_{i,j}^{\pm}=&\pm\left(\sum_{k=1}^{i}2 \Im(z_k)-\sum_{k=j+1}^{m}2\Im(z_k)\right).
		\end{split}
	\end{equation}
\item[(ii)] Along line $L_{q-1}$ with $\lambda_{q-1}\in \mathbb{C}\backslash(\ii \mathbb{R}\cup \mathbb{R})$, as $\tau \rightarrow \pm \infty$, the asymptotic expression on the region $R_{q}^{\pm}$ also could be written in \eqref{eq:solution-u-n-infy-R-cn} or \eqref{eq:solution-u-n-infy-R-dn} by replacing the parameter $s_{h,h}^{\pm}$ with $s_{h+1,h+1}^{\pm}$.
	\end{itemize}
\end{theorem}

\begin{theorem}\label{theorem:symm}
	When $c_i=1$, $i=1,2,3,\cdots,N$, the solutions $u^{[N]}( x,t)$ in equation \eqref{eq:mKdV-solution-xi-n-1} have the symmetric relation: $u^{[N]}( x,t)=u^{[N]}(- x,-t)$.
\end{theorem}

\begin{theorem}\label{theorem:Phi-k}
	 If $l=\frac{K'}{2}$ and $\lambda_i\in \ii \mathbb{R}$, $i=1,2,\cdots,N$, satisfying $|\Im(\lambda_i)|<\frac{\alpha(1-k')}{2}$, the multi elliptic-soliton solutions $u^{[N]}(x,t)$ degenerate into the constant solution. Excepting the above case, for $l=0$ and $l=\frac{K'}{2}$, as $k\rightarrow 0^+$, the multi elliptic-soliton/breather/soliton-breather solutions would degenerate into the multi soliton/breather/soliton-breather solutions under the constant or vanishing backgrounds, respectively.
\end{theorem}

\subsection{Outline for this work}
The organization of this work is as follows. In section \ref{sec:Darboux-transformation}, we obtain a uniform expression of the multi elliptic-localized solutions in theta functions and list examples to show different dynamic behaviors. In section \ref{sec:asymptotially analysis}, to begin, we prove a symmetric property of the multi elliptic-localized solutions $u^{[N]}(x,t)$ with appropriate restrictions, which reflects the strictly elastic collisions between the breathers and solitons of solution $u^{[N]}(x,t)$. Moreover, we take an asymptotic analysis of solutions along the line $L_k^{\pm}$ and on the region $R_{k}^{\pm}$ as $t\rightarrow\pm\infty$. Furthermore, in section \ref{sec:elliptic-constant}, the asymptotic expression of solutions as $k\rightarrow 0^+$ is obtained, which shows the relationship between multi elliptic-soliton/breather solutions and the solitons/breathers on the constant background. The conclusions and discussions are involved in section \ref{sec:conclusion-discussion}.

\section{The elliptic-localized solutions of mKdV equation}\label{sec:Darboux-transformation}

In this section, our primary goal is to construct the exact expression of solutions for the mKdV equation. Based on the Darboux-B\"{a}cklund transformation, we construct the solutions of the mKdV equation \eqref{eq:mKdV-equation} and express them in theta functions. Furthermore, we vividly exhibit figures, including different dynamic behaviors of elliptic-localized solutions.

\subsection{The Darboux-B\"{a}cklund transformation}
So far, the Darboux-B\"{a}cklund transformation \cite{Cieslinski-09,Guo-2012-Darboux,MatveevS-1991-darboux} has been very mature, so we would not repeat it too much. Here, we just provide the conclusions, which are useful in the following analysis. From the Darboux transformation, we know that the Darboux matrix $\mathbf{T}^{[1]}(\lambda;x,t)$ could admit a new equation
\begin{equation}\label{eq:Phi-1}
	\Phi_{ x}^{[1]}(x,t;\lambda)=\mathbf{U}^{[1]}(\lambda;u^{[1]})\Phi^{[1]}(x,t;\lambda), \,\, \Phi^{[1]}(x,t;\lambda):=\mathbf{T}^{[1]}(\lambda;x,t)\Phi(x,t;\lambda), \, \mathbf{U}^{[1]}(\lambda;u^{[1]})\equiv\mathbf{U}(\lambda;u^{[1]}).
\end{equation} 
Based on the symmetric properties of matrices $\mathbf{U}(\lambda;u)$ and $\mathbf{V}(\lambda;u)$ in equation \eqref{eq:UV-sym}, we obtain
	\begin{equation}\label{eq:Phi-inver-top}
		\Phi(x,t;\lambda)\Phi^{\top}(x,t;-\lambda)=\mathbb{I} \qquad \text{and}\qquad  \Phi(x,t;\lambda)\Phi^{\dagger}(x,t;\lambda^*)=\mathbb{I}.
	\end{equation}
	Combining equations \eqref{eq:UV-sym}, \eqref{eq:Phi-1} and \eqref{eq:Phi-inver-top}, we get that the matrix $\mathbf{T}^{[1]}(\lambda;x,t)$ satisfies the following properties.
\begin{prop}\label{prop:symmetric-T}
	The Darboux matrix $\mathbf{T}^{[1]}(\lambda;x,t)$ satisfies
	 \begin{equation}\label{eq:T-sym}
			(\mathbf{T}^{[1]}(\lambda;x,t))^{-1}=(\mathbf{T}^{[1]}(\lambda^*;x,t))^{\dagger}, \qquad (\mathbf{T}^{[1]}(\lambda;x,t))^{-1}=(\mathbf{T}^{[1]}(-\lambda;x,t))^{\top}.
	\end{equation} 
\end{prop}
Then, the Darboux matrix is divided into the following two cases to satisfy the symmetries provided in Proposition \ref{prop:symmetric-T}.

\begin{theorem}\label{theorem:T}
	To impose the symmetry \eqref{eq:T-sym}, the Darboux matrix $\mathbf{T}^{[1]}(\lambda; x,t)$ can be divided into the following two types:
	\begin{itemize}
		\item If $\lambda_1=\lambda(z_1)\in \ii \mathbb{R}$ and $\Phi_1\Phi_1^{\dagger}=(\Phi_1\Phi_1^{\dagger})^{\top}$, the Darboux matrix could be written as
		\begin{equation}\label{eq:T-1}
			\mathbf{T}^{\mathrm{P}}_1(\lambda; x,t)=\mathbb{I}-\frac{\lambda_1-\lambda_1^*}{\lambda-\lambda_1^*}\frac{\Phi_1\Phi_1^{\dagger}}{\Phi_1^{\dagger}\Phi_1}, \qquad \quad\Phi_1:=\Phi( x,t;\lambda_1)\mathbf{c}_1=\Phi( x,t;\lambda_1)\left[ 1,c_1\right]^{\top}, \quad c_1\in \mathbb{C}\setminus \left\{0\right\};
		\end{equation}
	 \item If $\lambda_1=\lambda(z_1)\in \mathbb{C}\setminus (\mathbb{R}\cup\ii \mathbb{R})$, the Darboux matrix could be written as
	 	\begin{equation}\label{eq:T-2}
	 	\mathbf{T}^{\mathrm{C}}_1(\lambda; x,t)
	 	=\mathbb{I}-
	 	\begin{bmatrix}
	 		\Phi_1 & \Phi_1^*
	 	\end{bmatrix}\mathbf{M}_2^{-1} \mathbf{D}_2^{-1}
	 	\begin{bmatrix}
	 		\Phi_1^{\dagger} \\[5pt] \Phi_1^{\top}
	 	\end{bmatrix}, 	\,\, 
 		\mathbf{D}_2={\rm diag}\left(\lambda-\lambda_1^*,\lambda+\lambda_1\right),\,\,	\mathbf{M}_2=\begin{bmatrix}
	 	\frac{\Phi_1^{\dagger} \Phi_1}{\lambda_1-\lambda_1^*} &
	 	\frac{\Phi_1^{\dagger} \Phi_1^*}{-\lambda_1^*-\lambda_1^*} \\
	 	\frac{\Phi_1^{\top}\Phi_1}{\lambda_1+\lambda_1} &
	 	\frac{\Phi_1^{\top}\Phi_1^*}{-\lambda_1^*+\lambda_1}
 	\end{bmatrix}, 
	 \end{equation}
 with $\Phi_1$ defined in \eqref{eq:T-1}.
	\end{itemize}
\end{theorem}

The proof of Theorem \ref{theorem:T} is given in \cite{LinglmS-21}, so we would not repeat it here. 

\begin{remark}\label{remark:T}
 The equation \eqref{eq:Phi-inver-top} implies $\Phi^{\top}(x,t;\lambda)=\Phi^{\dagger}(x,t;-\lambda^*)=\Phi^{\dagger}(x,t;\lambda)$, $\lambda\in \ii \mathbb{R}$. Combining the definition of $\Phi_1$ in equation \eqref{eq:T-1}, we obtain 
\begin{equation}
	\Phi_1\Phi_1^{\dagger}=\Phi(x,t;\lambda_1)\begin{bmatrix}1\\ c_1\end{bmatrix}\left[ 1,c_1^*\right]\Phi^{\dagger}(x,t;\lambda_1)=\Phi(x,t;\lambda_1)\begin{bmatrix}
		1 & c_1^* \\ c_1 & |c_1|^2
	\end{bmatrix}
\Phi^{\top}(x,t;\lambda_1).
\end{equation}
When $c_1\in \mathbb{R}$, $\lambda_1\equiv\lambda(z_1)\in \ii \mathbb{R}$, and $l=0$, $\frac{K'}{2}$, the equation $\Phi_1\Phi_1^{\dagger}=(\Phi_1\Phi_1^{\dagger})^{\top}$ holds. In addition, from \cite{LinglmS-21} we know that if $l=\frac{K'}{2}$, $\lambda_1\equiv\lambda(z_1)\in  \ii \mathbb{R}$ with $z_1$ satisfies $\Im(z_1)=\pm \frac{K'}{2}$ and $|c_1|=1$, the symmetry $\Phi_1\Phi_1^{\dagger}=(\Phi_1\Phi_1^{\dagger})^{\top}$ holds.  In addition, by function $\Phi(x,t;\lambda)$ in \eqref{eq:Lax-solution-Phi}, we obtain $\Phi^*(x,t;\lambda)=\Phi(x,t;\lambda)|_{z=-z^*+2l}$.
\end{remark}

Reviewing the Darboux matrix $\mathbf{T}^{[1]}(\lambda;x,t)$ defined in equation \eqref{eq:Phi-1} and combining the B\"{a}cklund transformation, we can obtain a new solution after an iteration. When we iterate over and over again, we can get a large number of solutions. Naturally, we will consider that how to get those solutions under the multi-fold iteration directly. Based on the elementary Darboux transformation $\mathbf{T}_1^{\mathrm{P}}(\lambda;x,t)$ and $\mathbf{T}_1^{\mathrm{C}}(\lambda;x,t)$ provided in Theorem \ref{theorem:T}, we can iterate them to obtain the multi-order ones, called the multi-fold Darboux matrix:  
	\begin{equation}\label{eq:multi-T}
		\mathbf{T}^{[N]}(\lambda;x,t)=\mathbf{T}^{\mathrm{J}}_{N}(\lambda; x,t)\mathbf{T}^{\mathrm{J}}_{N-1}(\lambda; x,t)\cdots\mathbf{T}^{\mathrm{J}}_{2}(\lambda; x,t)\mathbf{T}^{\mathrm{J}}_{1}(\lambda; x,t),\qquad \mathrm{J}=\mathrm{P},\mathrm{C},
	\end{equation}
	where $\mathbf{T}_1^{\mathrm{P}}(\lambda;x,t)$ and $\mathbf{T}_1^{\mathrm{C}}(\lambda;x,t)$ are defined in \eqref{eq:T-1} and \eqref{eq:T-2}, and matrices $\mathbf{T}_i^{\mathrm{J}}(\lambda;x,t)$, $\mathrm{J}=\mathrm{P},\mathrm{C}$, $i=2,3,\cdots,N,$ are defined as
	\begin{equation}\label{eq:T-i-P-C}
		\begin{split}
			\mathbf{T}^{\mathrm{P}}_i(\lambda; x,t)
			=&\mathbb{I}-\frac{\lambda_i-\lambda_i^*}{\lambda-\lambda_i^*}\frac{ \Phi_i^{[i-1]}(\Phi_i^{[i-1]})^{\dagger} }{(\Phi_i^{[i-1]})^{\dagger}\Phi_i^{[i-1]}},
			\Phi_i^{[i-1]}=\mathbf{T}_{i-1}^{\mathrm{J}}(\lambda_i; x,t)\mathbf{T}_{i-2}^{\mathrm{J}}(\lambda_i; x,t)\cdots\mathbf{T}_{1}^{\mathrm{J}}(\lambda_i; x,t)\Phi_i, \mathrm{J}=\mathrm{P}, \mathrm{C},\\
		\mathbf{T}^{\mathrm{C}}_i(\lambda; x,t)
		=&\mathbb{I}-
		\begin{bmatrix}
			\Phi_i^{[i-1]} & (\Phi_i^{[i-1]})^*
		\end{bmatrix}(\mathbf{M}_2^{[i-1]})^{-1} (\mathbf{D}_2^{i})^{-1}
		\begin{bmatrix}
			(\Phi_i^{[i-1]})^{\dagger} \\[5pt] (\Phi_i^{[i-1]})^{\top}
		\end{bmatrix}, \quad 	\mathbf{D}^{i}_2=\mathrm{diag}\left( \lambda-\lambda_i^*,\lambda+\lambda_i\right), \\
		\mathbf{M}^{[i-1]}_2=&\begin{bmatrix}
			\frac{(\Phi_i^{[i-1]})^{\dagger} \Phi_i^{[i-1]}}{\lambda_i-\lambda_i^*} &
			\frac{(\Phi_i^{[i-1]})^{\dagger} (\Phi_i^{[i-1]})^*}{-\lambda_i^*-\lambda_i^*} \\[5pt]
			\frac{(\Phi_i^{[i-1]})^{\top}\Phi_i^{[i-1]}}{\lambda_i+\lambda_i} &
			\frac{(\Phi_i^{[i-1]})^{\top}(\Phi_i^{[i-1]})^*}{-\lambda_i^*+\lambda_i}
		\end{bmatrix}, \quad \Phi_i:=\Phi( x,t;\lambda_i)\mathbf{c}_i=\Phi( x,t;\lambda_i)\left[ 1,c_i\right]^{\top}.
		\end{split}
	\end{equation}
It should be noticed that the value of $c_i$, $i=1,2,\cdots,N$, depends on the spectral parameter $\lambda_i$. If $\lambda_i\in \mathbb{C}\backslash(\ii \mathbb{R}\cup \mathbb{R})$, $i=1,2,\cdots,N$, the range of $c_i$ is $\mathbb{C}\backslash\{ 0 \}$. If $\lambda_i\in \ii \mathbb{R}$, $i=1,2,\cdots,N$, for satisfying the symmetric \eqref{eq:T-sym} of the Darboux matrix $\mathbf{T}^{\mathrm{P}}_i(\lambda; x,t)$, we could choose $c_i\in \mathbb{R}$ by Remark \ref{remark:T}. For the convenience of the following analysis, we define the numbers of $\mathbf{T}^{\mathrm{P}}_{i}(\lambda; x,t)$ and $\mathbf{T}^{\mathrm{C}}_{i}(\lambda; x,t)$, $i=1,2,\cdots,N$, in $\mathbf{T}^{[N]}(\lambda;x,t)$ \eqref{eq:multi-T} as $n_1$ and $n_2$ respectively, which denotes that the numbers of spectral parameters $\lambda_i\in \ii \mathbb{R}$ and $\lambda_i\in \mathbb{C}\backslash(\ii \mathbb{R}\cup \mathbb{R})$, $i=1,2,\cdots,N$, are $n_1$ and $n_2$, respectively. Therefore, based on the Darboux-B\"{a}cklund transformation, we obtain the new solutions of the mKdV equation \eqref{eq:mKdV-equation} as follows.

\begin{theorem}\label{theorem:u-N-breather}
	Based on the Theorem \ref{theorem:T}, the new solution of mKdV equation \eqref{eq:mKdV-equation} could be expressed as 
	\begin{equation}\label{eq:u-N-breather}
		u^{[N]}(x,t)=u( x,t)+2\ii\mathbf{X}_{m,1}\mathbf{M}_{m}^{-1}( x,t)\mathbf{X}^{\dagger}_{m,2}=\frac{u^{1-m}( x,t)}{\det(\mathbf{M}_m( x,t))} \det\left(u(x,t)\mathbf{M}_m( x,t)-2\ii \mathbf{X}^{\dagger}_{m,2}\mathbf{X}_{m,1}\right),
	\end{equation}
where 
\begin{equation}\label{eq:define-M-X}
	\mathbf{M}_m( x,t)=\left( \frac{\Phi_i^{\dagger}\Phi_j}{\lambda_j-\lambda_i^*} \right)_{1\le i,j\le m}, \qquad \mathbf{X}_{m}=[\Phi_1,\Phi_2,\cdots,\Phi_m],
\end{equation}
where $\Phi_1$ is defined in \eqref{eq:T-1} and $\Phi_i$, $i=2,3,\cdots,N$ are defined in \eqref{eq:T-i-P-C}. The dimension of $\mathbf{M}_m( x,t)$ and $\mathbf{X}_m$ depends on the transformation $\mathbf{T}^{[N]}(\lambda;x,t)$ in \eqref{eq:multi-T} satisfying $N\le m=n_1+2n_2\le 2N$, where $n_1$, $n_2$ are defined as the numbers of $\mathbf{T}^{\mathrm{P}}_{i}(\lambda; x,t)$ and $\mathbf{T}^{\mathrm{C}}_{i}(\lambda; x,t)$, $i=1,2,\cdots,N$, in the multi-Darboux matrix $\mathbf{T}^{[N]}(\lambda;x,t)$ \eqref{eq:multi-T}, respectively.
\end{theorem}

The proof of Theorem \ref{theorem:u-N-breather} is provided in Appendix \ref{appendix:Darboux transformation}. The right side of equation \eqref{eq:u-N-breather} could be obtained by the Sherman-Morrison-Woodbury-type matrix identity. And then, we aim to provide the exact expressions of solutions $u^{[N]}(x,t)$ in \eqref{eq:u-N-breather} by utilizing theta functions.

\subsection{The explicit expression of solutions}

Using theta functions to express the exact solution $u^{[N]}( x,t)$ is our primary purpose in this subsection. Firstly, we use the theta functions to express $\lambda(z)$, which is useful in studying the exact expressions for $u^{[N]}(x,t)$.

\begin{lemma}\label{lemma:lambda-i-j}
	The function $\lambda(z_j)-\lambda(z_i^*)$ defined in equation \eqref{eq:lambda-elliptic} with different values of $l$ could be rewritten as a uniform expression:
	\begin{equation}\label{eq:lambda-i-j}
			\lambda(z_j)-\lambda(z_i^*)
			=\frac{\ii \alpha \vartheta_2 \vartheta_4}{2\vartheta_3}\frac{\vartheta_3(\frac{2\ii l}{2K})\vartheta_1(\frac{\ii(z_j-z_i^*)}{2K})\vartheta_3(\frac{\ii(z_j+z_i^*)}{2K})}{\vartheta_2(\frac{\ii(z_j+l)}{2K})\vartheta_2(\frac{\ii(z_i^*+l)}{2K})\vartheta_4(\frac{\ii(z_j-l)}{2K})\vartheta_4(\frac{\ii(z_i^*-l)}{2K})}, \qquad l=0, \quad \text{and} \quad l=\frac{K'}{2}.
	\end{equation}
\end{lemma}

\begin{proof}
	Firstly, we consider the condition $l=0$. To simplify the notation and facilitate analysis, we set $\ii (z_j-l)$ as a whole variable $x$ and $\ii (z_i^*-l)$ as a constant $\beta$. Combining with equation \eqref{eq:lambda-elliptic-0}, we get 
	\begin{equation}\label{eq:lambda-0-ij}
		\lambda(z_j)-\lambda(z_i^*)=\frac{\ii \alpha}{2} \frac{\dn(x)\sn(x)\cn(\beta)-\dn(\beta)\sn(\beta)\cn(x)}{\cn(x)\cn(\beta)}.
	\end{equation}
It is straightforward to verify that $2K$ and $2\ii K'$ are two periods of function \eqref{eq:lambda-0-ij}. Therefore, we just consider a periodic parallelogram starting from $\left(-\frac{K}{2},-\frac{K'}{2}\right)$ and taking $\left(-\frac{K}{2},-\frac{K'}{2}\right)$, $\left(\frac{3K}{2},-\frac{K'}{2}\right)$, $\left(\frac{3K}{2},\frac{3K'}{2}\right)$, $\left(-\frac{K}{2},\frac{3K'}{2}\right)$ as vertices. By the zeros and poles of elliptic functions in equation \eqref{eq:J 0in}, we know that $x=K$ and $x=\ii K'$ are both the simple poles of equation \eqref{eq:lambda-0-ij}.
It is easy to obtain that $ x=\beta+2nK+2m\ii K'$, $m,n\in \mathbb{Z}$, are the simple zeros of equation \eqref{eq:lambda-0-ij}.  Since 
	\begin{equation} 
		\begin{split}
			&\dn(x)\sn(x)\cn(K+\ii K'-x)-\dn(K+\ii K'-x)\sn(K+\ii K'-x)\cn(x)
			=\ii k'\frac{\sn(x)\dn(x)-\sn(x)\dn(x)}{k\cn(x)},
		\end{split}
	\end{equation} 
obtained by shift formulas \eqref{eq:Jacobi-shift}, we get $ x=-\beta+(2n+1)K+(2m+1)\ii K'$, $m,n\in \mathbb{Z}$, which are also the simple zeros of equation \eqref{eq:lambda-0-ij}. Furthermore, it is easy to verify that there only exist two simple zeros in this periodic parallelogram. Thus, by the theory of theta functions and Liouville theorem, we obtain 
\begin{equation}\label{eq:lambda-0-ij-C}	\lambda(z_j)-\lambda(z_i^*)=C\frac{\vartheta_1(\frac{x-\beta}{2K})\vartheta_3(\frac{x+\beta}{2K})}{\vartheta_2(\frac{x}{2K})\vartheta_4(\frac{x}{2K})}
	=C\frac{\vartheta_1(\frac{\ii(z_j-z_i^*)}{2K})\vartheta_3(\frac{\ii(z_j+z_i^*)}{2K})}{\vartheta_2(\frac{\ii(z_j+l)}{2K})\vartheta_4(\frac{\ii(z_j-l)}{2K})}, \qquad l=0.
\end{equation}
	Plugging $\ii(z_j-l)=x=0$ into \eqref{eq:lambda-0-ij} and \eqref{eq:lambda-0-ij-C}, we get $-\frac{\ii \alpha}{2}\frac{ \sn(\ii z^*_i)\dn(\ii z^*_i)}{\cn(\ii z^*_i)}=-C\frac{\vartheta_1(\frac{\ii z_i^*}{2K})\vartheta_3(\frac{\ii z_i^*}{2K})}{\vartheta_2\vartheta_4}$,
	which implies $C=\frac{\ii \alpha \vartheta_2 \vartheta_4\vartheta_3(\frac{2\ii l}{2K})}{2\vartheta_3\vartheta_2(\frac{\ii z_i^*}{2K})\vartheta_4(\frac{\ii z_i^*}{2K})}$, by the conversion formulas \eqref{eq:sn vartheta}. Therefore, equation \eqref{eq:lambda-i-j} holds at $l=0$.
	
	Similarly, when $l=\frac{K'}{2}$, by equation \eqref{eq:lambda-elliptic-K}, we get
	\begin{equation}\label{eq:lambda-K-ij}
		\lambda(z_j)-\lambda(z_i^*)=\frac{\ii \alpha k^2}{2}\frac{\sn(x)\cn(x)\dn(\beta)-\sn(\beta)\cn(\beta)\dn(x)}{\dn(\beta)\dn(x)}. 	\end{equation}
	We consider the same periodic parallelogram as the case $l=0$. It is easy to obtain that $x=K+\ii K'$ and $x=\ii K'$ are two single poles of equation \eqref{eq:lambda-K-ij} in the above periodic parallelogram and $x=\beta$ is a simple zero of \eqref{eq:lambda-K-ij}. Due to 
	\begin{equation}
		\sn(K-x)\cn(K-x)\dn(x)-\sn(x)\cn(x)\dn(x+K)=\frac{k'\sn(x)\cn(x)-k'\sn(x)\cn(x)}{\dn(x)}=0,
	\end{equation} 
we get $x=-\beta+K$ is also a simple zero of  equation \eqref{eq:lambda-K-ij}. Similarly, we obtain 
\begin{equation}\label{eq:lambda-K-ij-C}
	\lambda(z_j)-\lambda(z_i^*)=C\frac{\vartheta_1(\frac{x-\beta}{2K})\vartheta_3(\frac{x+\beta+\ii K'}{2K})}{\vartheta_2(\frac{x+\ii K'}{2K})\vartheta_4(\frac{x}{2K})}
	=C\frac{\vartheta_1(\frac{\ii(z_j-z_i^*)}{2K})\vartheta_3(\frac{\ii(z_j+z_i^*)}{2K})}{\vartheta_2(\frac{\ii(z_j+l)}{2K})\vartheta_4(\frac{\ii(z_j-l)}{2K})}, \qquad l=\frac{K'}{2},
\end{equation}
by the Liouville theorem. Substituting $\ii(z_j-l)=x=0$ into equations \eqref{eq:lambda-K-ij} and \eqref{eq:lambda-K-ij-C}, we obtain equation \eqref{eq:lambda-i-j}.
\end{proof}

Combining with the matrix function $\Phi( x,t;\lambda)$ in \eqref{eq:Lax-solution-Phi} and addition formulas of the theta functions in \eqref{eq:Jacobi Theta uv}, we obtain the following functions directly:
	 \begin{equation}\label{eq:Phi-2-1-N}
	 	\begin{split}
	 		\Phi_{i,2}^*\Phi_{j,1}
	 		=&-\frac{\alpha^2 \vartheta_2^2\vartheta_4^2}{\vartheta_3^2\vartheta_4^2(\frac{\alpha\xi}{2K})\ee^{\alpha\xi Z(2\ii l+K)}}
	 		\mathbf{E}_i^{\dagger}
	 		\begin{bmatrix}
	 			\frac{\vartheta_3(\frac{\ii(z_i^*+l)+\alpha\xi}{2K})\vartheta_1(\frac{\ii(z_j-l)-\alpha\xi}{2K})}{\vartheta_2(\frac{\ii(z_i^*+l)}{2K})\vartheta_4(\frac{\ii(z_j+l)}{2K})}& 
	 			\frac{\vartheta_3(\frac{\ii(z_i^*+l)+\alpha\xi}{2K})}{\vartheta_2(\frac{\ii(z_i^*+l)}{2K})}\frac{\vartheta_3(\frac{\ii(z_j+l)+\alpha\xi}{2K})}{\vartheta_2(\frac{\ii(z_j+l)}{2K})} \\
	 			-\frac{\vartheta_1(\frac{\ii(z_i^*-l)-\alpha\xi}{2K})}{\vartheta_4(\frac{\ii(z_i^*-l)}{2K})}\frac{\vartheta_1(\frac{\ii(z_j-l)-\alpha\xi}{2K})}{\vartheta_4(\frac{\ii(z_j+l)}{2K})} & 
	 			-\frac{\vartheta_1(\frac{\ii(z_i^*-l)-\alpha\xi}{2K})}{\vartheta_4(\frac{\ii(z_i^*-l)}{2K})}\frac{\vartheta_3(\frac{\ii(z_j+l)+\alpha\xi}{2K})}{\vartheta_2(\frac{\ii(z_j+l)}{2K})}
	 		\end{bmatrix}\mathbf{E}_j,
	 	\end{split}
	 \end{equation}
 where $\mathbf{E}_j$ is defined in equation \eqref{eq:r_i}. Using the addition and shift formulas of theta functions in \eqref{eq:Jacobi Theta-K iK} and \eqref{eq:Jacobi Theta uv}, we obtain 
	\begin{equation}\label{eq:Phi-Phi}
	\begin{split}
		\Phi_{i,1}^*\Phi_{j,1}+\Phi_{i,2}^*\Phi_{j,2}
		=&\frac{\alpha^2 \vartheta_2^2\vartheta_4^2}{\vartheta_3^2\vartheta_4^2(\frac{\alpha\xi}{2K})}
		\mathbf{E}_i^{\dagger}
		\begin{bmatrix}
			-\frac{\vartheta_1(\frac{\ii(z_i^*-l)+\alpha\xi}{2K})\vartheta_1(\frac{\ii(z_j-l)-\alpha\xi}{2K})}{\vartheta_4(\frac{\ii(z_i^*-l)}{2K})\vartheta_4(\frac{\ii(z_j-l)}{2K})} & 
			-\frac{\vartheta_1(\frac{\ii(z_i^*-l)+\alpha\xi}{2K})\vartheta_3(\frac{\ii(z_j+l)+\alpha\xi}{2K})}{\vartheta_4(\frac{\ii(z_i^*-l)}{2K})\vartheta_2(\frac{\ii(z_j+l)}{2K})} \\
			\frac{\vartheta_3(\frac{\ii(z_i^*+l)-\alpha\xi}{2K})\vartheta_1(\frac{\ii(z_j-l)-\alpha\xi}{2K})}{\vartheta_2(\frac{\ii(z_i^*+l)}{2K})\vartheta_4(\frac{\ii(z_j-l)}{2K})} & 
			\frac{\vartheta_3(\frac{\ii(z_i^*+l)-\alpha\xi}{2K})\vartheta_3(\frac{\ii(z_j+l)+\alpha\xi}{2K})}{\vartheta_2(\frac{\ii(z_i^*+l)}{2K})\vartheta_2(\frac{\ii(z_j+l)}{2K})}
		\end{bmatrix}\mathbf{E}_j \\
		&+\frac{\alpha^2 \vartheta_2^2\vartheta_4^2}{\vartheta_3^2\vartheta_4^2(\frac{\alpha\xi}{2K})}
		\mathbf{E}_i^{\dagger}
		\begin{bmatrix}
			\frac{\vartheta_3(\frac{\ii(z_i^*+l)+\alpha\xi}{2K})\vartheta_3(\frac{\ii(z_j+l)-\alpha\xi}{2K})}{\vartheta_2(\frac{\ii(z_i^*+l)}{2K})\vartheta_2(\frac{\ii(z_j+l)}{2K})} & 
			\frac{\vartheta_3(\frac{\ii(z_i^*+l)+\alpha\xi}{2K})\vartheta_1(\frac{\ii(z_j-l)+\alpha\xi}{2K})}{\vartheta_2(\frac{\ii(z_i^*+l)}{2K})\vartheta_4(\frac{\ii(z_j-l)}{2K})} \\
			-\frac{\vartheta_1(\frac{\ii(z_i^*-l)-\alpha\xi}{2K})\vartheta_3(\frac{\ii(z_j+l)-\alpha\xi}{2K})}{\vartheta_4(\frac{\ii(z_i^*-l)}{2K})\vartheta_2(\frac{\ii(z_j+l)}{2K})} & 
			-\frac{\vartheta_1(\frac{\ii(z_i^*-l)-\alpha\xi}{2K})\vartheta_1(\frac{\ii(z_j-l)+\alpha\xi}{2K})}{\vartheta_4(\frac{\ii(z_i^*-l)}{2K})\vartheta_4(\frac{\ii(z_j-l)}{2K})}
		\end{bmatrix}
		\mathbf{E}_j \\
		=&\frac{\alpha^2 \vartheta_2^2\vartheta_4^2B}{\vartheta_3^2\vartheta_4(\frac{\alpha\xi}{2K})}
		\mathbf{E}_i^{\dagger}
		\begin{bmatrix}
			\vartheta_3(\frac{\ii(z_i^*+z_j)}{2K})\vartheta_4(\frac{\ii(z_i^*-z_j)+\alpha\xi}{2K}) & 
			-\vartheta_1(\frac{\ii(z_i^*-z_j)}{2K})\vartheta_2(\frac{\ii(z_j+z_i^*)+\alpha\xi}{2K}) \\
			-\vartheta_1(\frac{\ii(z_i^*-z_j)}{2K})\vartheta_2(\frac{\ii(z_j+z_i^*)-\alpha\xi}{2K}) & 
			\vartheta_3(\frac{\ii(z_i^*+z_j)}{2K})\vartheta_4(\frac{\ii(z_i^*-z_j)-\alpha\xi}{2K})
		\end{bmatrix}\mathbf{E}_j,
	\end{split}	
\end{equation}
	where $B=\frac{\vartheta_3(\frac{2\ii l}{2K})\vartheta_4(\frac{\ii(z_i^*-z_j)+\alpha\xi}{2K})}{\vartheta_2(\frac{\ii(z_i^*+l)}{2K})\vartheta_2(\frac{\ii(z_j+l)}{2K})\vartheta_4(\frac{\ii(z_i^*-l)}{2K})\vartheta_4(\frac{\ii(z_j-l)}{2K})}$, $\mathbf{E}_j$ is defined in \eqref{eq:r_i}. Therefore, combining with equation \eqref{eq:lambda-i-j}, we gain
	\begin{equation}\label{eq:M}
		\frac{\Phi^{\dagger}_i\Phi_j}{2(\lambda_j-\lambda_i^*)}=-\frac{\ii\alpha \vartheta_2\vartheta_4}{\vartheta_3\vartheta_4(\frac{\alpha\xi}{2K})}
		\mathbf{E}_i^{\dagger}
		\begin{bmatrix}
			-\frac{\vartheta_4(\frac{\ii(z_i^*-z_j)+\alpha\xi}{2K})}{\vartheta_1(\frac{\ii(z_i^*-z_j)}{2K})} & 
			\frac{\vartheta_2(\frac{\ii(z_i^*+z_j)+\alpha\xi}{2K})}{\vartheta_3(\frac{\ii(z_i^*+z_j)}{2K})} \\
			\frac{\vartheta_2(\frac{-\ii(z_i^*+z_j)+\alpha\xi}{2K})}{\vartheta_3(\frac{-\ii(z_i^*+z_j)}{2K})} & 
			\frac{\vartheta_4(\frac{\ii(z_j-z_i^*)+\alpha\xi}{2K})}{\vartheta_1(\frac{\ii(z_j-z_i^*)}{2K})}
		\end{bmatrix}
	\mathbf{E}_j.
	\end{equation}
	Furthermore, by formulas \eqref{eq:Jacobi Theta-K iK} and \eqref{eq:Jacobi Theta uv}, we also could obtain
	\begin{equation}\label{eq:uM-N}
		\begin{split}
		\frac{	u(x,t)\Phi^{\dagger}_i\Phi_j}{2(\lambda_j-\lambda_i^*)}-\ii\Phi_{i,2}^*\Phi_{j,1}
			=& -\frac{\ii\alpha^2 \vartheta_2^2\vartheta_4^2A}{\vartheta_3^2\vartheta_4^2(\frac{\alpha\xi}{2K})}
			\mathbf{E}_i^{\dagger}
			\begin{bmatrix}
				\frac{\vartheta_2(\frac{\alpha\xi+2\ii l}{2K})\vartheta_4(\frac{\ii(z_i^*-z_j)+\alpha\xi}{2K})}{\vartheta_3(\frac{2\ii l}{2K})\vartheta_1(\frac{\ii(z_j-z_i^*)}{2K})} & 
				\frac{\vartheta_2(\frac{\alpha\xi+2\ii l}{2K})\vartheta_2(\frac{\ii(z_i^*+z_j)+\alpha\xi}{2K})}{\vartheta_3(\frac{2\ii l}{2K})\vartheta_3(\frac{\ii(z_i^*+z_j)}{2K})} \\
				\frac{\vartheta_2(\frac{\alpha\xi+2\ii l}{2K})\vartheta_2(\frac{\ii(z_i^*+z_j)-\alpha\xi}{2K})}{\vartheta_3(\frac{2\ii l}{2K})\vartheta_3(\frac{\ii(z_i^*+z_j)}{2K})} & 
				\frac{\vartheta_2(\frac{\alpha\xi+2\ii l}{2K})\vartheta_4(\frac{\ii(z_i^*-z_j)-\alpha\xi}{2K})}{\vartheta_3(\frac{2\ii l}{2K})\vartheta_1(\frac{\ii(z_j-z_i^*)}{2K})}
			\end{bmatrix}
		\mathbf{E}_j \\
			&+\frac{\ii\alpha^2 \vartheta_2^2\vartheta_4^2A}{\vartheta_3^2\vartheta_4^2(\frac{\alpha\xi}{2K})}
			\mathbf{E}_i^{\dagger}
			\begin{bmatrix}
				\frac{\vartheta_3(\frac{\ii(z_i^*+l)+\alpha\xi}{2K})}{\vartheta_2(\frac{\ii(z_i^*+l)}{2K})}\frac{\vartheta_1(\frac{\ii(z_j-l)-\alpha\xi}{2K})}{\vartheta_4(\frac{\ii(z_j+l)}{2K})} & 
				\frac{\vartheta_3(\frac{\ii(z_i^*+l)+\alpha\xi}{2K})}{\vartheta_2(\frac{\ii(z_i^*+l)}{2K})}\frac{\vartheta_3(\frac{\ii(z_j+l)+\alpha\xi}{2K})}{\vartheta_2(\frac{\ii(z_j+l)}{2K})} \\
				-\frac{\vartheta_1(\frac{\ii(z_i^*-l)-\alpha\xi}{2K})}{\vartheta_4(\frac{\ii(z_i^*-l)}{2K})}\frac{\vartheta_1(\frac{\ii(z_j-l)-\alpha\xi}{2K})}{\vartheta_4(\frac{\ii(z_j+l)}{2K})} & 
				-\frac{\vartheta_1(\frac{\ii(z_i^*-l)-\alpha\xi}{2K})}{\vartheta_4(\frac{\ii(z_i^*-l)}{2K})}\frac{\vartheta_3(\frac{\ii(z_j+l)+\alpha\xi}{2K})}{\vartheta_2(\frac{\ii(z_j+l)}{2K})}
			\end{bmatrix}
		\mathbf{E}_j \\
			=&\frac{-\ii\alpha^2 \vartheta_2^2\vartheta_4^2A}{\vartheta_3^2\vartheta_3(\frac{2\ii l}{2K})\vartheta_4(\frac{\alpha\xi}{2K})}
			\mathbf{E}_i^{\dagger}\mathbf{r}_i^*
			\begin{bmatrix}
				-\frac{\vartheta_2(\frac{\ii(z_i^*-z_j+2l)+\alpha\xi}{2K})}{\vartheta_1(\frac{\ii(z_i^*-z_j)}{2K})} & 
				-\frac{\vartheta_4(\frac{\ii(z_i^*+z_j+2l)+\alpha\xi}{2K})}{\vartheta_3(\frac{\ii(z_i^*+z_j)}{2K})} \\
				\frac{\vartheta_4(\frac{\ii(z_i^*+z_j-2l)-\alpha\xi}{2K})}{\vartheta_3(\frac{\ii(z_i^*+z_j)}{2K})} & 
				\frac{\vartheta_2(\frac{\ii(z_j-z_i^*+2l)+\alpha\xi}{2K})}{\vartheta_1(\frac{\ii(z_j-z_i^*)}{2K})}
			\end{bmatrix}
		\mathbf{r}_j^{-1}\mathbf{E}_j,
		\end{split}
	\end{equation}
with $A=\ee^{-\alpha\xi Z(2\ii l+K)}$, where $\mathbf{E}_j$ and $\mathbf{r}_j$ are defined in \eqref{eq:r_i}. Collecting equations \eqref{eq:u-N-breather}, \eqref{eq:M} and \eqref{eq:uM-N}, function \eqref{eq:mKdV-solution-xi-n-1} holds. Thus, we obtain Theorem \ref{theorem:u-N}.

\begin{remark}\label{remark:mKdV-solution-new}
	Combining Theorem \ref{theorem:u-N}, solution \eqref{eq:mKdV-solution-xi-n-1} could be rewritten as
		\begin{equation}\label{eq:mKdV-solution-xi-n-1-vec}
		\begin{split}
			u^{[N]}( x,t)=&\frac{\alpha \vartheta_2\vartheta_4}{\vartheta_3\vartheta_3(\frac{2\ii l}{2K})}\left(\frac{\vartheta_4(\frac{\alpha \xi}{2K})}{\vartheta_2(\frac{\alpha\xi+2\ii l}{2K})}\right)^{m-1}
			\frac{\det\left(\sum_{i,j=1}^2(\ii)^{j+i}\mathbf{X}_j^{\dagger}\mathbf{Y}^{[2]}_j \mathcal{M}^{[j,i]}\mathbf{Y}^{[1]}_{i}\mathbf{X}_i \right)}{\det\left(\sum_{i,j=1}^2(-1)^{i} \mathbf{X}_j^{\dagger}\mathcal{D}^{[j,i]}\mathbf{X}_i\right)}\ee^{-\alpha\xi Z(2\ii l+K)},
		\end{split}
	\end{equation}
		where $\xi=x-st$,
		\begin{equation}\label{eq:M-H-zi-zj}
			\begin{split}
				\mathcal{D}^{[h,n]}=&\left(\frac{\vartheta_4\left(\frac{\alpha\xi+\zeta^{[h]}_i+\mu^{[n]}_j}{2K}\right)}{\vartheta_1\left(\frac{\zeta^{[h]}_i+\mu^{[n]}_j}{2K}\right)}\right)_{1\le i,j\le m}, \qquad \mathcal{M}^{[h,n]}=\left(\frac{\vartheta_2\left(\frac{\alpha\xi+2\ii l+\zeta^{[h]}_i+\mu^{[n]}_j}{2K}\right)}{\vartheta_1\left(\frac{\zeta^{[h]}_i+\mu^{[n]}_j}{2K}\right)}\right)_{1\le i,j\le m},
			\end{split}	
		\end{equation}
		\begin{equation}\label{eq:X1-X2}
			\mathbf{X}_1:=\mathrm{diag} \left(
			\ee^{-\eta_1}, \ee^{-\eta_2}, \cdots , \ee^{-\eta_{h-1}}, \overbrace{1,\cdots ,1}^{m-h}
			\right), \qquad 
			\mathbf{X}_2:=\mathrm{diag} \left( \overbrace{1,\cdots ,1}^{h-1},
			\ee^{\eta_h},\ee^{\eta_{h+1}}, \cdots , \ee^{\eta_{m}}
			\right),
		\end{equation}
		or 
		\begin{equation}\label{eq:X1-X2-2}
			\mathbf{X}_1:=\mathrm{diag} \left(
			\overbrace{1,\cdots ,1}^{h-1},
			-\ee^{\eta_h},-\ee^{\eta_{h+1}}, \cdots , -\ee^{\eta_{m}}
			\right), \qquad 
			\mathbf{X}_2:=\mathrm{diag} \left( \ee^{\eta_1}, \ee^{\eta_2}, \cdots , \ee^{\eta_{h-1}}, \overbrace{1,\cdots ,1}^{m-h}
			\right),
		\end{equation}
		with
		\begin{equation}\label{eq:Y}
			\begin{split}
			\mathbf{Y}^{[1]}_{1}=& \mathbf{Y},\quad  \mathbf{Y}^{[1]}_2= \mathbf{Y}^{-1}, \quad 
			\mathbf{Y}^{[2]}_1= \mathbf{Y}^{*-1},\quad \mathbf{Y}^{[2]}_2= \mathbf{Y}^{*},\qquad	\mathbf{Y}:=\exp\left(\frac{l\pi}{2K}\right)\mathrm{diag} \begin{bmatrix}
					r_1, r_2, \cdots , r_m
				\end{bmatrix},
			\end{split}
		\end{equation}
	and 
	\begin{equation}\label{eq:E-2-hat}
			\eta_i:=\gamma_i+I_i \xi-\Omega_i t,\quad
			I_i:=- \left(\alpha Z(\ii(z_i+l)+K+\ii K')+\alpha Z(\ii (z_i-l))\right), \quad 
			\Omega_i:=8 \ii \lambda_i y_i, \quad 
			\gamma_i:=\ln c_i.
	\end{equation}	
		Furthermore, we list the concrete sequence of $\zeta^{[j]}$ and $\mu^{[j]}$ as:
		\begin{equation}\label{eq:zeta-mu}
			\begin{split}
				\zeta^{[1]}=&\left(\zeta^{[1]}_1,\zeta^{[1]}_2,\cdots,\zeta^{[1]}_m\right)
				=\left(\ii \mathbf{Z}_1^*, \ii \mathbf{Z}_2^*, \cdots, \ii \mathbf{Z}_N^*\right),\\
				\zeta^{[2]}=&\left(\zeta^{[2]}_1,\zeta^{[2]}_2,\cdots,\zeta^{[2]}_m\right)
				=\left(-\ii \mathbf{Z}_1^*-(K+\ii K'), -\ii \mathbf{Z}_2^*-(K+\ii K'), \cdots, -\ii \mathbf{Z}_{N}^*-(K+\ii K')\right),\\
				\mu^{[1]}=&\left(\mu^{[1]}_1,\mu^{[1]}_2,\cdots,\mu^{[1]}_m\right)
				=\left(-\ii \mathbf{Z}_1, -\ii \mathbf{Z}_2, \cdots,-\ii \mathbf{Z}_N\right),\\
				\mu^{[2]}=&\left(\mu^{[2]}_1,\mu^{[2]}_2,\cdots,\mu^{[2]}_m\right)
				=\left(\ii \mathbf{Z}_1+K+\ii K', \ii \mathbf{Z}_2+K+\ii K', \cdots, \ii \mathbf{Z}_{N}+K+\ii K'\right),
			\end{split}
		\end{equation}
	where $\mathbf{Z}_i=z_i,$ if $\lambda(z_i)\in \ii \mathbb{R}, c_i\in \mathbb{R}$ or $\mathbf{Z}_i=\left[z_i,2l-z_i^*\right]$, if $\lambda(z_i)\in \mathbb{C} \backslash  (\ii  \mathbb{R}\cup \mathbb{R})$, $c_i\in \mathbb{C}\backslash\{0\}$. 
	\end{remark}
	
	For any positive integer $N$, we could have infinite solutions by choosing different values of the uniform parameters $z_i\neq z_j$, $z_i\in S$, $i=1,2,\cdots,m$. The exact expressions of solutions $u^{[N]}(x,t)$ are obtained by equations \eqref{eq:mKdV-solution-xi-n-1} or \eqref{eq:mKdV-solution-xi-n-1-vec}. Then, we provide solutions $u^{[N]}(x,t)$ under the different elliptic function backgrounds in the following subsection.
	 
\subsection{The dynamic behaviors of solutions}\label{subsec:cn-case}

In this subsection, we illustrate solutions $u^{[N]}(x,t)$, $N=1,2$, that exhibit different dynamic behaviors, including the elliptic-soliton, elliptic-breather, two-elliptic-solitons, two-elliptic-breathers, and elliptic-soliton-breather solutions, based on Theorem \ref{theorem:u-N}.

By the expressions of solution $u^{[1]}(x,t)$ in \eqref{eq:mKdV-solution-xi-n-1}, we consider the elliptic-soliton and elliptic-breather solutions, which are divided by the parameter $\lambda(z)$.
\begin{itemize}
	\item [Case I-1: ] If $\lambda_1=\lambda(z_1)\in \ii \mathbb{R}$ and $c_1\in \mathbb{R}$, based on the Darboux matrix $\mathbf{T}^{[1]}(\lambda;x,t)=\mathbf{T}_1^{\mathrm{P}}(\lambda;x,t)$ and the B\"{a}cklund transformation, we obtain the elliptic-soliton solution $u^{[1]}(x,t)$ by equation \eqref{eq:mKdV-solution-xi-n-1} with the $1\times 1$ matrices $\mathcal{M}$ and $\mathcal{D}$.
	
	\item[Case I-2: ] If $\lambda_1=\lambda(z_1)\in \mathbb{C}\backslash(\ii \mathbb{R}\cup\mathbb{R})$ and $c_1\in \mathbb{C}\backslash\{0\}$, we obtain the elliptic-breather solution $u^{[1]}(x,t)$ by equation \eqref{eq:mKdV-solution-xi-n-1} with $2\times 2$ matrices $\mathcal{M}$ and $\mathcal{D}$ based on the Darboux matrix $\mathbf{T}^{[1]}(\lambda;x,t)=\mathbf{T}_1^{\mathrm{C}}(\lambda;x,t)$ and the B\"{a}cklund transformation. 
	From Theorem \ref{theorem:u-N-breather} and Remark \ref{remark:T}, parameters $z_1$ and $z_2$ in solution $u^{[1]}(x,t)$ \eqref{eq:mKdV-solution-xi-n-1} must satisfy $z_2=-z_1^*+2l$.
\end{itemize}

Following, we give three examples including the elliptic-soliton solutions and elliptic-breather solutions under the $\cn$-type background, i.e., $l=0$. Firstly, taking $z_1=K'+\frac{K}{3}\ii$, $k=\frac{1}{2}$ and $\alpha=1$ into equation \eqref{eq:lambda-elliptic-0}, we get $\lambda_1=\lambda(z_1)\approx -0.836 \ii\in \ii \mathbb{R}$. Therefore, plugging them into equation \eqref{eq:mKdV-solution-xi-n-1}, we obtain an elliptic-soliton solution $u^{[1]}(x,t)$ under the $\cn$-type background shown in Figure \ref{fig:cn-1b-3d}(a). By equation \eqref{eq:u-N-breather}, the range of solution  $u^{[1]}(x,t)=u(x,t)+\frac{2\ii(\lambda_1^*-\lambda_1)\Phi_{1,1}\Phi_{1,2}^*}{|\Phi_{1,1}|^2+|\Phi_{1,2}|^2}$ is $[\min(u(x,t))-2|\Im(\lambda_1)|, \max(u(x,t))+2|\Im(\lambda_1)|]$. And only when $\frac{\Phi_{1,1}}{\Phi_{1,2}^*}=\pm 1$, the solution $u^{[1]}(x,t)$ reaches the maximum or minimum value. After numerical calculations, we know that solution $u^{[1]}(x,t)$ reaches its maximum value at the origin $(0,0)$ and the maximum value of the solution $u^{[1]}(x,t)$ is approximately equal to $2.172$.

Then, we provide two solutions by Case I-2 (i.e., the elliptic-breather solution).
Selecting the parameters $\alpha=1$, $k=\frac{1}{2}$,
$z_1=\frac{3K'}{4}+\frac{K}{3}\ii$ and $z_2=-\frac{3K'}{4}+\frac{K}{3}\ii$ and plugging them into equations \eqref{eq:lambda-elliptic-0} and \eqref{eq:mKdV-solution-xi-n-1}, we obtain an elliptic-breather solution $u^{[1]}(x,t)$ and draw in Figure \ref{fig:cn-1b-3d}(b). In addition, we also display an elliptic-breather solution whose velocity is zero, by plugging $\alpha=1$, $k=\frac{1}{2}$, $c_1=1.375-\ii$,
$z_1=\frac{K'}{3}+\ii \frac{2K}{5}$, and $z_2=-\frac{K'}{3}+\ii \frac{2K}{5}$ into equation \eqref{eq:lambda-elliptic-0} and \eqref{eq:mKdV-solution-xi-n-1}. This elliptic-breather solution $u^{[1]}(x,t)$ is shown in the Figure \ref{fig:cn-1b-3d}(c).

\begin{figure}[h]
	\centering
	\includegraphics[width=1\linewidth]{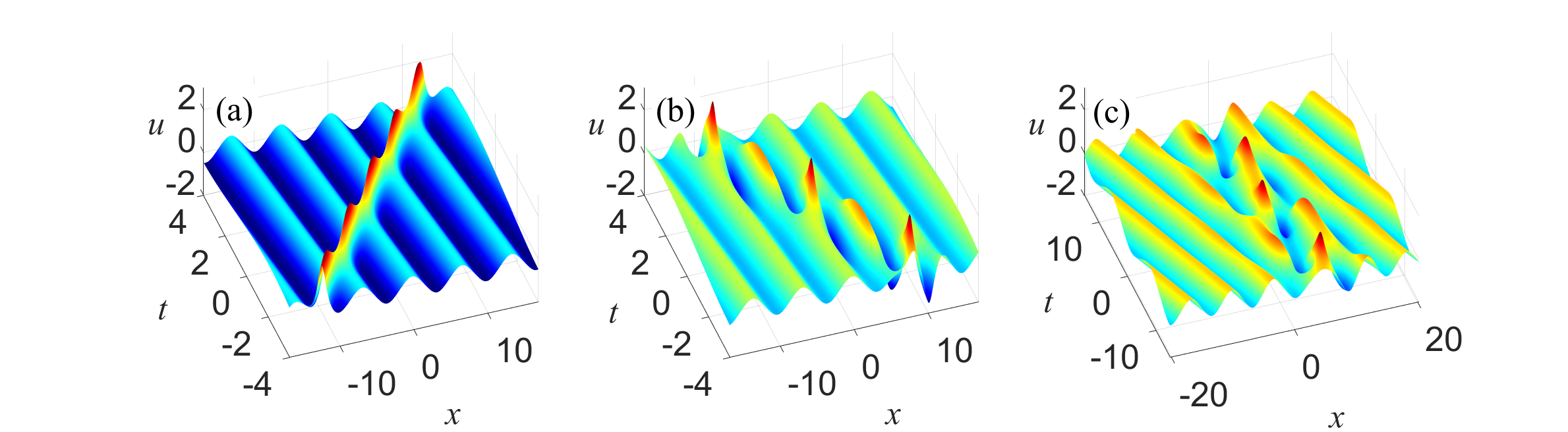}
	\caption{The 3d-plot of solutions $u^{[1]}(x,t)$ of mKdV equation \eqref{eq:mKdV-equation} under the $\cn$-type background. (a):  elliptic-soliton solution. (b) and (c): elliptic-breather solution.}
	\label{fig:cn-1b-3d}
\end{figure}

In addition, we list three solutions under the $\dn$-type background, i.e., $l=\frac{K'}{2}$. Substituting $\alpha=1$, $k=\frac{9}{10}$,
$z_1=\frac{K'}{2}+\frac{2K}{5}\ii$ into equation \eqref{eq:lambda-elliptic-K}, we find $\lambda(z_1)\approx-0.269\ii \in \ii \mathbb{R}$ and obtain an elliptic-breather solution $u^{[1]}(x,t)$ in equation \eqref{eq:mKdV-solution-xi-n-1}. Thus, the solution $u^{[1]}(x,t)$ under the $\dn$-type background in Figure \ref{fig:dn-1b-3d}(a) reaches its maximum value approximately equal to $1.539$ at the origin $(x,t)=(0,0)$. For $z_1=-\frac{K'}{2}+\ii \frac{K}{5}$, i.e., $\lambda(z_1)\approx-1.189\ii$, similarly to above, we obtain an elliptic-soliton solution $u^{[1]}(x,t)$ drawn in Figure \ref{fig:dn-1b-3d}(b) with the maximum value approximately equal to $3.378$ reaching at the origin $(x,t)=(0,0)$. Comparing Figure \ref{fig:dn-1b-3d}(a) with Figure \ref{fig:dn-1b-3d}(b), we find that the parameter $\lambda_1$ is an important factor leading to the variation of peaks.

For the Case I-2 (i.e., the elliptic-breather solution) under the $\dn$-type background, we plug parameters $\alpha=1$, $k=\frac{9}{10}$,
$z_1=-\frac{K'}{8}+\frac{2K}{5}\ii$ and $z_2=\frac{9K'}{8}+\frac{2K}{5}\ii$ into equations \eqref{eq:lambda-elliptic-K} and \eqref{eq:mKdV-solution-xi-n-1}. It follows that the elliptic-breather solution $u^{[1]}(x,t)$ with the spectral parameter $\lambda_1\approx0.103-0.535\ii$ can be obtained (see Figure \ref{fig:cn-1b-3d}(c)), which shows the elliptic-breather solution with zero velocity.

\begin{figure}[h]
	\centering
	\includegraphics[width=1\linewidth]{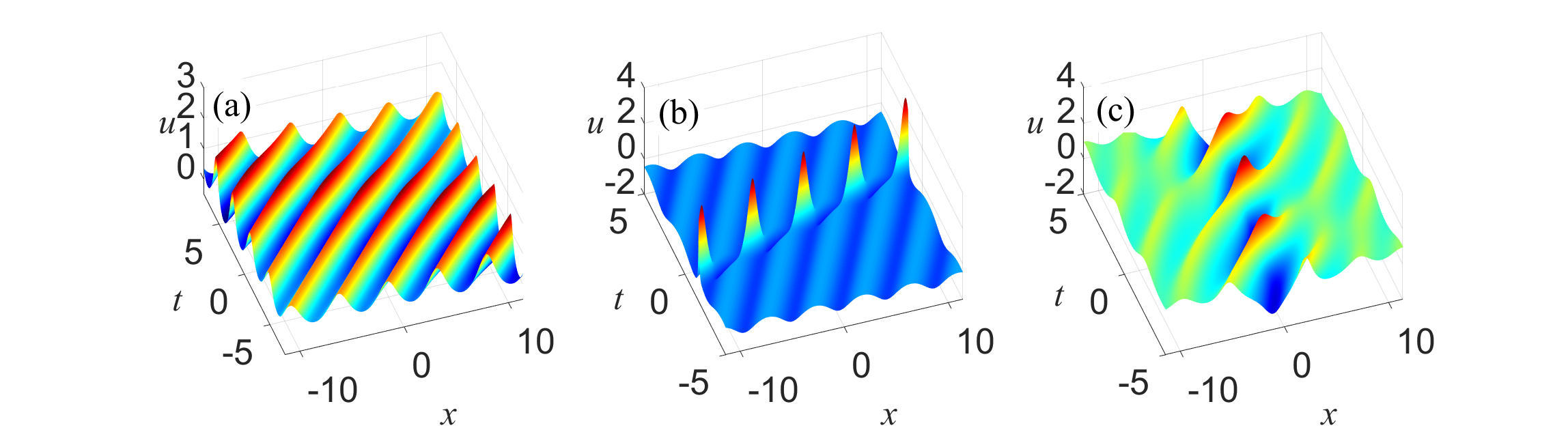}
	\caption{The 3d-plot of solutions $u^{[1]}( x,t)$ for mKdV equation \eqref{eq:mKdV-equation} under the $\dn$-type background. (a) and (b): elliptic-soliton solution. (c): elliptic-breather solution.}
	\label{fig:dn-1b-3d}
\end{figure}

By the expression of the multi elliptic-localized solutions \eqref{eq:u-N-breather} and the Darboux matrix \eqref{eq:multi-T} we choosing, the solutions $u^{[2]}(x,t)$ provide three different types, which are referred to as the two-elliptic-solitons solution, the two-elliptic-breathers solution, and the elliptic-soliton-breather solution. Now, we show this three solutions as follows:
\begin{itemize}
	\item[Case II-1:] If $\lambda_1=\lambda(z_1)$, $\lambda_2=\lambda(z_2)\in \ii \mathbb{R}$ and $c_1,c_2\in\mathbb{R}$, the solution $u^{[2]}(x,t)$ in equation \eqref{eq:mKdV-solution-xi-n-1} is called the two-elliptic-solitons solution, which is obtained by $2\times 2$ matrices $\mathcal{M}$ and $\mathcal{D}$.
	
	\item[Case II-2:] If $\lambda_1\in \ii \mathbb{R}$, $\lambda_2\in \mathbb{C}\backslash(\ii \mathbb{R}\cup \mathbb{R})$, $c_1\in \mathbb{R}, c_2\in \mathbb{C}\backslash\{0\}$, or $\lambda_1\in \mathbb{C}\backslash(\ii \mathbb{R}\cup \mathbb{R})$, $\lambda_2\in \ii \mathbb{R}$, $c_1\in \mathbb{C}\backslash\{0\}, c_2\in \mathbb{R}$, we could get the elliptic-soliton-breather solution $u^{[2]}(x,t)$ from equation \eqref{eq:mKdV-solution-xi-n-1} with $3\times 3$ matrices $\mathcal{M}$ and $\mathcal{D}$. By Theorem \ref{theorem:u-N-breather} and Remark \ref{remark:T}, parameters $z_1$, $z_2$ and $z_3$ in solution $u^{[2]}(x,t)$ satisfy $z_3=-z_2^*+2l$, $z_1\neq z_2$, or $z_3=-z_1^*+2l$, $z_1\neq z_2$.
	
	\item[Case II-3:] If $\lambda_1,\lambda_2\in \mathbb{C}\backslash(\ii \mathbb{R}\cup \mathbb{R})$ and $c_1,c_2\in \mathbb{C}\backslash\{0\}$, the solution $u^{[2]}(x,t)$ in equation \eqref{eq:mKdV-solution-xi-n-1} is called the two-elliptic-breathers solution, which is obtained by the $4\times 4$ matrices $\mathcal{M}$ and $\mathcal{D}$. And parameters $z_1$, $z_2$, $z_3$ and $z_4$ in solution $u^{[2]}(x,t)$ satisfy $z_3=-z_1^*+2l$, $z_4=-z_2^*+2l$ and $z_1\neq z_2$, based on the Theorem \ref{theorem:u-N-breather} and Remark \ref{remark:T}.
\end{itemize}

We present examples for the above three cases under the $\cn$-type and $\dn$-type backgrounds. Under the $\cn$-type background, i.e. $l=0$, by parameters $k=\frac{1}{2}$, $\alpha=1$, $z_1=K'+\frac{2K}{9}\ii$, $z_2=K'+\frac{K}{3}\ii$, $c_1=c_2=1$, $\lambda_1=\lambda(z_1)\approx -1.301\ii\in \ii \mathbb{R}$ and $\lambda_2=\lambda(z_2)\approx-0.836\ii\in \ii \mathbb{R}$, satisfying Case II-1, a two-elliptic-solitons solution is drawn in Figure \ref{fig:cn-2b-3d}(a). If $z_1=\frac{3K'}{4}+\frac{K}{3}\ii $, $z_2=K'+\frac{2K}{5}\ii$, $z_3=-\frac{3K'}{4}+\frac{K}{3}\ii$, $\alpha=1$ and $k=\frac{1}{2}$, an elliptic-soliton-breather solution $u^{[2]}(x,t)$ in \eqref{eq:mKdV-solution-xi-n-1} is drawn in Figure \ref{fig:cn-2b-3d}(b) under Case II-2 with $\lambda_1\approx-0.504-0.429\ii $ and $\lambda_2\approx-0.673\ii $ and $c_1=c_2=1$. Plugging $k=\frac{1}{2}$, $\alpha=1$, $z_1=-\frac{3K'}{4}+\frac{K}{4}\ii$,
$z_2=\frac{K'}{3}+\frac{2K}{5}\ii$, $z_3=\frac{3K'}{4}+\frac{K}{4}\ii$, and $z_4=-\frac{K'}{3}+\frac{2K}{5}\ii$, into equation \eqref{eq:mKdV-solution-xi-n-1}, we get a two-elliptic-breathers solution $u^{[2]}(x,t)$ in Figure \ref{fig:cn-2b-3d}(c), which satisfies the Case II-3 with parameters $\lambda_1\approx 0.626-0.427\ii$, $\lambda_2\approx -0.368-0.257\ii$, $c_1=3+4\ii$, and $c_2=1.375-\ii$.

\begin{figure}[h]
	\centering
	\subfigure {\includegraphics[width=1\linewidth]{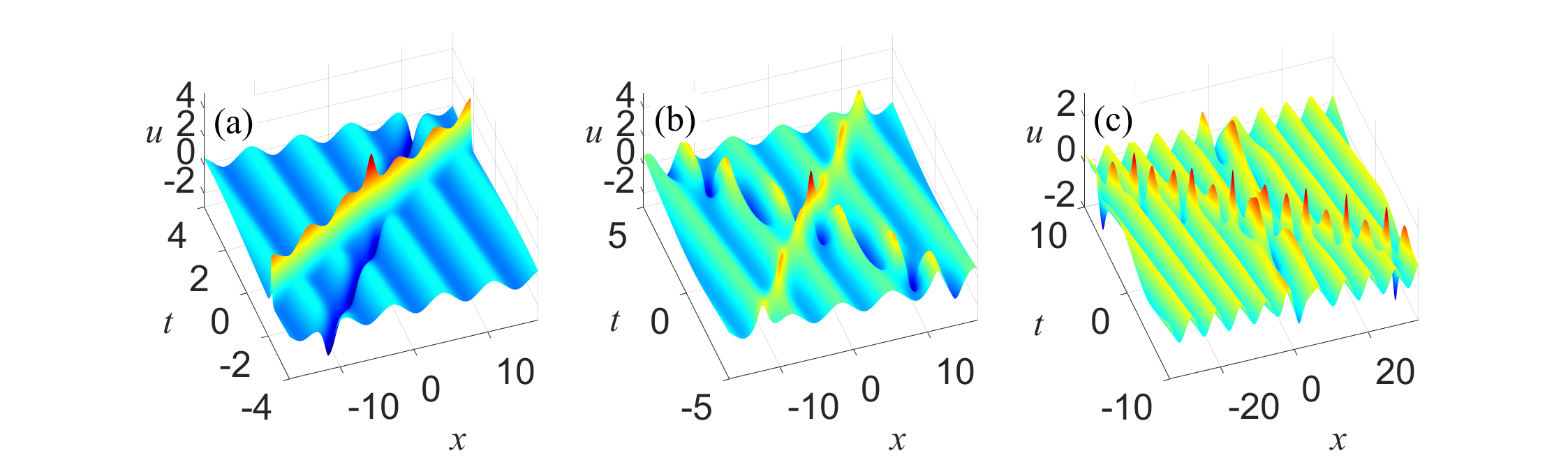}}
	\subfigure {\includegraphics[width=1\linewidth]{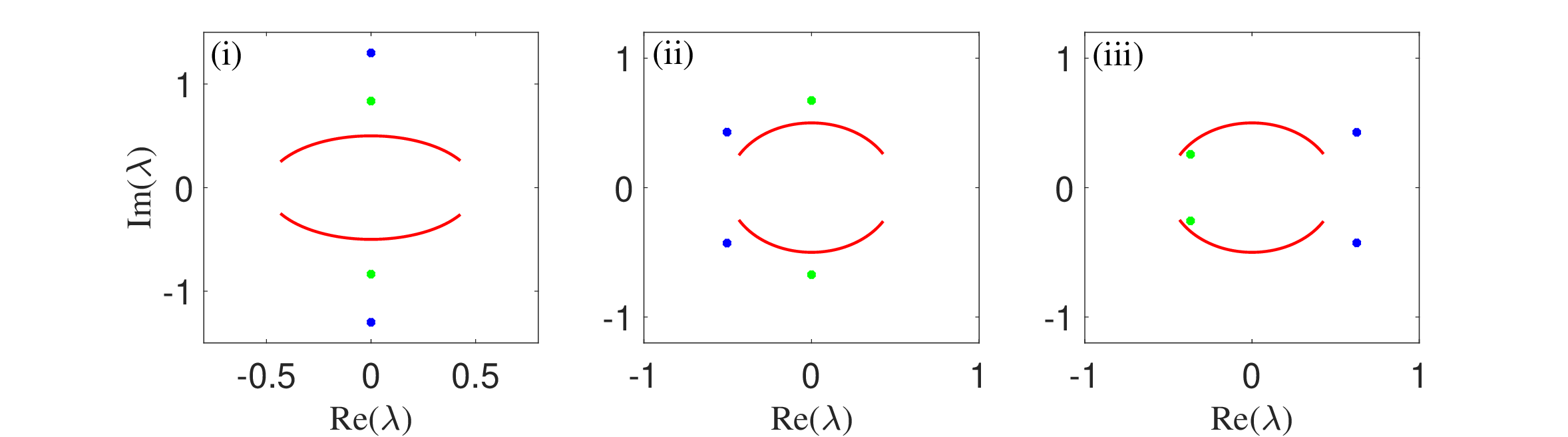}}
	\caption{The 3d-plot of solutions $u^{[2]}(x,t)$ for mKdV equation \eqref{eq:mKdV-equation} and corresponding spectral parameters $\lambda_1$, $\lambda_2$, under the $\cn$-type background.
	The solutions $u^{[2]}(x,t)$ in figures (a), (b) and (c) are called the  two-elliptic-solitons solution, the elliptic-soliton-breather solution and the  two-elliptic-breathers solution, respectively. The red curves in figures (i), (ii) and (iii) describe the cuts in $\lambda$-plane. The green and blue dots in them are the spectral points $\lambda_1$ and $\lambda_2$, respectively.}
	\label{fig:cn-2b-3d}
\end{figure}

For the $\dn$-type background, i.e. $l=\frac{K'}{2}$, we choose $z_1=\frac{K'}{2}+\frac{2K}{5}\ii$, $z_2=-\frac{K'}{2}+\frac{K}{5}\ii$, to gain a two-elliptic-solitons solution $u^{[2]}(x,t)$ by Case II-1, since $\lambda_1\approx-0.269\ii$, $\lambda_2\approx-1.189\ii$ and $c_1=c_2=1$. Plugging the above parameters and $k=\frac{9}{10}$, $\alpha=1$ into equation \eqref{eq:mKdV-solution-xi-n-1}, we construct a two-elliptic-solitons solution $u^{[2]}(x,t)$ and plot it in Figure \ref{fig:dn-2b-3d}(a). The different peaks of solitons in Figure \ref{fig:dn-2b-3d}(a) are mainly determined by parameters $\lambda_1$ and $\lambda_2$. Substituting $z_1=-\frac{K'}{8}+\frac{2K}{5}\ii$, $z_2=-\frac{K'}{2}+\frac{K}{5}\ii$, $z_3=\frac{9K'}{8}+\frac{2K}{5}\ii$, $\alpha=1$ with $\lambda_1\approx0.103-0.535\ii$, $\lambda_2\approx-1.189\ii$ and $c_1=c_2=1$ by Case II-2 into equation \eqref{eq:mKdV-solution-xi-n-1}, we obtain an elliptic-soliton-breather solution $u^{[2]}(x,t)$ and plot it in Figure \ref{fig:dn-2b-3d}(b). When $\alpha=1$, $k=\frac{9}{10}$, $z_1=-\frac{K'}{3}+\frac{K}{3}\ii$, $z_2=-\frac{K'}{8}+\frac{2K}{5}\ii$, $z_3=\frac{4K'}{3}+\frac{K}{3}\ii$, $z_4=\frac{9K'}{8}+\frac{2K}{5}\ii$ with $\lambda_1\approx0.143-0.737\ii$ and $\lambda_2\approx0.103- 0.535\ii$ in Case II-3, a two-elliptic-breathers solution $u^{[2]}(x,t)$ is obtained by equation \eqref{eq:mKdV-solution-xi-n-1} and drawn in Figure \ref{fig:dn-2b-3d}(c).  Figures \ref{fig:dn-2b-3d} (i), (ii) and (iii) show the corresponding spectral parameters $\lambda_1$ and $\lambda_2$ of the above three solutions, respectively.

\begin{figure}[h]
	\centering
	\subfigure {\includegraphics[width=1\linewidth]{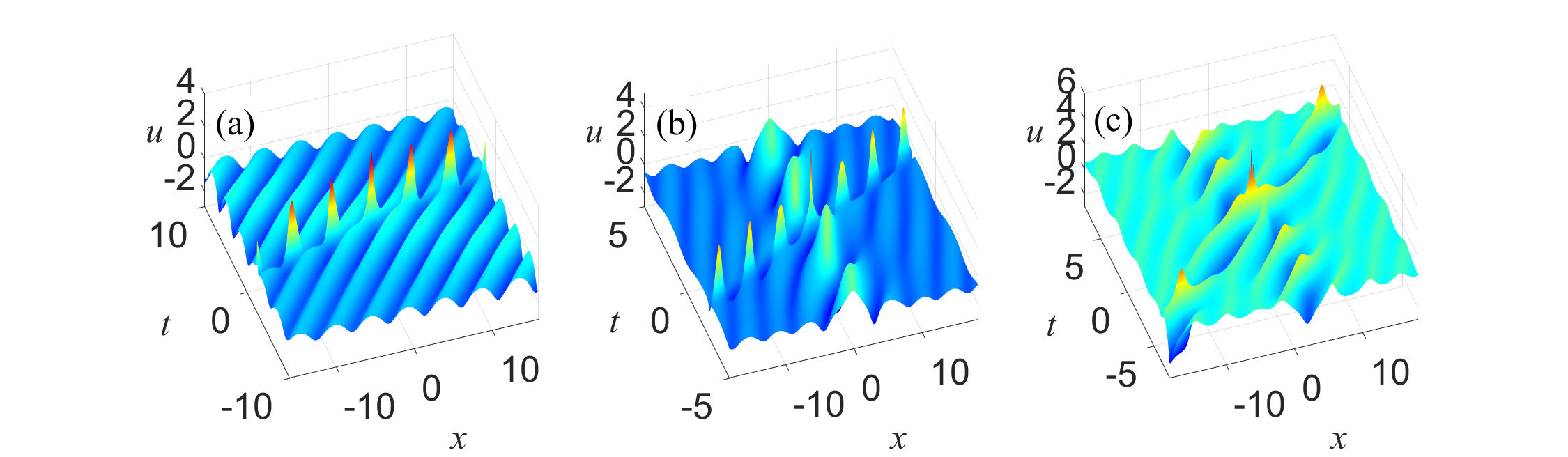}}
	\subfigure {\includegraphics[width=1\linewidth]{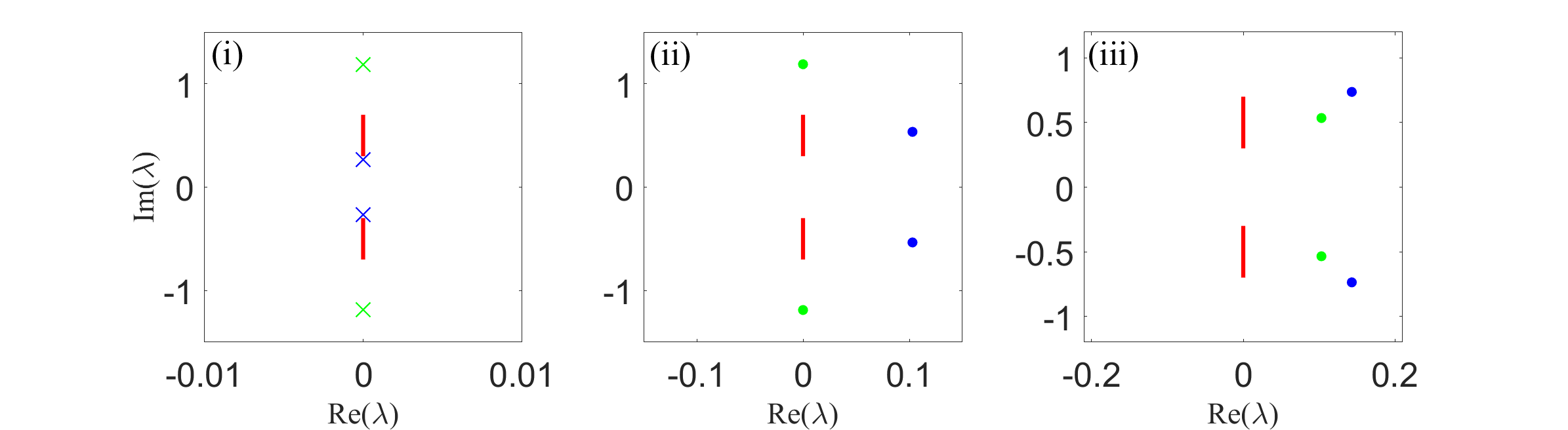}}
	\caption{The 3d-plot of two-elliptic-localized solutions $u^{[2]}(x,t)$  and the corresponding spectral parameters $\lambda_1$, $\lambda_2$ under the $\dn$-type background. The solutions $u^{[2]}(x,t)$ in figures (a), (b) and (c) are called the two-elliptic-solitons solution, the elliptic-soliton-breather solution and the two-elliptic-breathers solution respectively. The red curves in figure (i), (ii), (iii) are the cuts in $\lambda$-plane and the green and blue points are spectral points $\lambda_1$ and $\lambda_2$, respectively.}
	\label{fig:dn-2b-3d}
\end{figure}

In this section, we first show multi elliptic-localized solutions of the mKdV equation under the elliptic function background. Then the different dynamic behaviors of solutions (the elliptic-soliton solutions, elliptic-breather solutions, two-elliptic-solitons solutions, two-elliptic-breathers solutions, and elliptic-soliton-breather solutions) are exhibited with the aid of the computer graph. In what follows, we would like to describe the asymptotic dynamics of the multi elliptic-localized solutions systematically.

\section{Asymptotic behaviors for the multi elliptic-localized solutions}\label{sec:asymptotially analysis}
In the previous section, we have provided the general formulas for the multi elliptic-localized solutions of the mKdV equation \eqref{eq:mKdV-equation} and exhibited the dynamics behaviors of those solutions. From the figures of two-elliptic-localized solutions, we can see that a crucial feature of them is their elastic interaction. In this section, we would like to disclose the properties strictly by the asymptotic analysis to the multi elliptic-localized solutions, which can be regarded as the analog of multi-solitons or multi-breathers. Especially, we find a symmetry condition intimately related to the strictly elastic interaction for the multi elliptic-localized solutions.

\subsection{The variation of the velocity}

The velocities of each elliptic-soliton/breather for the multi elliptic-localized solutions $u^{[N]}(x,t)$ are crucial to study their asymptotic behaviors. We aim to study the velocities of those breathers and solitons and their asymptotic expressions. For the velocity of the multi soliton and breather solutions, it is easy to obtain the expressions of velocity in a polynomial form. But for the multi elliptic-localized solutions, the velocity is given by a harmonic function of the spectral parameter. For ease of analysis, we consider the translation $\xi=x-st$, $\tau=t$ on solution $u^{[N]}(x,t)$ defined in equation \eqref{eq:xi-x-solution}. Then, we study the relationship between the velocity $v$ and the uniform parameter $z$, and get the relationship between the velocity $v$ and the spectral parameter $\lambda$ based on the conformal mapping $\lambda(z)$.

\begin{define}\label{define:li}
	Define the line $L_i$ as 
	\begin{equation}\label{eq:li}
		L_i:=\Re(\gamma_i)+\Re(I_i) \xi-\Re(\Omega_i) t=\Re(I_i)\left(\xi-v_i t+\hat{c}_i\right)\equiv C,\qquad v_i\equiv \frac{\Re(\Omega_i)}{\Re(I_i)}, \qquad i=1,2,\cdots,N,
	\end{equation}
where $C$ is a constant and $\gamma_i,I_i=I(z_i),\Omega_i=\Omega(z_i)$ are defined in equation \eqref{eq:E-2-hat}.
\end{define}
From the dynamic behaviors of solution $\hat{u}^{[N]}(\xi,\tau)$, we find that there exist $n_1$ solitons and $n_2$ breathers in solution $\hat{u}^{[N]}(\xi,\tau)$, i.e., $N=n_1+n_2$. However, $\hat{u}^{[N]}(\xi,\tau)$ is constructed by choosing $m$ different parameters $z_i$'s with $m\times m$ matrices $\mathcal{M}$ and $\mathcal{D}$ in equation \eqref{eq:mKdV-solution-xi-n-1}. Therefore, the relationship between $z_i$ and $L_i$ is not a one-to-one correspondence. Sometimes, different parameters $z_i\neq z_{j}$, $i,j=1,2,\cdots,N$, represent the same $L_i$. The solution $\hat{u}^{[N]}(\xi,\tau)$ divides the whole space into $2N$ regions, by lines $L_i^{\pm}$, $i=1,2,\cdots,N$, as shown in Figure \ref{fig:li}. In order to have a clearer understanding of the above solution, we mainly consider the following aspects: (i) Study the relation between the velocity $v$ and the parameter $z$ in this subsection. (ii) Analyze the asymptotic behavior of multi elliptic-localized solutions along the line $L_i^{\pm},i=1,2,\cdots,N$, and the region $R_i^{\pm},i=1,2,\cdots,N$, as $\tau\rightarrow \pm \infty$ in the next subsection.
\begin{figure}[h]
	\centering
	\includegraphics[width=0.45\linewidth]{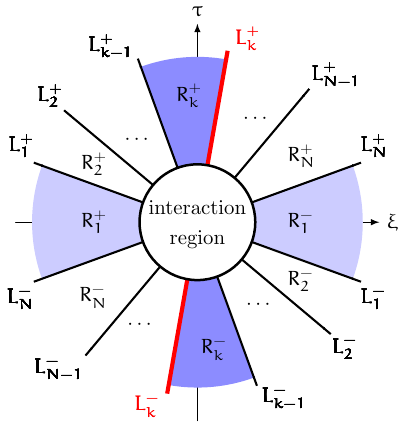}
	\caption{The sketch map for the interaction of solution $\hat{u}^{[N]}(\xi,\tau)$. The region $R_{k}^{\pm}$ is defined between the lines $L_{k-1}^{\pm}$ and $L_k^{\pm}$, $k=1,2,\cdots,N$.}
	\label{fig:li}
\end{figure}

For the definition of functions $I(z)$ and $\Omega(z)$ in equation \eqref{eq:E-2-hat}, the velocity $v(z)\equiv \frac{\Re(\Omega)}{\Re(I)}$ is mainly determined by the following two conditions, one is $\Re(I)=\Re(I(z))= \Re(-\alpha Z(\ii (z+l)+K+\ii K')-\alpha Z(\ii (z-l)))$ and the other is $\Re(\Omega)=\Re(\Omega(z))$ in equation \eqref{eq:Omega}. Both of them are the meromorphic functions of variable $z\in S$ defined in equation \eqref{eq:set-S}. To study the meromorphic function $v(z)$ clearly, we must obtain all possibilities of $z\in S$ satisfying $\Re(\Omega)\equiv 0 $, $\Re(I)\equiv 0 $, $\Re(\Omega)\equiv \infty$  and $\Re(I)\equiv \infty$. The condition $\Re(I)\equiv 0$ is studied in \cite{LinglmS-21}. Therefore, we just analyze the case of $\Re(\Omega)\equiv 0, z\in S$. For clarity, we introduce the notation $S_1$ defined as the first quadrants of $S$ in \eqref{eq:set-S}.

\begin{lemma}\label{lemma:S-boundary}
	On the boundary of $S_1$, marked as $\partial S_1$, the points satisfying $\Re(\Omega)\equiv 0$ could be classified as follows:
	\begin{itemize}
		\item If $l=0$ and $k\in \left(0,\frac{\sqrt{2}}{2}\right)$, we get that when $z\in \partial S_1\cap \mathbb{R}$,  $\Re(\Omega)=0$. On the upper boundaries of $S_1$, there exist three points $z=z_{R1}+\ii \frac{K}{2}$, $\frac{K'}{2}+\ii\frac{K}{2}$ and $z_{R2}+\ii \frac{K}{2}$,  $z_{R1}<\frac{K'}{2}<z_{R2}$, satisfying $\Re(\Omega)=0$. On the remaining left and right boundary of $S_1$, there do not exist any points such that $\Re(\Omega)=0$.
		\item If $l=0$ and $k\in \left(\frac{\sqrt{2}}{2},1\right)$, we get that when $z\in \partial S_1\cap \mathbb{R}$,  $\Re(\Omega)=0$. On the upper boundary of $S_1$, i.e., $z\in \partial S_1\cap\left\{z\left|z=z_R+\frac{\ii K}{2}\right.\right\}$, only exist one point $z=\frac{K'}{2}+\ii\frac{K}{2}$ such that $\Re(\Omega)=0$. The remaining left and right boundary of $S_1$ do not exist any points satisfying $\Re(\Omega)=0$.
	\end{itemize}
\end{lemma}

\begin{proof}
	 When $l=0$, the function $\Omega(z)$ could be simplified as 
	\begin{equation}\label{eq:Omega-0}
			\Omega(z)=-\alpha^3 \left( k^2 \dn(\ii z)\sn(\ii z)\cn(\ii z)+k'^2\frac{\dn(\ii z)\sn(\ii z)}{\cn^3(\ii z)}\right)=-\alpha^3 \frac{2\dn(2\ii z)(1-\cn(2\ii z))^2}{\sn^3(2\ii z)},
	\end{equation}
which is obtained by studying the zeros and poles of the elliptic function $\Omega(z)$ as follows. By the zeros and poles of Jacobi elliptic functions \eqref{eq:J 0in} and the particular values of function $\cn(\ii z)$ in \cite{ByrdF-54}, we obtain that all zeros of the middle term of equation \eqref{eq:Omega-0} are $2 nK'+2\ii mK$, $(2n+1) K'+\ii(2m+1)K$, $(n+\frac{1}{2}) K'+\ii(m+\frac{1}{2})K$ with multiplicity one. Similarly, all poles of the middle term of equation \eqref{eq:Omega-0} are $2n K'+\ii(2m+1)K$, $(2n+1) K'+\ii 2mK$ with multiplicity three. Furthermore, we could verify that the right side of equation \eqref{eq:Omega-0} has the same zeros and poles as the middle term. Plugging $z=-\ii \frac{K}{2}$ into $\Omega(z)$ and combining Liouville's theorem, we get equation \eqref{eq:Omega-0}.

And then, we consider the value of $\Re(\Omega(z))$ for $z\in \partial S_1$.
	Since $\dn(2\ii z),\cn(2\ii z)\in \mathbb{R},z\in \mathbb{R}$ and $\sn(2\ii z)\in \ii \mathbb{R},z\in \mathbb{R}$, we get that the equation $\Re(\Omega(z))\equiv 0$ holds on the real line for any $k\in(0,1)$. Consider the left boundary of set $S_1$, i.e., $z\in \ii \mathbb{R}\cap \partial S_1$ and set $z=\ii z_I, z_I\in \left(0,\frac{K}{2}\right)$. Substituting $z=\ii z_I$ into equation \eqref{eq:Omega-0}, we obtain $\Re(\Omega(z))\neq 0$ with $z_I\neq 0$, which implies $\Re(\Omega(z))\neq 0, z\in \ii \mathbb{R}\cap( \partial S_1\backslash\{0\})$. If the parameter $z$ in the right boundary of set $\partial S_1$, we plug $z=K'+\ii z_I$ into equation \eqref{eq:Omega-0} and get
	 \begin{equation}
	 	\begin{split}
	 		\Omega(z)=&-\alpha^3 \frac{2\dn(-2z_I+2\ii K')(1-\cn(-2z_I+2\ii K'))^2}{\sn^3(-2z_I+2\ii K')}
	 		=-\alpha^3 \frac{2\dn(2z_I)(1+\cn(2z_I))^2}{\sn^3(2z_I)},
	 	\end{split}
	 \end{equation}
	which implies that for any $z_I\in\left(0,\frac{K}{2}\right]$, $\Re(\Omega(z))\neq 0$. Then, we study the upper boundary of set $\partial S_1$. 
	Plugging $z=z_R+\ii \frac{K}{2}$ into \eqref{eq:Omega-0} and using the translation formula \eqref{eq:Jacobi-shift} and imaginary arguments formula \eqref{eq:Jacobi-I}, we get
		\begin{equation}
		\begin{split}
			\Re(\Omega(z)) \xlongequal[ ]{\eqref{eq:Jacobi-shift}}
			\Re\left(2\alpha^3k' \frac{(\dn(2\ii z_R)-k'\sn(2\ii z_R))^2}{\cn^3(2\ii z_R)}\right)
			\xlongequal[ ]{\eqref{eq:Jacobi-I}}2\alpha^3k'\cn(2z_R,k')\left(1-2k'^2\sn^2(2 z_R,k')\right).
		\end{split}
	\end{equation}
	Letting $\Re(\Omega(z))=0$, we have $\cn(2z_R,k')=0$ or $1-2k'^2\sn^2(2 z_R,k')=0$. It is easy to verify that if $z_R=\frac{K'}{2}$, equation $\cn(2z_R,k')=0$ holds. Solving $1-2k'^2\sn^2(2 z_R,k')=0$, we know that only if $k\in \left(0,\frac{\sqrt{2}}{2}\right)$, equation $1-2k'^2\sn^2(2 z_R,k')=0$ have two different roots $z_{R1},z_{R2}$ satisfying $z_{R1}<\frac{K'}{2}< z_{R2}< K'$. When $k=\frac{\sqrt{2}}{2}$, equation $1-2k'^2\sn^2(2 z_R,k')=0$ has the multiple root $z_{R1}=z_{R2}=\frac{K'}{2}$. Otherwise, when $k\in \left(\frac{\sqrt{2}}{2},1\right)$, $1-2k'^2\sn^2(2 z_R,k')>0, z_R\in [0,K']$.
\end{proof}

\begin{remark}\label{remark:boundary}
	 Substituting $z=\frac{K'}{2}+\ii z_I$ into equation \eqref{eq:Omega-0}, we obtain
	\begin{equation}\nonumber
			\Omega(z)
			=-\alpha^3 \frac{2\dn(\ii K'-2z_I)(1-\cn(\ii K'-2z_I))^2}{\sn^3(\ii K'-2z_I)}
			=-2\alpha^3k\cn(2z_I)\left( 2k\sn(2z_I)\dn(2z_I)+\ii (2k^2\sn^2(2z_I)-1) \right),
	\end{equation}
	which implies that on the line $\left\{z\left|z=\frac{K'}{2}+\ii z_I\right.\right\}\cap S_1$, only when $z=\frac{K'}{2}, \frac{K'}{2}+\ii \frac{K}{2}$, the equation $\Re(\Omega(z))=0$ holds. 
\end{remark}

For the different values of $l$, we study the function $\Omega(z)$ to obtain the curve $\Re(\Omega(z))\equiv 0$. Based on the properties of elliptic functions and equation \eqref{eq:Omega-0}, it is easy to verify that $\Omega(-z)=-\Omega(z)$ and $\Omega(z^*)=-\Omega^*(z)$ in equation \eqref{eq:Omega-0}, which implies that the function $\Re(\Omega)$ is symmetric about the real axis $\Im(z)=0$ and the line $\Re(z)=0$. Therefore, when $l=0$, we just need to study the region $S_1$. 

\begin{prop}\label{prop:Omega=0}
	The curve $\Re(\Omega(z))\equiv 0$ could be divided into the following cases:
	\begin{itemize}
		\item[(i)] If $l=0$ and $k\in (0,\frac{\sqrt{6}}{4}-\frac{\sqrt{2}}{4})$, there are three line segments in $S_1$ satisfying $\Re(\Omega(z))\equiv 0$. Those three line segments start with points $z=z_{R1}+\ii \frac{K}{2}$, $z=\frac{K'}{2}+\ii \frac{K}{2}$ and  $z=z_{R2}+\ii \frac{K}{2}$, $z_{R1}<\frac{K'}{2}<z_{R2}<K'$, on the upper boundary of $S_1$ and end with points $z=z_1$, $z=z_2$ and $z=K'$, $z_1<z_2<K'$ on the real axis. Furthermore, the above curves do not intersect with each other.
		\item[(ii)] If $l=0$ and $k\in (\frac{\sqrt{6}}{4}-\frac{\sqrt{2}}{4},\frac{\sqrt{2}}{2})$, there are two line segments in $S_1$ satisfying $\Re(\Omega(z))\equiv 0$. One of those line segments starts with $z=z_{R1}+\ii \frac{K}{2}$ and ends with  $z=\frac{K'}{2}+\ii \frac{K}{2}$.
		The other curve starts with $z=z_{R2}+\ii \frac{K}{2}$ and ends with $z=K'$. Furthermore, the above curves do not intersect with each other.
		\item[(iii)] If $l=0$ and $k\in (\frac{\sqrt{2}}{2},1)$, there is only one curve in $S_1$ meeting $\Re(\Omega(z))\equiv 0$, which connects points $\frac{K'}{2}+\ii \frac{K}{2}$ and $K'$.
	\end{itemize}
\end{prop}

\begin{proof}
	By the definition of $\Omega(z)$ in equation \eqref{eq:Omega}, it is easy to obtain that function $\Omega(z)$ is a meromorphic function on the whole complex plane. Considering the derivative of $\Omega(z)$, we get the curve $\Re(\Omega(z))=0$ generated by the tangent vector 
	\begin{equation}\label{eq:Omega-tangent}
		\left(-\frac{\dd \Re(\Omega(z))}{\dd z_I}, \frac{\dd \Re(\Omega(z))}{\dd z_R}\right)=\left(\Im \left(\frac{\dd \Omega(z)}{\dd z}\right), \Re \left(\frac{\dd \Omega(z)}{\dd z}\right)\right).
	\end{equation}
	Based on the above iteration, the curve $\Re(\Omega(z))\equiv 0$ could be divided into the following two categories:
	\begin{itemize}
		\item[(1)] The curve $\Re(\Omega(z))\equiv 0$ forms a closed circle. According to the maximum value principle of harmonic function, if there is no singularity in the area, all the values of $\Omega(z)$ in the closed area satisfy $\Re(\Omega(z))\equiv 0$. Otherwise, there must exist at least one pole in this closed region.
		\item[(2)] Without any closed loops, the curve must end up at the point $z_0$ satisfying $\Omega'(z_0)=\infty$ or $\Omega'(z_0)=0$ or at the boundary $\partial S_1$. It should be noticed that when there exist two different curves end at the same point $z_0$ satisfying $\Omega'(z_0)=0$, the Taylor expanding at this point could be written as $\Omega(z)=\Omega(z_0)+\frac{\Omega''(z_0)}{2}(z-z_0)^2+\mathcal{O}\left((z-z_0)^3\right)$, with $\Omega''(z_0)\neq 0$.
	\end{itemize}

	The derivatives of function $\Omega(z)$ in equation \eqref{eq:Omega-0} could be written as
\begin{equation}\label{eq:Omega-d-0}
		\begin{split}
			\Omega'(z)
			=&-\alpha^3 \left(\frac{2\dn(2\ii z)(1-\cn(2\ii z))^2}{\sn^3(2\ii z)}\right)'
			=-\frac{4\ii \alpha^3(\cn(2\ii z)-1)^2(2k^2\cn^2(2\ii z)-\cn(2\ii z)-2k^2+2)}{\sn^4(2\ii z)}.
		\end{split}
\end{equation}
By function $\Omega(z)$ in \eqref{eq:Omega-0}, we know that the periods of $\Omega(z)$ and $\Omega'(z)$ are both $2\ii K$ and $\ii K+K'$. By the zeros and poles of the Jacobi elliptic functions in equation \eqref{eq:J 0in}, we know that only when $z=(2m_1+1)K'+ (2n_1+1)\ii K$ and $2n_2K'+(2m_2+1)\ii K$, $n_1,m_1,n_2,m_2\in \mathbb{Z}$ the derivative function reaches infinity, i.e. $\Omega'(z)=\infty$, and both of them are four order poles. Consider a period parallelogram starting from $\left(0,0\right)$ and taking $\left(0,0\right)$, $\left(0,-2\ii K\right)$, $\left(K',-\ii K\right)$, $\left(K',\ii K\right)$ as vertices. Based on the above studies of the poles of $\Omega'(z)$, we claim that there are four poles in this region including multiple numbers. 

And then, we consider the zeros of $\Omega'(z)$ in this period parallelogram. It is easy to know that only if the numerator of $\Omega'(z)$ in \eqref{eq:Omega-d-0} is zero, the derivative function $\Omega'(z)$ is zero, which reflects that if and only if $\cn(2\ii z)-1=0$ or $2k^2\cn^2(2\ii z)-\cn(2\ii z)-2k^2+2=0$, the equation $\Omega'(z)=0$ holds. Considering the equation $f(z):=2k^2\cn^2(2\ii z)-\cn(2\ii z)-2k^2+2$, we know that if $16k^4-16k^2+1>0$, the equation $f(z)=0$ has two real roots $z_1<z_2$ satisfying $1<\cn(2\ii z_1)<\cn(2\ii z_2)<\infty$. Because $k\in (0,1)$ and the value of function $\cn(2\ii z)=1/\cn(2z,k')$, $z\in \mathbb{R}$, we obtain that if $k\in(0,\frac{\sqrt{6}-\sqrt{2}}{4})$, the above two real roots satisfy $0<z_1<z_2<\frac{K'}{2}$. If $k=\frac{\sqrt{6}-\sqrt{2}}{4}$, on the real axis, there only exists a two-order root $z_1$ such that $f(z_1)=0$. For $\cn(2\ii z)-1=0$, we know that $z=0,-2\ii K,K'-\ii K,K'+\ii K$ are four two-order zero points of function $\Omega'(z)$. Since the above four points are the vertices of a period parallelogram and $z_i$, $i=1,2$ are in the period parallelogram, we claim that we obtain four zero points including multiples in period parallelogram, when $k\in (0,\frac{\sqrt{6}-\sqrt{2}}{4})$. Because in a periodic parallelogram, the elliptic function $\Omega'(z)$ has the same number of zeros and poles, including the multiple points, we could claim that all zero points are gained when $k\in \left( 0,\frac{\sqrt{6}}{4}-\frac{\sqrt{2}}{4} \right)$. If $k\in \left( \frac{\sqrt{6}-\sqrt{2}}{4},1 \right)$, $f(z)$ does not exist any real roots $z$ such that $f(z)=0$.
	
	Combining with Lemma \ref{lemma:S-boundary} and Remark \ref{remark:boundary}, we divide modules $k$ into the following three cases in region $S_1$.
	\begin{itemize}
		\item[(i)] When $k\in (0,\frac{\sqrt{6}}{4}-\frac{\sqrt{2}}{4})$, we consider three points $z=z_{R1}+\ii \frac{K}{2}$, $z=\frac{K'}{2}+\ii \frac{K}{2}$ and  $z=z_{R2}+\ii \frac{K}{2}$ in the upper region of $\partial S_1$. Along the tangent vector, we want to show that they would end at points $z=z_1$, $z=z_2$ and $z=K'$. If not, there must exist two curves intersecting with others at point $z_0\in S_1$ such that $\Omega'(z_0)=0$, which contradicts the fact that all zero points could be written as $z=z_1+m_1K'+2n_1\ii K$ or $z=z_2+m_2K'+2n_2\ii K$, $n_1,n_2,m_1,m_2\in \mathbb{Z}$. Thus, we obtain that except for line $\mathbb{R}\cap \partial S_1$, three curves exist in region $S_1$ satisfying $\Re(\Omega(z))=0$. They start at points $z=z_{R1}+\ii \frac{K}{2}$, $z=\frac{K'}{2}+\ii \frac{K}{2}$ and  $z=z_{R2}+\ii \frac{K}{2}$, along the tangent vector \eqref{eq:Omega-tangent} and end at points $z=z_1$, $z=z_2$ and $z=K'$. 
		\item[(ii)] When $k\in (\frac{\sqrt{6}}{4}-\frac{\sqrt{2}}{4},\frac{\sqrt{2}}{2})$, we consider three points $z=z_{R1}+\ii \frac{K}{2}$, $z=\frac{K'}{2}+\ii \frac{K}{2}$ and  $z=z_{R2}+\ii \frac{K}{2}$ in the upper region of $\partial S_1$. Firstly, we consider the point $z=z_{R1}+\ii \frac{K}{2}$. From Remark \ref{remark:boundary}, we know that this curve does not cross the line $\Re(z)=\frac{K'}{2}$. Furthermore, on the real line there do not exist any points such that $\Omega'(z)=0$. Thus, it must end up at the boundary point $z=\frac{K'}{2}+\ii \frac{K}{2}$, by Lemma \ref{lemma:S-boundary}. Similarly, the point $z=z_{R2}+\ii \frac{K}{2}$ must end up at point $z=K'$. Therefore, except for the boundary of $S_1$, there exist two curves in the region $S_1$. They start at the points $\left(z_{R2},\frac{K}{2}\right)$, $\left(\frac{K'}{2},\frac{K}{2}\right)$ and end at the points $\left(z_{R1},\frac{K}{2}\right)$ and $\left(K',0\right)$. Since there is not any point on the line $z=\frac{K'}{2}+\ii z_I, z_I\in \left(0,\frac{K}{2}\right)$ satisfying $\Re(\Omega(z))=0$, the above two curves do not intersect with each other. The difference from the above condition is that there do not exist  points on line $z\in \mathbb{R}$ such that $\Omega'(z)=0$.
		
		\item[(iii)] When $k\in (\frac{\sqrt{2}}{2},1)$, only one curve exists in region $S_1$. It will start at the point $z=\frac{K'}{2}+\ii \frac{K}{2}$ and end at point $z=K'$. Since on the upper boundary of $\partial S_1$, there only exists one point satisfying $\Re(\Omega(z))=0$ and on the line $z\in \mathbb{R}$, $\Omega'(z)\neq 0$.
	\end{itemize}
\end{proof}

Then, we consider the velocity $v(z)=\frac{\Re(\Omega(z))}{\Re(I(z))}$. In our previous studies \cite{LinglmS-21}, we have proved that in the region $S_1$, there is only one curve satisfying $\Re(I(z))=0$. When $k\in (0,0.9089)$ this curve intersects with the real axis and when $k\in (0.9089,1)$ this curve intersects with the imaginary axis. We also prove that the curves $\Re(I(z))=0$ and $\Re(\Omega)=0$ do not intersect with each other. Moreover, only when $z=\pm K'+l$, $\Re(I(z))=\infty$ and $\Re(\Omega(z))=\infty$. Together with the curves satisfying $\Re(\Omega(z))=0$, $\Re(I(z))=0$ and the points satisfying $\Omega(z)=\infty$ and $I(z)=\infty$, we could obtain that the region of velocity $v$ can be classified into four different cases of parameter $k$ in the different regions.
	\begin{itemize}
		\item When $k\in (0,\frac{\sqrt{6}-\sqrt{2}}{4})$, the value of $v=v(z)$ in region $S_1$ is divided into five areas. The above regions are separated by three curves $\Re(\Omega(z))=0$ satisfying the case (i) of Proposition \ref{prop:Omega=0} and one curve $\Re(I(z))=0$ intersecting with a real line (See in Figure \ref{fig:Omega-v-0}(i));
		\item When $k\in (\frac{\sqrt{6}-\sqrt{2}}{4}, \frac{\sqrt{2}}{2})$, the value of $v=v(z),z\in S_1$ is divided into four areas. The above regions are separated by two curves $\Re(\Omega(z))=0$ satisfying the case (ii) of Proposition \ref{prop:Omega=0} and one curve $\Re(I(z))=0$ intersecting with a real line (See in Figure \ref{fig:Omega-v-0}(ii));
		\item When $k\in (\frac{\sqrt{2}}{2},0.9089)$, the value of $v=v(z),z\in S_1$ is divided into three regions, which are separated by a curve $\Re(\Omega(z))=0$ satisfying the case (iii) of Proposition \ref{prop:Omega=0} and a curve $\Re(I(z))=0$ intersecting with a real line (See in Figure \ref{fig:Omega-v-0}(iii));
		\item When $k\in (0.9089, 1)$, the value of $v=v(z),z\in S_1$ is also divided into three regions. However, unlike the above, the division of regions is different. The above regions are separated by a curve $\Re(\Omega(z))=0$ satisfying the case (iii) of Proposition \ref{prop:Omega=0} and a curve $\Re(I(z))=0$ intersecting with imaginary axis (See in Figure \ref{fig:Omega-v-0}(iv)).
\end{itemize}
For the above four cases, we give corresponding values to depict them clearly in Figure \ref{fig:Omega-v-0}. We choose the modulus $k$ as $0.25,0.35,0.85,0.95$, which satisfy the above four cases, respectively. The first four figures (i), (ii), (iii), (iv) show the velocity $v$ in the $z$-plane and the rest four figures (a), (b), (c), (d) describe the velocity $v$ in the $\lambda$-plane with cuts in yellow. The green regions in them represent the condition $v>0$, and the white ones represent $v<0$. In Figure \ref{fig:Omega-v-0}, we could find that the above regions with different signs of velocity $v$ are separated by curves $\Re(\Omega(z))=0$ in red and curves $\Re(I(z))=0$ in blue. It is worth noting that we do not use the same method as the $z$-plane above to obtain the variation of $v$ in the $\lambda$-plane. We get it based on the conformal map between the $\lambda$-plane and the $z$-plane.

\begin{figure}[h]
	\centering
	\includegraphics[width=1\textwidth]{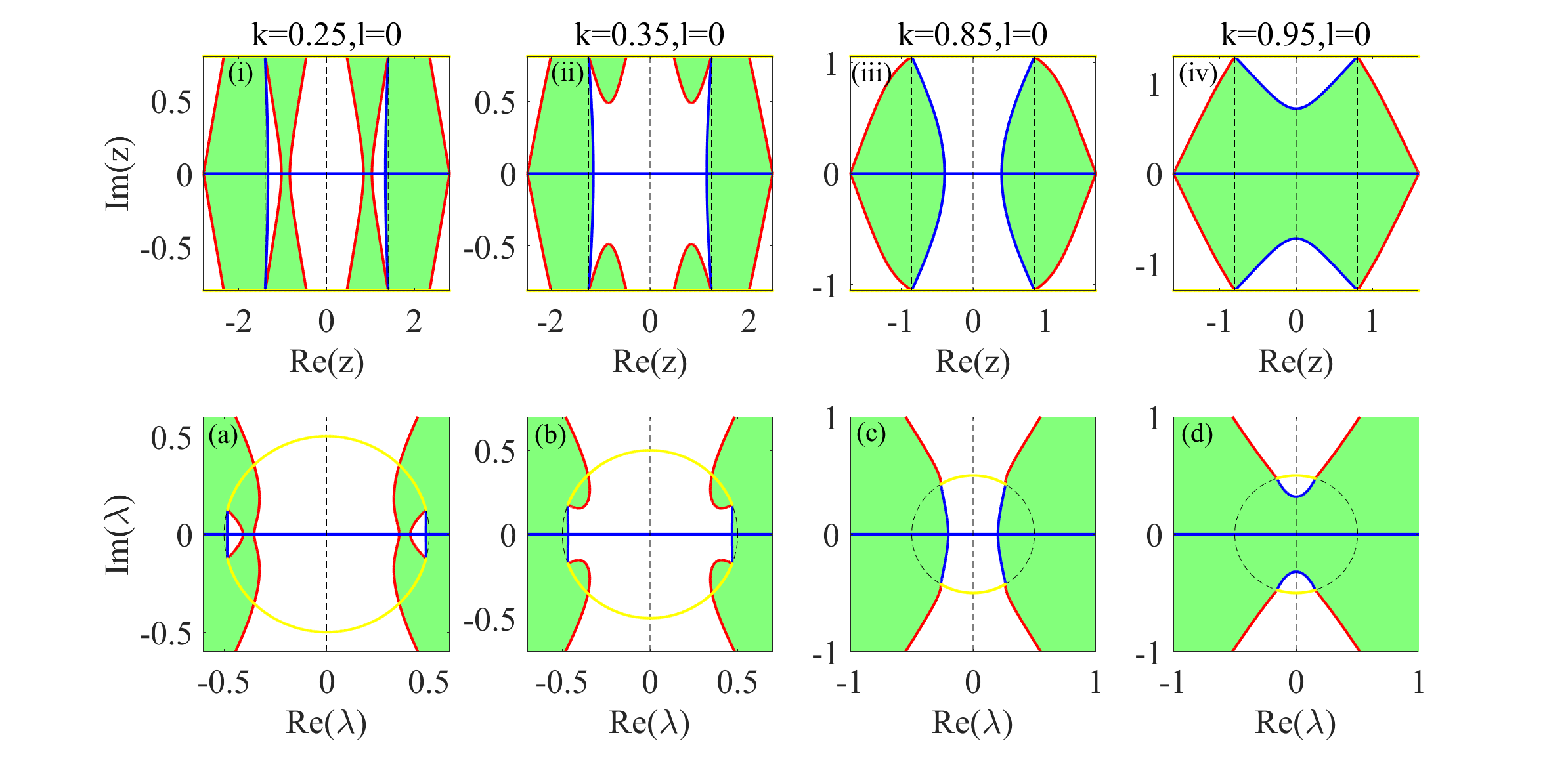}
	\caption{The first four images (i), (ii), (iii), (iv) draws $v>0$ in green and $v<0$ in white on the $z$-plane. The red curves describe the condition $\Re(\Omega(z))=0$, and the blue curves describe the condition $\Re(I(z))=0$. The rest four images (a), (b), (c), (d) draws $v>0$ in green and $v<0$ in white on the $\lambda$-plane. The red curves describe the condition $\Re(\Omega)=0$ and the blue curves describe the condition $\Re(I)=0$. Figures (a), (b), (c) and (d) are obtained from the $z$-plane through the conformal map $\lambda(z)$.}
	\label{fig:Omega-v-0}
\end{figure}

\begin{remark}	
	Consider the condition $l=\frac{K'}{2}$. The function $\Omega(z)$ could be written as 
	\begin{equation}\label{eq:Omega-l-K1}
		\begin{split}
			\Omega(z)=&-\alpha^3 \left( k^2 \dn(\ii(z-l))\sn(\ii(z-l))\cn(\ii(z-l))+k'^2\frac{\dn(\ii(z+l))\sn(\ii(z+l))}{\cn^3(\ii(z+l))}\right).
		\end{split}
	\end{equation}
	From work \cite{LinglmS-21}, we know that when $z\in \mathbb{R}$ and $z=z_R\pm \ii \frac{K}{2}$, the equation $\Re(\Omega(z))=0$ holds. As previously mentioned, through studying the derivative of function $\Omega(z)$, we obtain that there only exist points $z=K'+\ii \frac{K}{2}$ and $z=\ii \frac{K}{2}$ such that $\Omega'(z)=0$. Therefore, we get that excepting the above two lines, there are only two curves in the upper half plane of $S$,  which connect the points $z=K'+\ii \frac{K}{2}$, $z=\frac{3K'}{2}$ and the points $z=\ii \frac{K}{2}$, $z=-\frac{K'}{2}$.
\end{remark}
We also describe the case $l=\frac{K'}{2}$ with $k=\frac{9}{10}$ in Figure \ref{fig:Omega-v-K}. The Figure \ref{fig:Omega-v-K}(i) shows the velocity $v$ in the $z$-plane and Figure \ref{fig:Omega-v-K}(a) 
describes the velocity $v$ in the $\lambda$-plane with cuts in yellow. The green areas represent the condition $v>0$, and the white represents $v<0$. We could find that the above regions with the different signs of velocity $v$ are separated by curves satisfying $\Re(\Omega(z))=0$ in red. The result of Figure \ref{fig:Omega-v-K}(a) in the $\lambda$-plane depends on the conformal map between the $\lambda$-plane and the $z$-plane. The blue points and crosses in Figures \ref{fig:Omega-v-K}(i) and \ref{fig:Omega-v-K}(a) are the parameters $z_1$ and $\lambda_1$ we used in drawing Figure \ref{fig:dn-1b-3d}. Points 2-(i), 2-(ii), 2-(iii) represent the parameters we used in Figure \ref{fig:dn-1b-3d}(i), Figure \ref{fig:dn-1b-3d}(ii), and Figure \ref{fig:dn-1b-3d}(iii), respectively.

Looking back at the variation between $(\xi,\tau)$ and $(x,t)$, we know that the above velocity $v$ describes the variation between $\xi$ and $\tau$. Combining the rotation $\xi=x-st$, $\tau=t$, the line $L_i$ in equation \eqref{eq:li} could be converted into $x-(s+v(z_i))t+\hat{c_i}$. 
\begin{figure}[h]
	\centering
	\includegraphics[width=0.8\textwidth]{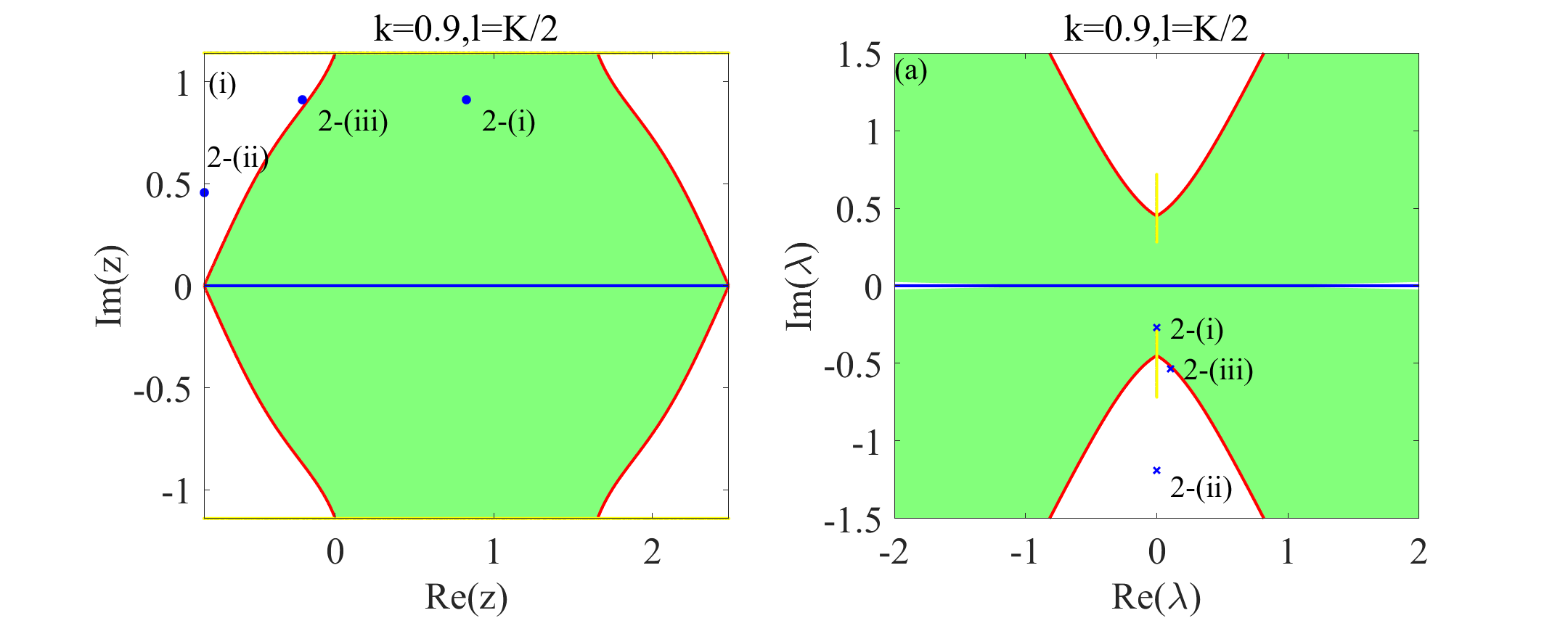}
	\caption{The curve $\Re(\Omega(z))$ draws the case of $l=\frac{K'}{2},k=0.90$. The white region represents $\Re(\Omega(z))<0$ and the green region represents $\Re(\Omega(z))>0$. Points 2-(i), 2-(ii), 2-(iii) represent the parameters we used in Figure \ref{fig:dn-1b-3d}(i), Figure \ref{fig:dn-1b-3d}(ii), and Figure \ref{fig:dn-1b-3d}(iii), respectively.}
	\label{fig:Omega-v-K}
\end{figure}

\subsection{Asymptotic analysis for the multi elliptic-localized solutions}

Assume that the multi elliptic-localized solution $u^{[N]}(x,t)$ has $N$ different velocities. However, the matrices $\mathcal{M}$ and $\mathcal{D}$ of solution $u^{[N]}(x,t)$ in \eqref{eq:mKdV-solution-xi-n-1} are $m\times m$, $N\le m\le 2N$, which implies that there must exist different uniform parameters $z_i\neq z_{j}$ with the same velocity $v(z_i)=v(z_{j})$. The main reason for the above conditions lies in the different kinds of the Darboux matrix $\mathbf{T}_i^{\mathrm{C}}(\lambda;x,t)$ and $\mathbf{T}_i^{\mathrm{P}}(\lambda;x,t)$ we choose. 

Then, we study the asymptotic expression of solution $\hat{u}^{[N]}(\xi,\tau)$ alone the line $L_k^{\pm}$, $k=1,2,\cdots,N$ shown in Figure \ref{fig:li}. Those $2N$ lines divided the $(\xi,\tau)$ plane into $2N$ different pieces of the region $R_i^{\pm}, i=1,2,\cdots,N$. The sketch map of the above conditions is shown in Figure \ref{fig:li}. 
We should notice that the direction depends only on the number $N$ and not the dimension $m$. 

\newenvironment{aproof-asy}{\emph{Proof of Theorem \ref{theorem:exact-N-solution}.}}{\hfill$\Box$\medskip}
\begin{aproof-asy}
By the definition of $\eta_i$, $i=1,2,\cdots ,m$, in equation \eqref{eq:E-2-hat}, it is easy to obtain
	\begin{equation}
		\begin{split}
			\exp\left(\eta_i\right)=\exp\left(\gamma_i+ I_i\xi-\Omega_i \tau\right)
			=\exp\left(\gamma_i+\ii\Im(I_i)\xi-\ii\Im(\Omega)\tau+\Re(I_i)(\xi-v_h\tau)+\Re(I_i)(v_h-v_i)\tau\right).
		\end{split}
\end{equation}
Without losing the generality, we can assume $v_i=v(z_i)\le v_{i+1}=v(z_{i+1}), i=1,2,\cdots,m-1$ and $I_i=I(z_i)>0, i=1,2,\cdots,m$. Along the line $L_q^{\pm}$ with $\lambda_q\in \ii \mathbb{R}$, we obtain that only one parameter $z_h$ such that $\eta_h=\text{const}$ in \eqref{eq:E-2-hat}. The relations between parameters $h$ and $q$ are defined in equation \eqref{eq:defie-q}. Then, we obtain that as $\tau\rightarrow +\infty$, 
the value $\eta_i\rightarrow -\infty$, $\ee^{\eta_i}=\mathcal{O}(\ee^{\Re(I_i)(v_h-v_i)\tau})$, $i=h+1,h+2,\cdots, m$, and $-\eta_i\rightarrow -\infty$, $\ee^{-\eta_i}=\mathcal{O}(\ee^{-\Re(I_i)(v_h-v_i)\tau})$, $i=1,2,\cdots, h-1$. Similarly, as $\tau\rightarrow -\infty$, the value $-\eta_i\rightarrow -\infty$, $\ee^{-\eta_i}=\mathcal{O}(\ee^{-\Re(I_i)(v_h-v_i)\tau})$, $i=h+1,h+2,\cdots, m$ and $\eta_i\rightarrow -\infty$, $\ee^{\eta_i}=\mathcal{O}(\ee^{\Re(I_i)(v_h-v_i)\tau})$, $i=1,2,\cdots, h-1$. Thus, combining equations \eqref{eq:X1-X2} and \eqref{eq:X1-X2-2}, we prove the first case in Theorem \ref{theorem:exact-N-solution}.

Considering the second case that there are two parameters $z_h$ and $z_{h+1}$ such that $\eta_h=\eta_{h+1}=\text{const}$, along the line $L_q^{\pm}$ with $\lambda_q\in \mathbb{C}\backslash(\ii \mathbb{R}\cup \mathbb{R})$, which occurs at the point $\mathbf{Z}_q=\left[z_q,2l-z_q^*\right]$. The difference from the previous case is that in addition to $\eta_h$, the parameter $ \eta_{h+1}=\text{const}$. And then, we get that as $\tau\rightarrow +\infty$, $\eta_i\rightarrow -\infty$, $\ee^{\eta_i}=\mathcal{O}(\ee^{\Re(I_i)(v_h-v_i)\tau})$, $i=h+2,\cdots, m$, $-\eta_i\rightarrow -\infty$, $\ee^{-\eta_i}=\mathcal{O}(\ee^{-\Re(I_i)(v_h-v_i)\tau})$, $i=1,2,\cdots, h-1$ and as $\tau\rightarrow -\infty$, $-\eta_i\rightarrow -\infty$, $\ee^{-\eta_i}=\mathcal{O}(\ee^{-\Re(I_i)(v_h-v_i)\tau})$, $i=h+2,\cdots, m$, $\eta_i\rightarrow -\infty$, $\ee^{\eta_i}=\mathcal{O}(\ee^{\Re(I_i)(v_h-v_i)\tau})$, $i=1,2,\cdots, h-1$. Plugging them into equation \eqref{eq:mKdV-solution-xi-n-1-vec}, we prove case (ii) in Theorem \ref{theorem:exact-N-solution}.
\end{aproof-asy}

Consider case (i) in Theorem \ref{theorem:exact-N-solution}. For $\tau\rightarrow +\infty$ along the line $L_q$ with $\lambda_q\in \ii \mathbb{R}$ and only one parameter satisfying $\eta_h=\text{const}$, equation \eqref{eq:mKdV-solution-xi-n-k-vec} could be rewritten as the following forms:
\begin{equation}\label{eq:solution-u-n-infy+}
	\begin{split}
		\hat{u}^{[N]}( \xi,\tau;L_q^+)\rightarrow&\frac{\alpha \vartheta_2\vartheta_4 }{\vartheta_3\vartheta_3(\frac{2\ii l}{2K})}\left(\frac{\vartheta_4(\frac{\alpha\xi}{2K})}{\vartheta_2(\frac{\alpha\xi+2\ii l}{2K})}\right)^{m-1}\mathbf{r}_{h,h}^{+}\\
		&
		\frac{\sum_{a,b=1}^{2}r_h^{2(a-1)*}r_h^{-2(b-1)}(\ii)^{a+b}\det( \mathcal{M}^{[a,b]}_{+})\ee^{(a-1)(\eta_h-\frac{2l}{2K}\pi)+(b-1)(\eta_h^*+\frac{2l}{2K}\pi)}}{\sum_{a,b=1}^{2} (-1)^{b}\det( \mathcal{D}^{[a,b]}_{+})\ee^{(a-1)\eta_h+(b-1)\eta_h^*}},
	\end{split}
\end{equation}
where $r_i,i=1,2,\cdots,m$ are defined in equation \eqref{eq:r_i} and
\begin{equation}\label{eq:define-r-i-j}
	\mathbf{r}^{+}_{i,j}=\prod_{k=1}^{i}\frac{r_k^*}{r_k}\prod_{k=j+1}^{m}\frac{r_k}{r_k^*}, \qquad  \mathbf{r}^{-}_{i,j}=\left(\mathbf{r}^{+}_{i,j}\right)^{-1},
\end{equation}
\begin{equation}
	\mathcal{D}^{[a,b]}_{+}=\left(\frac{\vartheta_4\left(\frac{\alpha\xi+\zeta^{[a,+]}_i+\mu^{[b,+]}_j}{2K}\right)}{\vartheta_1\left(\frac{\zeta^{[a,+]}_i+\mu^{[b,+]}_j}{2K}\right)}\right)_{1\le i,j\le m}, \qquad \mathcal{M}^{[a,b]}_{+}=\left(\frac{\vartheta_2\left(\frac{\alpha\xi+2\ii l+\zeta^{[a,+]}_i+\mu^{[b,+]}_j}{2K}\right)}{\vartheta_1\left(\frac{\zeta^{[a,+]}_i+\mu^{[b,+]}_j}{2K}\right)}\right)_{1\le i,j\le m}.
\end{equation}
Furthermore, we list the concrete sequence of $\zeta^{[a,+]}_j$ and $\mu^{[a,+]}_j$, $a=1,2$ :
\begin{equation}\nonumber
	\begin{split}
		\zeta^{[1,+]}
		=&\left(\zeta^{[1,+]}_1,\zeta^{[1,+]}_2,\cdots,\zeta^{[1,+]}_{h-1},\zeta^{[1,+]}_{h},\zeta^{[1,+]}_{h+1},\cdots,\zeta^{[1,+]}_m\right)\\	
		=&\left(-\ii z_1^*-(K+\ii K'), -\ii z_2^*-(K+\ii K'), \cdots, -\ii z_{h-1}^*-(K+\ii K'), \ii z_h^*, \ii z_{h+1}^*, \cdots, \ii z_m^*\right)\\
		\zeta^{[2,+]}=
		&\left(\zeta^{[2,+]}_1,\zeta^{[2,+]}_2,\cdots,\zeta^{[2,+]}_{h-1},\zeta^{[2,+]}_{h},\zeta^{[2,+]}_{h+1},\cdots,\zeta^{[2,+]}_m\right)\\
		=&\left(-\ii z_1^*-(K+\ii K'), -\ii z_2^*-(K+\ii K'), \cdots, -\ii z_{h-1}^*-(K+\ii K'), -\ii z_h^*-(K+\ii K'), \ii z_{h+1}^*, \cdots, \ii z_m^*\right)\\
		\mu^{[1,+]}
		=&\left(\mu^{[1,+]}_1,	\mu^{[1,+]}_2,\cdots,	\mu^{[1,+]}_{h-1},	\mu^{[1,+]}_{h},	\mu^{[1,+]}_{h+1},\cdots,	\mu^{[1,+]}_m\right)\\
		=&\left(\ii z_1+K+\ii K', \ii z_2+K+\ii K', \cdots, \ii z_{h-1}+K+\ii K', -\ii z_h, -\ii z_{h+1}, \cdots, -\ii z_m\right)\\
		\mu^{[2,+]}
		=&\left(\mu^{[2,+]}_1,\mu^{[2,+]}_2,\cdots,\mu^{[2,+]}_{h-1},\mu^{[2,+]}_{h},\mu^{[2,+]}_{h+1},\cdots,\mu^{[2,+]}_m\right)\\
		=&\left(\ii z_1+K+\ii K', \ii z_2+K+\ii K', \cdots, \ii z_{h-1}+K+\ii K', \ii z_h+K+\ii K', -\ii z_{h+1}, \cdots, -\ii z_m\right)\\
	\end{split}
\end{equation}

\begin{remark}\label{remark:mKdV-solution-asy-v-simp-i}
	Considering the case (i) of the asymptotic analysis of function $\hat{u}^{[N]}( \xi,\tau;L_q^{\pm})$ in Theorem \ref{theorem:exact-N-solution} and combining the formulas of Jacobi theta functions and the result in Appendix of \cite{Takahashi-16}, we could rewrite the solution $\hat{u}^{[N]}(\xi,\tau;L_q^{\pm})$ as $\tau\rightarrow +\infty$ in equation \eqref{eq:solution-u-n-infy+} as follows
	\begin{equation}\label{eq:solution-asy-v-simp-i}
		\hat{u}^{[N]}(\xi,\tau;L_q^{\pm}) =\frac{\alpha \vartheta_2\vartheta_4 }{\vartheta_3\vartheta_3(\frac{2\ii l}{2K})}\mathbf{r}_{h-1,h}^{\pm}
		\frac{\hat{\mathcal{M}}_h^{\pm \infty} }{\hat{\mathcal{D}}_h^{\pm \infty}}+\mathcal{O}\left(\exp\left(-\min_{j\neq h}
		\Re(I_j)|v_h-v_j||\tau|\right)\right), \qquad \tau
		\rightarrow \pm \infty,
	\end{equation}
where $\mathbf{r}_{h-1,h}^{\pm}$ are defined in equation \eqref{eq:define-r-i-j}; the relations between $q$ and $h$ are defined in equation \eqref{eq:defie-q}; and
	\begin{equation}
		\begin{split}
			\hat{\mathcal{M}}_h^{\pm \infty}=&\begin{bmatrix}
				1 & \Delta_h^{\pm *}\ee^{\eta_h^*+\frac{2l\pi}{2K}}
			\end{bmatrix}\mathbf{r}^*_h
			\begin{bmatrix}
				-\frac{\vartheta_2(\frac{\alpha \xi+2\ii l +s^{\pm}_{h,h}+\ii z_h^*-\ii z_h }{2K})}{\vartheta_1(\frac{\ii z_h^*-\ii z_h}{2K})} 
				& -\ii \frac{\vartheta_2(\frac{\alpha \xi+2\ii l +s^{\pm}_{h,h}+\ii z_h^*+\ii z_h+K+\ii K' }{2K})}{\vartheta_1(\frac{\ii z_h^*+\ii z_h+K+\ii K'}{2K})}
				\\ -\ii \frac{\vartheta_2(\frac{\alpha \xi+2\ii l +s^{\pm}_{h,h}-\ii z_h^*-\ii z_h-K-\ii K' }{2K})}{\vartheta_1(\frac{-\ii z_h^*-\ii z_h-K-\ii K'}{2K})}
				& \frac{\vartheta_2(\frac{\alpha \xi+2\ii l +s^{\pm}_{h,h}-\ii z_h^*+\ii z_h }{2K})}{\vartheta_1(\frac{-\ii z_k^*+\ii z_k}{2K})}
			\end{bmatrix}\mathbf{r}_h^{-1}
			\begin{bmatrix}
				1 \\ \ee^{\eta_h-\frac{2l\pi}{2K}}\Delta_h^{\pm }
			\end{bmatrix},\\
			\hat{\mathcal{D}}_h^{\pm \infty}=&\begin{bmatrix}
				1 & \Delta_h^{\pm *}\ee^{\eta_h^*}
			\end{bmatrix}
			\begin{bmatrix}
				-\frac{\vartheta_4(\frac{\alpha \xi+s^{\pm}_{h,h}+\ii z_h^*-\ii z_h }{2K})}{\vartheta_1(\frac{\ii z_h^*-\ii z_h}{2K})} 
				& \frac{\vartheta_4(\frac{\alpha \xi +s^{\pm}_{h,h}+\ii z_h^*+\ii z_h+K+\ii K' }{2K})}{\vartheta_1(\frac{\ii z_h^*+\ii z_h+K+\ii K'}{2K})}
				\\ -\frac{\vartheta_4(\frac{\alpha \xi +s^{\pm}_{h,h}-\ii z_h^*-\ii z_h-K-\ii K' }{2K})}{\vartheta_1(\frac{-\ii z_h^*-\ii z_h-K-\ii K'}{2K})}
				& \frac{\vartheta_4(\frac{\alpha \xi+s^{\pm}_{h,h}-\ii z_h^*+\ii z_h }{2K})}{\vartheta_1(\frac{-\ii z_h^*+\ii z_h}{2K})}
			\end{bmatrix}
			\begin{bmatrix}
				1 \\ \ee^{\eta_h}\Delta_h^{\pm}
			\end{bmatrix},
		\end{split}
	\end{equation}
with $\mathbf{r}_h$ and $s^{\pm}_{h,h}$ defined in equations \eqref{eq:r_i} and \eqref{eq:s-i-j} respectively.
	Furthermore,
	\begin{equation}\label{eq:r,s}
		\begin{split}
			\Delta_h^{+} =\prod_{j=1}^{h-1}\frac{\vartheta_1(\frac{\ii z_j-\ii z_h}{2K})\vartheta_1(\frac{\ii z_j^*+\ii z_h+K+\ii K'}{2K})}{\vartheta_1(\frac{\ii z_j^*-\ii z_h}{2K})\vartheta_1(\frac{\ii z_j+\ii z_h+K+\ii K'}{2K})}\prod_{j=h+1}^{n}\frac{\vartheta_1(\frac{\ii z_j^*-\ii z_h}{2K})\vartheta_1(\frac{\ii z_j+\ii z_h+K+\ii K'}{2K})}{\vartheta_1(\frac{\ii z_j-\ii z_h}{2K})\vartheta_1(\frac{\ii z_j^*+\ii z_h+K+\ii K'}{2K})}, \quad \Delta_h^{-}=(\Delta_h^{+})^{-1}.
		\end{split}
	\end{equation}
\end{remark}

The above remark actually gives the rigorous proof of the elastic interaction of multi elliptic-localized waves since the elliptic-breathers/solitons have the consistent expressions before and after the interaction. Similarly, we can reduce the formulas in case (ii) of Theorem \ref{theorem:exact-N-solution} to verify their elastic interaction. Here we do not repeat to give complicated formulas. 

\begin{lemma}\label{lemma:asy-R}
	The asymptotically analysis on the region  $R_{q}^{\pm},q=1,2,\cdots,n$ could also be classified into the following two cases: 
	\begin{itemize}
		\item[(i)] If along the line $L_{q-1}$ with $\lambda_{q-1}\in \ii \mathbb{R}$,
		 as $\tau\rightarrow \pm \infty$, 
		 the asymptotic expression on the region $R_{q}^{\pm}$ could be written as
		\begin{equation}\label{eq:solution-u-n-infy-R}
			\hat{u}^{[N]}( \xi,\tau;R_{q}^{\pm})\longrightarrow\frac{\alpha \vartheta_2\vartheta_4 }{A\vartheta_3\vartheta_3(\frac{2\ii l}{2K})}\mathbf{r}_{h,h}^{\pm}
			\frac{\vartheta_2(\frac{\alpha \xi+2\ii l +s_{h,h}^{\pm}}{2K}) }{\vartheta_4(\frac{\alpha \xi +s_{h,h}^{\pm}}{2K})}, \qquad \tau\rightarrow \pm \infty,
		\end{equation}  
		where the relations between parameters $q$ and $h$ are defined in \eqref{eq:defie-q} and $s_{h,h}^{\pm}$ and $\mathbf{r}_{h,h}^{\pm}$ are defined in equations \eqref{eq:s-i-j} and \eqref{eq:define-r-i-j} respectively. 
			\item[(ii)] If along the line $L_{q-1}$ with $\lambda_{q-1}\in \mathbb{C}\backslash(\ii \mathbb{R}\cup\mathbb{R})$, 
		as $\tau\rightarrow \pm \infty$, the asymptotic expression on the region $R_{q}^{\pm}$ could be written in \eqref{eq:solution-u-n-infy-R} by replacing parameters $s_{h,h}^{\pm}$ and $\mathbf{r}_{h,h}^{\pm}$ with  $s_{h,h+1}^{\pm}$ and $\mathbf{r}_{h+1,h+1}^{\pm}$.
	\end{itemize}
\end{lemma}

\begin{proof}
	
To consider the asymptotic analysis of function $\hat{u}^{[N]}(\xi,\tau)$ on the region $R_{k}^{\pm},k=1,2,\cdots,N$, as $\tau\rightarrow \pm \infty$. Similarly, we divide it into two conditions. One is that along the line $L_{k-1}$, if there is only one parameter $z_h$ such that $\eta_{h}=\text{const}$, as $\tau\rightarrow \pm \infty$.
Then, we consider that on the region $R_{k}^+$, the value $\eta_j$ satisfies $\eta_j\rightarrow -\infty,j=h+1,h+2,\cdots, m$ and $-\eta_j\rightarrow -\infty,j=1,2,\cdots, h $, as $\tau\rightarrow +\infty$.  Similarly, as $\tau\rightarrow -\infty$, the values $-\eta_j\rightarrow -\infty,j=h+1,h+2,\cdots, m$ and $\eta_j\rightarrow -\infty,j=1,2,\cdots, h$. Thus, combining with equations \eqref{eq:X1-X2} and \eqref{eq:X1-X2-2}, we prove case (i) in Lemma \ref{lemma:asy-R}. Similarly, we also could obtain case (ii) in Lemma \ref{lemma:asy-R}.
\end{proof}

\begin{lemma}\label{lemma:r-r}
	The function $r_i$ have the following properties:
	\begin{itemize}
		\item For $\ii(z_i-l)=\pm \ii K'-z_{iI}$, $z_{iI}=\Im(z_i)$ i.e., $\lambda(z_i)\in \ii \mathbb{R}$, $|\Im(\lambda_i)|>\frac{\alpha(1+k')}{2}$ when $l=\frac{K'}{2}$ or $|\Im(\lambda_i)|>\frac{\alpha}{2}$ when $l=0$, we get
		\begin{equation}\label{eq:r-r-1}
			\begin{split}
				\frac{r_i}{r_i^*}=&-1, \qquad l=0, \qquad \text{and} \qquad 
				\frac{r_i}{r_i^*}=-\ee^{\frac{2\ii\Im(z_i)\pi }{2K}}, \qquad l=\frac{K'}{2};
			\end{split}
		\end{equation}  
		\item For $\ii(z_i-l)=-z_{iI}$, $z_{iI}=\Im(z_i)$, i.e., $\lambda(z_i)\in\ii \mathbb{R}$, $|\Im(\lambda_i)|<\frac{\alpha(1-k')}{2}$ when $l=\frac{K'}{2}$ or $|\Im(\lambda_i)|<\frac{\alpha}{2}$ when $l=0$, we get
		\begin{equation}\label{eq:r-r-3}
			\begin{split}
				\frac{r_i}{r_i^*}=& 1, \qquad l=0, \qquad \text{and} \qquad 
				\frac{r_i}{r_i^*}=\ee^{\frac{2\ii \Im(z_i)\pi }{2K}}, \qquad l=\frac{K'}{2};
			\end{split}
		\end{equation}  
		\item For  $\ii(z_i-l)\neq\pm \ii K'-z_{iI}$, $z_{iI}=\Im(z_i)$ and $z_{i+1}=-z_i^*+2l$, i.e., $\lambda(z_i)\in \mathbb{C} \backslash (\ii \mathbb{R}\cup \mathbb{R})$, we get 
		\begin{equation}\label{eq:r-r-2}	
			\frac{r_ir_{i+1}}{r_i^*r_{i+1}^*}=1, \qquad l=0,\qquad \text{and} \qquad
			\frac{r_ir_{i+1}}{r_i^*r_{i+1}^*}=\ee^{\frac{4\ii \Im(z_i)\pi }{2K}}, \qquad l=\frac{K'}{2}.
		\end{equation}
	\end{itemize}
\end{lemma}

\begin{proof}	
	By the function $r_i$ defined in equation \eqref{eq:r_i} and the shift formulas of the Jacobi theta function \eqref{eq:Jacobi Theta-K iK}, we know that if $\ii(z_i-l)=\pm \ii K'-z_{iI}$,
	\begin{equation}\label{eq:r-r}
		\begin{split}
			r_i=&\frac{\vartheta_2(\frac{\pm \ii K'-z_{iI}}{2K})}{\vartheta_4(\frac{\pm \ii K'-z_{iI}}{2K})}=\pm \ii \frac{\vartheta_3(\frac{z_{iI} }{2K})}{\vartheta_1(\frac{z_{iI}}{2K})}, \qquad l=0,\\
			r_i=&\frac{\vartheta_2(\frac{\ii(z_i-l)+\ii K'}{2K})}{\vartheta_4(\frac{\ii(z_i-l)}{2K})}=\pm \ii \frac{\vartheta_2(\frac{z_{iI} }{2K})}{\vartheta_1(\frac{z_{iI}}{2K})}\exp\left( \pm\frac{3K'+2\ii z_{iI}}{4K}\pi \right), \qquad l=\frac{K'}{2},
		\end{split}
	\end{equation}
	which implies that $\frac{r_i}{r_i^*}=-1$ when $l=0$ and $\frac{r_i}{r_i^*}=-\ee^{\frac{2\ii\Im(z_i)\pi }{2K}}$ when $l=\frac{K'}{2}$. Thus, equation \eqref{eq:r-r-1} holds. 
	
	When $\ii(z_i-l)=-z_{iI}$, plugging them into equation \eqref{eq:r_i}, we obtain
	\begin{equation}
		\begin{split}
			r_i=&\frac{\vartheta_2(\frac{z_{iI}}{2K})}{\vartheta_4(\frac{z_{iI}}{2K})}, \quad l=0,\qquad \text{and} \qquad
			r_i=\frac{\vartheta_2(\frac{-z_{iI}+2\ii l}{2K})}{\vartheta_4(-\frac{z_{iI}}{2K})}=\frac{\vartheta_3(\frac{z_{iI}}{2K})}{\vartheta_4(\frac{z_{iI}}{2K})}\exp\left(\frac{2\ii z_{iI}+K'}{4K}\pi\right), \quad l=\frac{K'}{2}.
		\end{split}	
	\end{equation}
	Therefore, we get that  $\frac{r_i}{r_i^*}=1$ when $l=0$ and $\frac{r_i}{r_i^*}=\ee^{\frac{2\ii\Im(z_i)\pi }{2K}}$ when $l=\frac{K'}{2}$. Thus, equation \eqref{eq:r-r-3} holds. 
	
	Then, when $z_{i+1}=-z_i^*+2l$, we consider 
	\begin{equation}
		\begin{split}
			r_{i+1}=&\frac{\vartheta_2(\frac{\ii(-z_i^*+2l+l)}{2K})}{\vartheta_4(\frac{\ii(-z_i^*+2l-l)}{2K})}=\frac{\vartheta_2(\frac{-\ii(z_i^*+l)}{2K})}{\vartheta_4(\frac{-\ii(z_i^*-l)}{2K})}=r_i^*, \qquad l=0,\\
			r_{i+1}=&\frac{\vartheta_2(\frac{-\ii(z_i^*+l)+4\ii l}{2K})}{\vartheta_4(\frac{-\ii(z_i^*-l)}{2K})}=\frac{\vartheta_2(\frac{-\ii(z_i^*+l)}{2K})}{\vartheta_4(\frac{-\ii(z_i^*-l)}{2K})}\exp\left(\frac{-(z_i^*-l)}{K}\pi\right)=r_i^*\ee^{\frac{-(z_i^*-l)}{K}\pi}, \qquad l=\frac{K'}{2},
		\end{split}
	\end{equation}
	which implies that equation \eqref{eq:r-r-2} holds.
\end{proof}

\newenvironment{aproof-R}{\emph{Proof of Theorem \ref{theorem:asy-R-elliptic}.}}{\hfill$\Box$\medskip}
\begin{aproof-R}
	We set that the number of $z_i,i=1,2,\cdots,m,$ satisfying $\ii(z_i-l)\in \mathbb{R}$ is $n_1-p$. Combining with Lemma \ref{lemma:asy-R} and Lemma \ref{lemma:r-r}, we obtain that when $l=0$, equation \eqref{eq:solution-u-n-infy-R} could be rewritten as 
\begin{equation}
	\begin{split}
		\hat{u}^{[N]}(\xi,\tau ;R_{q}^{\pm})\longrightarrow&\frac{\alpha \vartheta_2\vartheta_4 }{\vartheta_3\vartheta_3}\mathbf{r}_{h,h}^{\pm}
		\frac{\vartheta_2(\frac{\alpha \xi +s_{h,h}^{\pm}}{2K}) }{\vartheta_4(\frac{\alpha \xi +s_{h,h}^{\pm}}{2K})}
		=(-1)^{p}\alpha k \cn(\alpha \xi +s_{h,h}^{\pm}),
	\end{split}
\end{equation}
where $s_{h,h}^{\pm}$ and $\mathbf{r}_{h,h}^{\pm}$ are defined in equations \eqref{eq:s-i-j} and \eqref{eq:define-r-i-j} respectively, and  $p$ is the number of spectral parameter $\lambda_i\in \ii \mathbb{R}$, $i=1,2,\cdots,N$, satisfying $|\Im(\lambda_i)|>\frac{\alpha(1+k')}{2}$ when $l=\frac{K'}{2}$ or $|\Im(\lambda_i)|>\frac{\alpha}{2}$ when $l=0$. If $l=\frac{K'}{2}$, by formulas \eqref{eq:Jacobi Theta-K iK}, 
equation \eqref{eq:solution-u-n-infy-R} could be rewritten as 
\begin{equation}
	\begin{split}
		\hat{u}^{[N]}(\xi,\tau;R_{k}^{\pm}) \longrightarrow&\frac{\alpha \vartheta_2\vartheta_4 }{A\vartheta_3\vartheta_3(\frac{\ii K'}{2K})}\mathbf{r}_{h,h}^{\pm}
		\frac{\vartheta_2(\frac{\alpha \xi+\ii K' +s_{h,h}^{\pm}}{2K}) }{\vartheta_4(\frac{\alpha \xi +s_{h,h}^{\pm}}{2K})}
		=(-1)^{p}\frac{\alpha \vartheta_2\vartheta_4 }{\vartheta_3\vartheta_2}
		\frac{\vartheta_3(\frac{\alpha \xi +s_{h,h}^{\pm}}{2K}) }{\vartheta_4(\frac{\alpha \xi +s_{h,h}^{\pm}}{2K})}
		=(-1)^{p}\alpha \dn(\alpha \xi +s_{h,h}^{\pm}),
	\end{split}
\end{equation}
where $s_{h,h}^{\pm}$ and $\mathbf{r}_{h,h}^{\pm}$ are defined in equations \eqref{eq:s-i-j} and \eqref{eq:define-r-i-j} respectively, and  $p$ is the number of spectral parameter $\lambda_i\in \ii \mathbb{R}$, $i=1,2,\cdots,N$, satisfying $|\Im(\lambda_i)|>\frac{\alpha(1+k')}{2}$ when $l=\frac{K'}{2}$ or $|\Im(\lambda_i)|>\frac{\alpha}{2}$ when $l=0$.
Then, by the Lemma \ref{lemma:asy-R}, we obtain the Theorem \ref{theorem:asy-R-elliptic}.
\end{aproof-R}

We now perform the exact asymptotic analysis based on the above-mentioned elliptic-localized solution as an example to vividly show the asymptotic behavior. For the Type-I, we consider the two-elliptic-soliton solution $u^{[2]}(x,t)$ in the Case I-2. Setting $l=0$, $k=\frac{1}{2}$, $\alpha=1$, $c_1=c_2=1$, $z_1=K'+\frac{2K}{9}\ii$ and $z_2=K'+\frac{K}{3}\ii$, we obtain the 3d-figure of function $u^{[2]}( x,t)$ under the $(x,t)$-axis in Figure \ref{fig:cn-2b-3d}(a). Then, we consider the corresponding asymptotic analysis under the $(\xi,\tau)$-axis. 
Plugging the above parameters into equations \eqref{eq:mKdV-solution-xi-n-k-vec} and \eqref{eq:X1-X2-hat} or substituting them into equation \eqref{eq:solution-asy-v-simp-i}, we get the asymptotic expressions and draw their graphs in Figure \ref{fig:cn-asy-l}. The blue curves in the Figure \ref{fig:cn-asy-l-1} and Figure \ref{fig:cn-asy-l-2} describe the function $\hat{u}^{[2]}(\xi,\pm 6)$, respectively. The red curve describes $\hat{u}^{[2]}(\xi,6;L_1^+)$ and $\hat{u}^{[2]}(\xi,6;L_2^+)$ in Figure \ref{fig:cn-asy-l-1}. The purple curve describes $\hat{u}^{[2]}(\xi,-6;L_1^-)$ and $\hat{u}^{[2]}(\xi,-6;L_2^-)$ in Figure \ref{fig:cn-asy-l-2}.

\begin{figure}[h]
	\centering
	\subfigure[The asymptotic analysis for the breathers at $t=6$.]{\includegraphics[width=0.43\textwidth]{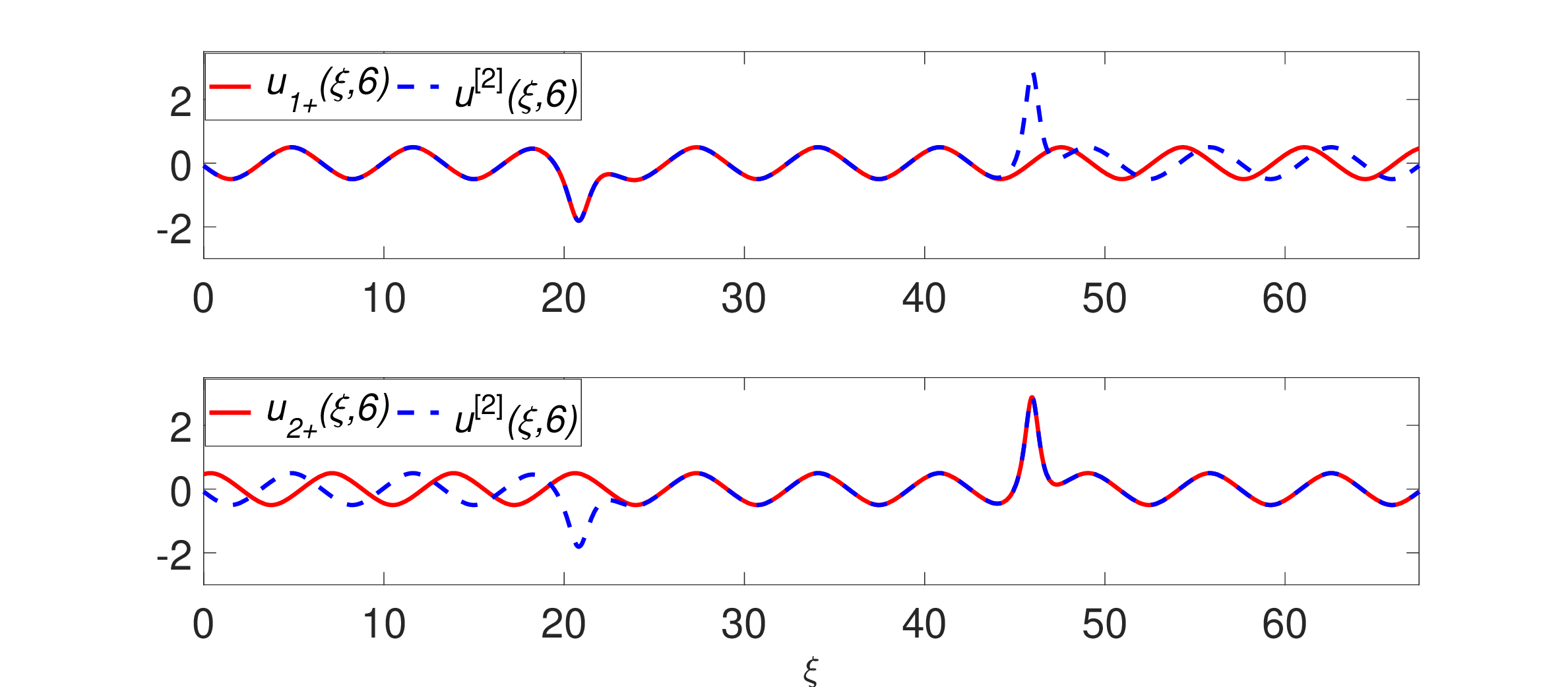}\label{fig:cn-asy-l-1}}
	\subfigure[The asymptotic analysis for the breathers at $t=-6$. ]{\includegraphics[width=0.49\textwidth]{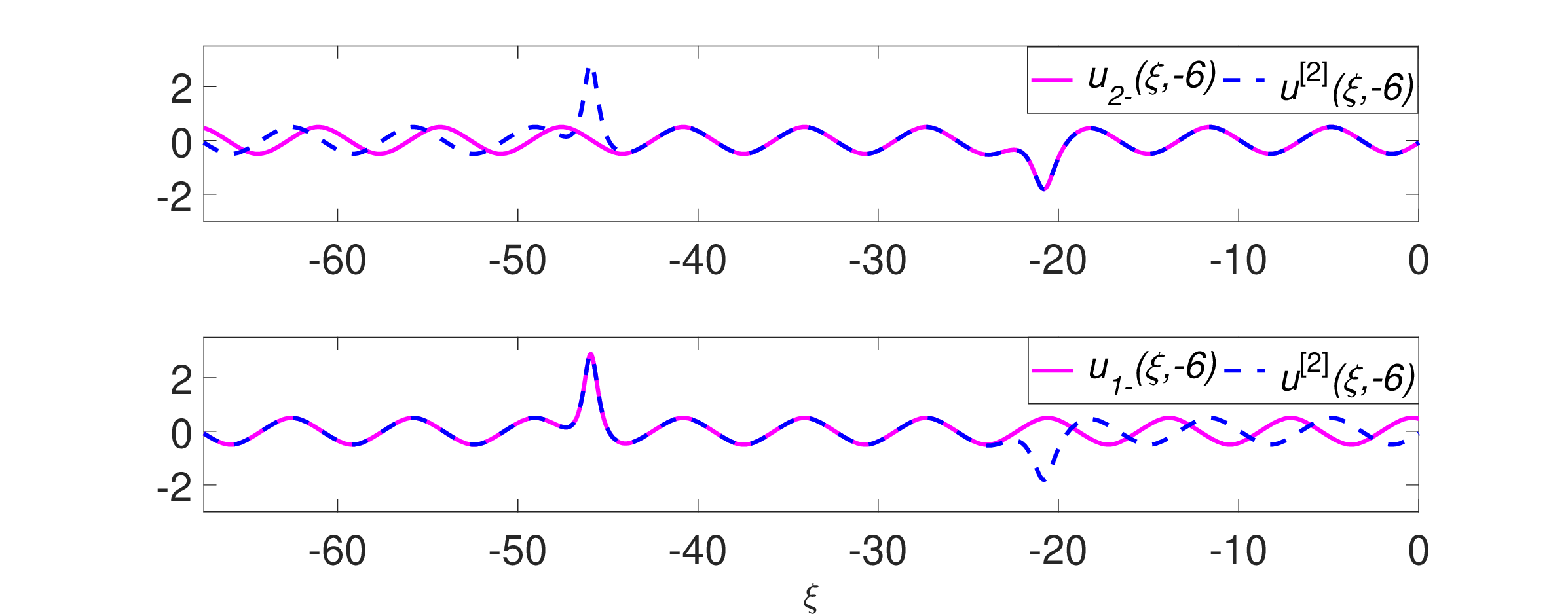}\label{fig:cn-asy-l-2}}
	\caption{The asymptotic analysis for  breathers of solution $\hat{u}^{[2]}(\xi,\tau)$ at $t=\pm 6$. The blue curves describe the solution $\hat{u}^{[2]}(\xi,\pm 6)$. The red and purple curves show the asymptotic expression in equations \eqref{eq:mKdV-solution-xi-n-k-vec} and \eqref{eq:X1-X2-hat} with $l=0$, $k=\frac{1}{2}$, $\alpha=1$, $c_1=c_2=0$, $z_1=K'+\frac{2K}{9}\ii$ and $z_2=K'+\frac{K}{3}\ii$.
	}
	\label{fig:cn-asy-l}
\end{figure}

Then, we consider the asymptotic expressions of solutions on the regions $R_1^+$, $R_2^+$, $R_1^-$. Combining with equation \eqref{eq:solution-u-n-infy-R-cn}, we obtain the following three asymptotic solutions as $\tau\rightarrow+\infty$:
\begin{equation}\label{eq:asy-cn-2-2-R}
	\begin{split}
		\hat{u}^{[2]}( \xi,\tau;R_1^+)
		\rightarrow&
		-k\alpha\cn(\alpha \xi -2\Im(z_1)-2\Im(z_2))=-k\alpha\cn\left(\alpha \xi -\frac{10K}{9}\right),\\
		\hat{u}^{[2]}(\xi,\tau;R_2^+)
		\rightarrow&
		-k\alpha\cn(\alpha \xi +2\Im(z_1)-2\Im(z_2))=-k\alpha\cn\left(\alpha \xi -\frac{2K}{9}\right),\\
		\hat{u}^{[2]}(\xi,\tau;R_1^-)
		\rightarrow&
		-k\alpha\cn(\alpha \xi +2\Im(z_1)+2\Im(z_2))=-k\alpha\cn\left(\alpha \xi+\frac{10K}{9}\right).
	\end{split}
\end{equation}
Similarly, as $\tau\rightarrow-\infty$, we get 
\begin{equation}\label{eq:asy-cn-2-2-R-n}
	\begin{split}
		\hat{u}^{[2]}( \xi,\tau;R_1^-)
		\rightarrow&
		-k\alpha\cn(\alpha \xi +2\Im(z_1)+2\Im(z_2))=-k\alpha\cn\left(\alpha \xi +\frac{10K}{9}\right),\\
		\hat{u}^{[2]}(\xi,\tau;R_2^-)
		\rightarrow&
		-k\alpha\cn(\alpha \xi -2\Im(z_1)+2\Im(z_2))=-k\alpha\cn\left(\alpha \xi +\frac{2K}{9}\right),\\
		\hat{u}^{[2]}(\xi,\tau;R_1^+)
		\rightarrow&
		-k\alpha\cn(\alpha \xi -2\Im(z_1)-2\Im(z_2))=-k\alpha\cn\left(\alpha \xi -\frac{10K}{9}\right).
	\end{split}
\end{equation}
Above six functions in equations \eqref{eq:asy-cn-2-2-R} and \eqref{eq:asy-cn-2-2-R-n} are plotted in Figures \ref{fig:cn-asy-R}(a) and  \ref{fig:cn-asy-R}(b), respectively.

\begin{figure}[h]
	\centering
	\subfigure[The asymptotic analysis for the breathers at $\tau=6$.]{\includegraphics[width=0.495\textwidth]{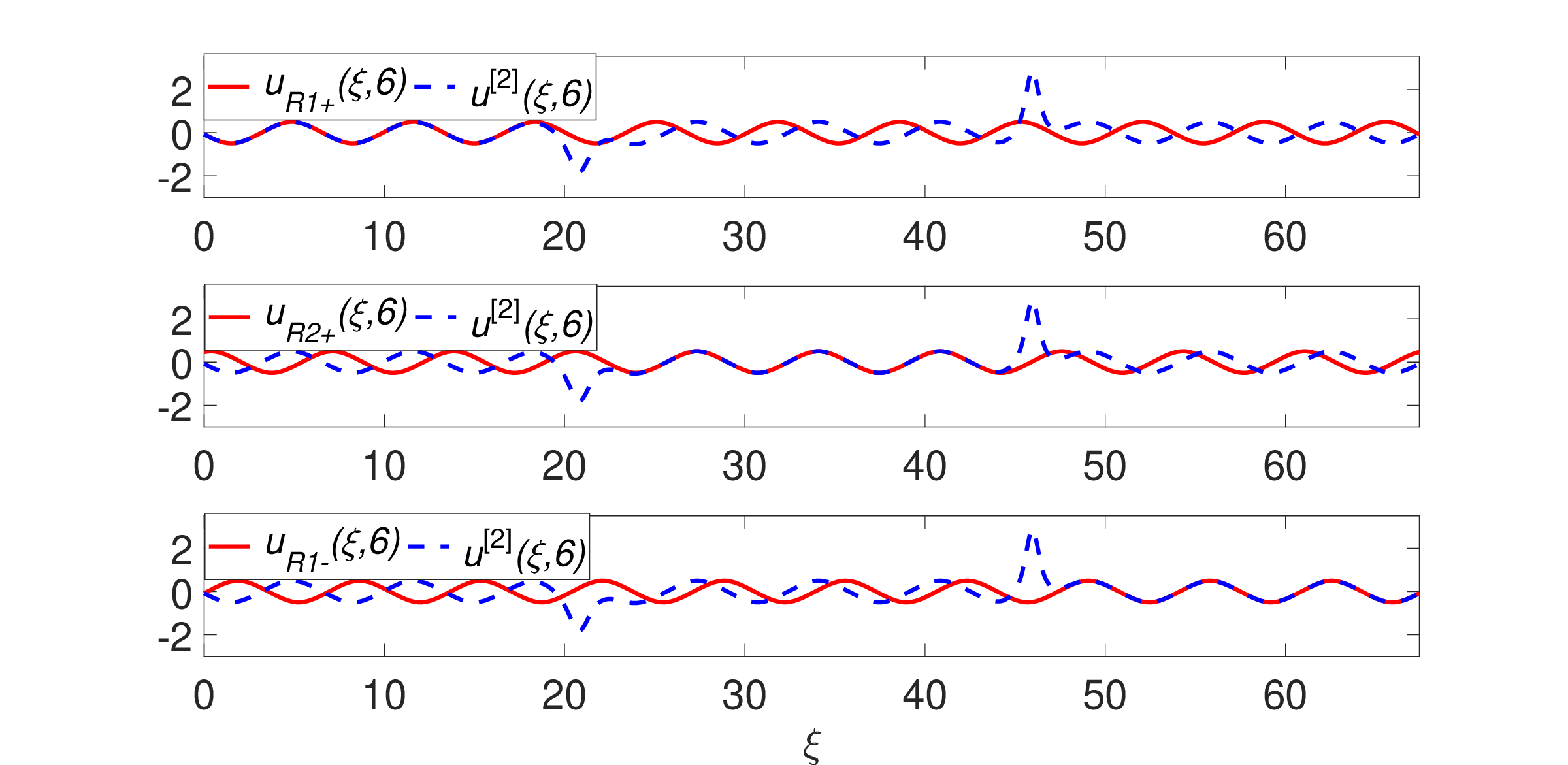}\label{fig:cn-asy-R-1}}
	\subfigure[The asymptotic analysis for the breathers at $\tau=-6$.]{\includegraphics[width=0.495\textwidth]{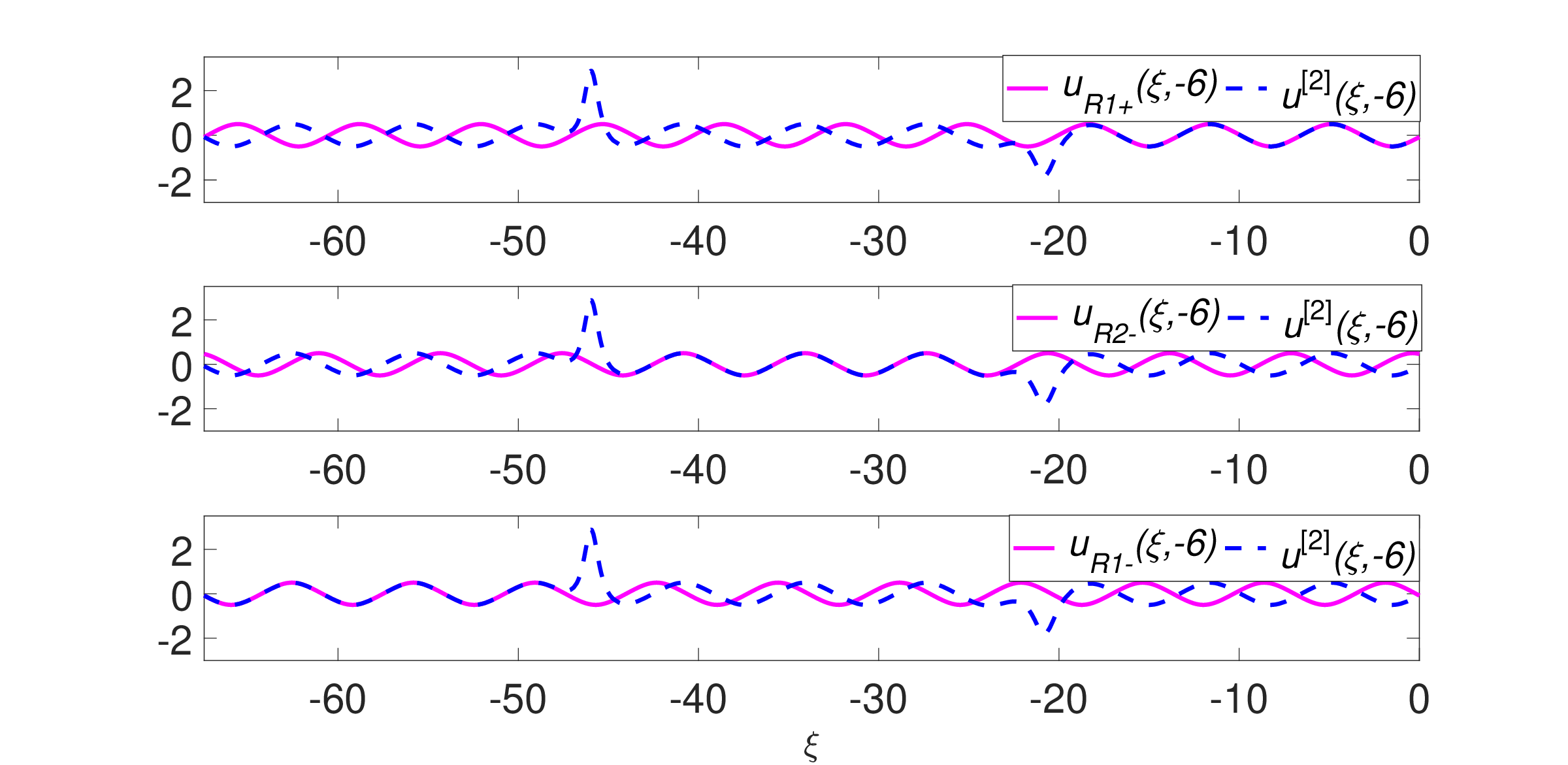}\label{fig:cn-asy-R-2}}
	\caption{The asymptotic analysis for breathers of solution $\hat{u}^{[2]}(\xi,\tau)$ at $\tau=\pm 6$. The blue curves describe the solution $\hat{u}^{[2]}(\xi,\pm 6)$. The red and purple curves show the asymptotic expression on the region $R_{1,2}^{\pm}$ in equations \eqref{eq:asy-cn-2-2-R} and \eqref{eq:asy-cn-2-2-R-n} with $l=0$, $k=\frac{1}{2}$, $\alpha=1$, $c_1=c_2=0$, $z_1=K'+\frac{2K}{9}\ii$ and $z_2=K'+\frac{K}{3}\ii$.
	}
	\label{fig:cn-asy-R}
\end{figure}

\subsection{The symmetry and strictly elastic collision of multi elliptic-localized solutions}\label{sec:symmetric}

In this subsection, we aim to introduce the symmetric property of multi elliptic-localized solutions. From the asymptotic expressions of solutions provided in Theorem \ref{theorem:exact-N-solution}, we know that their interactions are elastic. But if multi elliptic-localized solutions have the symmetry $u^{[N]}(x,t)=u^{[N]}(-x,-t)$, we can claim that the collisions between the breathers and solitons are strictly elastic.

Based on the expression of solution $\Phi(x,t;\lambda)$ \eqref{eq:Lax-solution-Phi} of the Lax pair \eqref{eq:Lax-pair}, the solution $u^{[N]}( x,t)$ has the symmetry $u^{[N]}( x,t)=u^{[N]}( -x,-t)$, if it satisfies the conditions $c_i=1,i=1,2,\cdots,m$ given in Theorem \ref{theorem:symm}.

\newenvironment{aproof-symm}{\emph{Proof of Theorem \ref{theorem:symm}.}}{\hfill$\Box$\medskip}
\begin{aproof-symm}
	By the formula \eqref{eq:Lax-solution-Phi} and $E_2^{-1}( x,t)=\exp(\alpha \xi Z(2\ii l+K))E_1(x,t)$ (obtained in \cite{LinglmS-21}), when $c_i=1$,	$i=1,2,\cdots,N$, it is easy to verify that 
		\begin{equation}\label{eq:Phi-sym}
		\begin{bmatrix}
		\Phi_{i,1}(- x,-t) \\
		\Phi_{i,2}(- x,-t)
		\end{bmatrix}\equiv\Phi(- x,-t;\lambda_i)	\begin{bmatrix}
		1  \\ 1
		\end{bmatrix}=
		\begin{bmatrix}
		0  & -1  \\ -1 & 0
		\end{bmatrix}\Phi( x,t;\lambda_i)
		\begin{bmatrix}
		0  & 1  \\ 1 & 0
		\end{bmatrix}	\begin{bmatrix}
		1  \\ 1
		\end{bmatrix}
		=\begin{bmatrix}
		-\Phi_{i,2}( x,t) \\
		-\Phi_{i,1}( x,t)
		\end{bmatrix},
		\end{equation}
		which implies $(\Phi_i^{\dagger}(- x,-t)	\Phi_j(- x,-t))^{\dagger}=\Phi_j^{\dagger}(- x,-t)\Phi_i(- x,-t)=\Phi_j^{\dagger}( x,t)\Phi_i( x,t)$.
		Combining the formula $u^{[N]}( x,t)$ in equation \eqref{eq:u-N-breather}, we get that the symmetry $u^{[N]}( -x,-t)=u^{[N]}( x,t)$ holds.
\end{aproof-symm}

From the Theorem \ref{theorem:symm}, we know that the dynamic behavior is consistent at times $t$ and $-t$, because of $u^{[N]}(x,t)=u^{[N]}(-x,-t)$. Thus, the collision dynamics between the breathers are strictly elastic, which means that the shape of breathers does not change after the collision. Two typical examples of the strictly elastic collisions are shown in Figure \ref{fig:cn-elas} and Figure \ref{fig:dn-elas}. Then, we will describe each of the above situations.

Plugging $k=\frac{1}{2}$, $l=0$, $z_1=K'+\ii \frac{K}{3}$, $z_2=K'+\ii \frac{2K}{9}$, $c_1=c_2=1$ into equation \eqref{eq:mKdV-solution-xi-n-1}, a two-elliptic-solitons solution $u^{[2]}(x,t)$ is obtained under the $\cn$-type background. The Figure \ref{fig:cn-elas-den} shows the density evolution of solution $u^{[2]}(x,t)$ and the Figure \ref{fig:cn-2b-3d}(a) is the 3d-plot of this solution. Two solitons of $u^{[2]}(x,t)$ in Figure \ref{fig:cn-elas-den} collide at the moment $t=0$. By the symmetry of solution $u^{[2]}(x,t)$ proved in Theorem \ref{theorem:symm}, we know that the behaviors of the above two solitons in solution $u^{[2]}(x,t)$ are the same at the corresponding moments before and after the collision. Here, $t = 5$, for example, we plot the sectional view of functions $u^{[2]}(x,5)$ and $u^{[2]}(-x,-5)$ shown in Figure \ref{fig:cn-elas-2d} to reflect the variation before the collision ($t=-5$) and after the collision ($t=5$). The above variation clearly depicts the same solitons in the above two positions, which provide a vivid image of the strictly elastic collision.

\begin{figure}[h]
	\centering
	\subfigure[The dentisy plot of function $u^{[2]}(x,t)$.]{\includegraphics[width=0.4\textwidth]{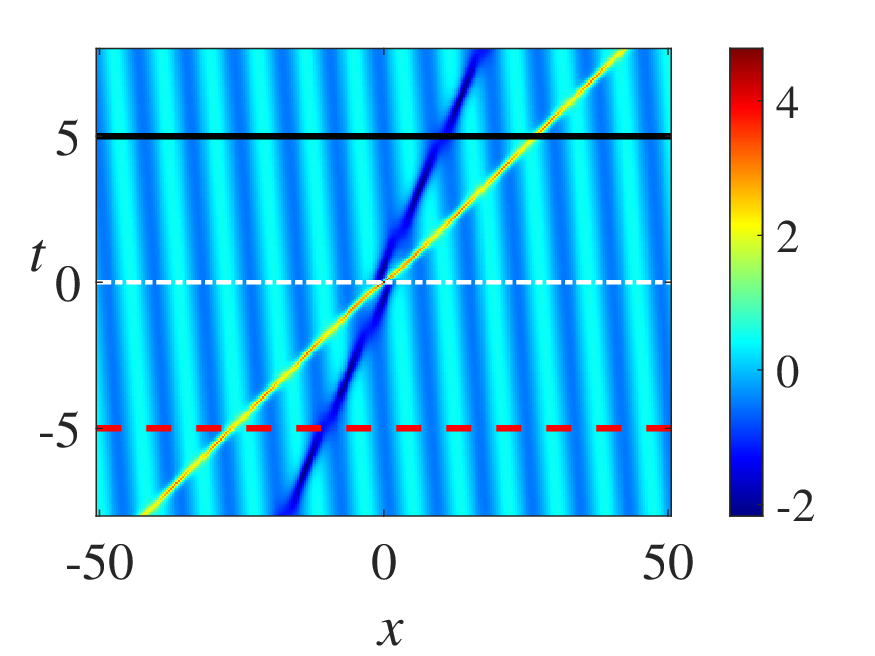}\label{fig:cn-elas-den}}
	\subfigure[The sectional view of functions $u^{[2]}(x,5)$ and $u^{[2]}(-x,-5)$.  ]{\raisebox{0.85\height}{\includegraphics[width=0.58\textwidth]{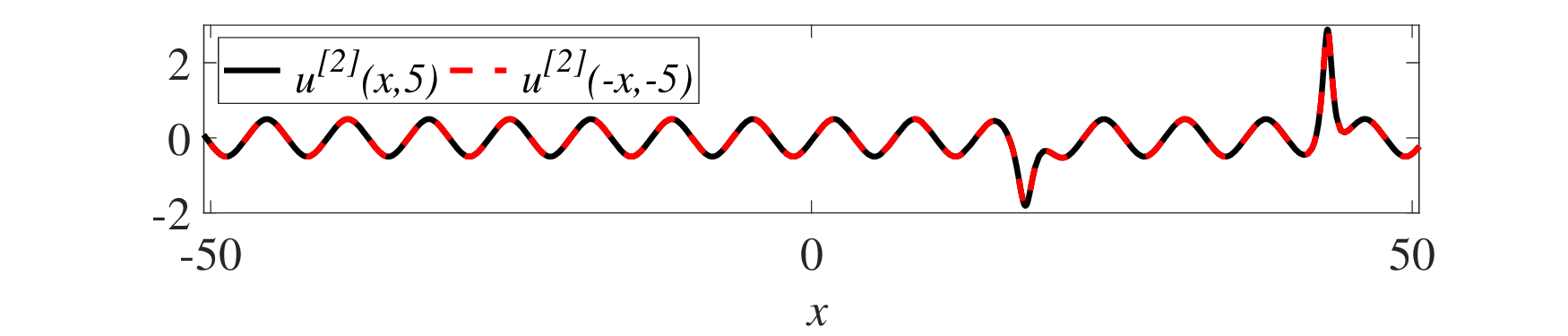}\label{fig:cn-elas-2d}}}
	\caption{The two-elliptic-solitons solution $u^{[2]}(x,t)$ under the $\cn$-type background describes the strictly elastic collision between two bound-state solitons.}
	\label{fig:cn-elas}
\end{figure}

Similarly, the Figure \ref{fig:dn-elas} also reflects a strictly elastic collision, which is obtained under the $\dn$-type background with parameters $k=\frac{9}{10}$, $l=\frac{K'}{2}$, $z_1=-\frac{K'}{8}+\ii \frac{2K}{5}$, $z_2=-\frac{K'}{2}+\ii \frac{2K}{5}$, $z_3=\frac{9K'}{8}+\ii \frac{2K}{5}$ and $c_1=c_2=1$. The Figure \ref{fig:dn-elas-den} is a density evolution of the elliptic-soliton-breather solution $u^{[2]}(x,t)$ showing the collision between a breather and a soliton at the time $t=0$. And the 3d-plot of solution $u^{[2]}(x,t)$ is shown in Figure \ref{fig:dn-2b-3d}(b). In Figure \ref{fig:dn-elas-2d}, we draw the sectional view of functions $u^{[2]}(x,6)$ and $u^{[2]}(-x,-6)$, which reflects the variation before ($t=-6$) and after ($t=6$) the collision and shows a complete consistency between them.

\begin{figure}[h]
	\centering
	\subfigure[The dentisy plot of function $u^{[2]}(x,t)$.]{\includegraphics[width=0.46\textwidth]{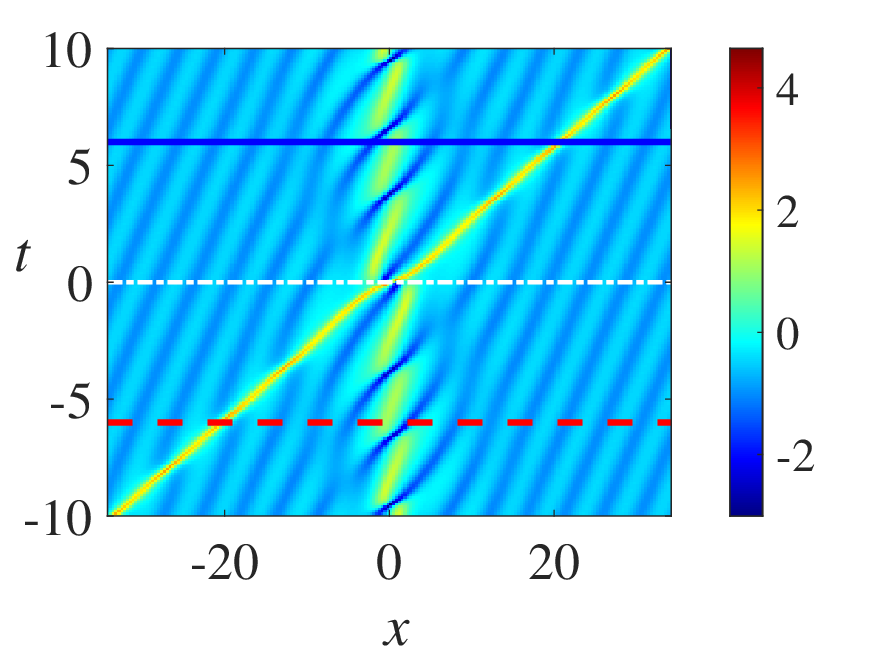}\label{fig:dn-elas-den}}
	\subfigure[The sectional view of functions $u^{[2]}(x,6)$ and $u^{[2]}(-x,-6)$. ]{\raisebox{0.9\height}{\includegraphics[width=0.51\textwidth]{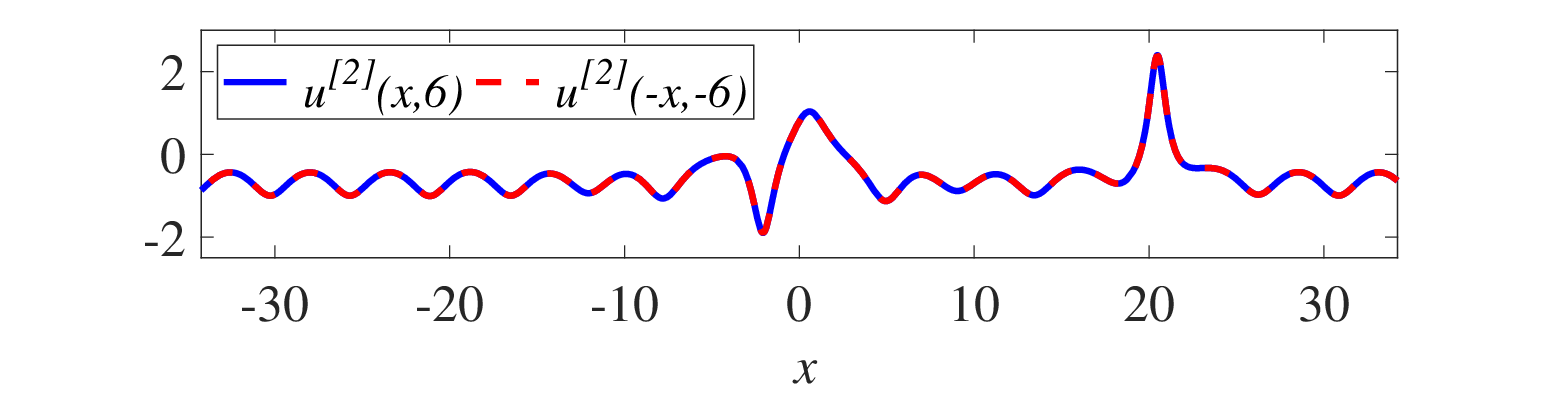}\label{fig:dn-elas-2d}}}
	\caption{The elliptic-soliton-breather solution $u^{[2]}(x,t)$ under the $\dn$-type background describes a strictly elastic collision between the bound-state soliton and the bound-state breather.}
	\label{fig:dn-elas}
\end{figure}

Suppose the multi elliptic-localized solutions $u^{[N]}(x,t)$ do not satisfy the condition $c_i= 1,i=1,2,\cdots,N$, provided in Theorem \ref{theorem:symm}. In that case, it is hard to obtain whether solutions $u^{[N]}(x,t)$ are symmetric or not. Therefore, we do not know whether collisions between breathers and solitons are strictly elastic or not.
	
However, when we consider the asymptotic expression of the multi elliptic-localized solutions, 
we could get that as $\tau\rightarrow \pm \infty$ along the evolution direction of breathers and solitons, the asymptotic expressions of them $\hat{u}^{[N]}(\xi,\tau;L_i^{\pm})$, $i=1,2,\cdots,N$, provided in Theorem \ref{theorem:exact-N-solution} are consistent. Under the transformation $\xi=x-st$, $\tau=t$ in \eqref{eq:xi-x-solution}, we know that as $t\rightarrow\pm \infty$ the asymptotic expressions of multi elliptic-localized solutions $u^{[N]}(x,t)$ are also consistent, which reflects that before and after the collision, the asymptotic expressions of solutions are consistent. In that case, we could claim that the collisions between breathers and solitons of the multi elliptic-localized solutions $u^{[N]}(x,t)$ in equation \eqref{eq:mKdV-solution-xi-n-1} are elastic. Without condition $c_i= 1,i=1,2,\cdots,N$ of the multi elliptic-localized solutions $u^{[N]}(x,t)$, we do not know whether the collisions between the breathers and solitons are strictly elastic or not.

\section{The degeneration of the multi elliptic-localized solutions}\label{sec:elliptic-constant}

Considering elliptic function solutions $u(x,t)$ \eqref{eq:mKdV-solution-cn-dn} of equation \eqref{eq:mKdV-equation}, we find that as $k\rightarrow 0^+$ above solutions degenerate into constant solutions, which link the periodic solutions and constant solutions together. Then, we want to consider whether the relationship exists between solutions under the background of elliptic functions and solutions on the vanishing or constant background.

Utilizing the approximation formulas \eqref{eq:cd-app}, we can easily get the results
\begin{equation}\label{eq:mKdV-solution-k-0}
	\lim_{k\rightarrow 0^+}\alpha\dn(\alpha(x-st),k)=\alpha,\qquad \qquad \qquad
    \lim_{k\rightarrow 0^+}\alpha k\cn(\alpha(x-st),k)=0.
\end{equation} 
Accordingly, we consider the limitation of multi elliptic-localized solutions $u^{[N]}(x,t)$ in equation \eqref{eq:mKdV-solution-xi-n-1} as $k\rightarrow 0^+$. Here, we mainly study the function $\Phi(x,t;\lambda)$ in equation \eqref{eq:Lax-solution-Phi}. If the fundamental solution $\Phi(x,t;\lambda)$ of Lax pair \eqref{eq:Lax-pair} could degenerate into the fundamental solution of the corresponding Lax pair with constant-valued function $u(x,t)$, we could obtain that the multi elliptic-localized solutions could degenerate into solitons or breathers under the vanishing or constant backgrounds by using the Darboux-B\"{a}cklund transformation.

Reviewing the results in our previous work \cite{LinglmS-21} (Theorem 6), we obtain that the fundamental solution of Lax pair \eqref{eq:Lax-pair} could also be rewritten as
\begin{equation}\label{eq:Lax-pair-funda}
	\Phi(x,t;\lambda)=
	\begin{bmatrix}
		\sqrt{u^2(x,t)-\beta_1}\exp(\theta_1) & 
		\sqrt{u^2(x,t)-\beta_2}\exp(\theta_2) \\
		-\sqrt{u^2(x,t)-\beta_2}\exp(-\theta_2) &
		-\sqrt{u^2(x,t)-\beta_1}\exp(-\theta_1) 
	\end{bmatrix},
\end{equation} 
where 
\begin{equation}\label{eq:beta-theta}
	\beta_{1}=2\lambda^2+\frac{s}{2}-2y, \quad  \beta_{2}=2\lambda^2+\frac{s}{2}+2y, \qquad \theta_{i}=\int_{0}^{\xi}\frac{2\ii \lambda \beta_i}{u^2(x)-\beta_i}\dd x+\ii \lambda \xi \pm 4\ii \lambda y t,\quad i=1,2,
\end{equation}
$s$ is defined in equation \eqref{eq:mKdV-solution-cn-dn}; $y$ satisfies the algebraic curve \begin{equation}\label{eq:y}
	y^2=(\lambda-\hat{\lambda}_1)(\lambda-\hat{\lambda}_2)(\lambda-\hat{\lambda}_3)(\lambda-\hat{\lambda}_4),
\end{equation}
and the value of $\hat{\lambda}_i$ are shown in Appendix \ref{appendix:conformal-approximation}.

\begin{prop}\label{prop:k-cn-dn}
	 As $k\rightarrow 0^+$, we analyze the value of $s$ and the expression of $y^2$ as follows:
	 \begin{itemize}
	 	\item Under the $\cn$-type background, i.e., $l=0$, we get $\lim_{k\rightarrow 0^+}s=-\alpha^2$ and $\lim_{k\rightarrow 0^+}y^2= \left(\lambda^2-\frac{\alpha^2}{4}\right)^2$;
	 	\item Under the $\dn$-type background, i.e., $l=\frac{K'}{2}$, we get $\lim_{k\rightarrow 0^+}s=2\alpha^2$ and $\lim_{k\rightarrow 0^+}y^2 =\lambda^2\left(\lambda^2+\alpha^2\right)$.
	 \end{itemize}
\end{prop}

\begin{proof}
	From the solution \eqref{eq:mKdV-solution-cn-dn}, we know that the parameter $s$ in the different backgrounds is different, so the limitation of $s$ are divided into the following two cases:
	\begin{equation}\label{eq:s-k}
		\lim_{k\rightarrow 0^+}s=\lim_{k\rightarrow 0^+}\alpha^2(2k^2-1)=-\alpha^2, \quad l=0 \qquad  \text{and} \qquad \lim_{k\rightarrow 0^+}s=\lim_{k\rightarrow 0^+}\alpha^2(2-k^2)=2\alpha^2, \quad l=\frac{K'}{2}. 
	\end{equation}
	Then, we consider the value of $y$ in equation \eqref{eq:y}.
	The parameters $\hat{\lambda}_i=\lambda(\hat{z}_i)$ could be obtained by the following four points $\hat{z}=\pm \frac{K'}{2}\pm \ii \frac{K}{2}$. Combining with equations \eqref{eq:lambda-elliptic}, \eqref{eq:lambda-i-0} and \eqref{eq:lambda-i-K1}, we obtain 
	\begin{equation}\label{eq:lambda-k-0}
		\begin{split}
			&\lim_{k\rightarrow 0^+}\hat{\lambda}_1= \lim_{k\rightarrow 0^+}\hat{\lambda}_2=\frac{\alpha}{2},\qquad 	\qquad 
			\lim_{k\rightarrow 0^+}\hat{\lambda}_3=
			\lim_{k\rightarrow 0^+}\hat{\lambda}_4=-\frac{\alpha}{2}, \qquad l=0,\\
			&\lim_{k\rightarrow 0^+}\hat{\lambda}_1=\lim_{k\rightarrow 0^+}\hat{\lambda}_2=0,\qquad
			\lim_{k\rightarrow 0^+}\hat{\lambda}_3=\ii \alpha, \qquad \lim_{k\rightarrow 0^+}\hat{\lambda}_4=-\ii \alpha, \qquad l=\frac{K'}{2}
		\end{split}
	\end{equation} 
The above calculation process is given in equations \eqref{eq:lambda-i-0} and \eqref{eq:lambda-i-K1} of  Appendix \ref{appendix:conformal-approximation}.
 Combining above results, we get
	\begin{equation}\label{eq:asy-y-0}
		\lim_{k\rightarrow 0^+} y^2= \left(\lambda^2-\frac{\alpha^2}{4}\right)^2,\quad l=0,
		 \qquad  \qquad
		  \lim_{k\rightarrow 0^+} y^2=\lambda^2 \left(\lambda^2+\alpha^2\right), \quad l=\frac{K'}{2}.
	\end{equation}
\end{proof}

\begin{remark}
	 Considering the expression of the fundamental solution $\Phi(x,t;\lambda)$ and definition of parameters $\beta_i,\theta_i$ in equation \eqref{eq:beta-theta}, we could obtain that if we change the sign of $y$, i.e., from $\sqrt{\prod_{i=1}^{4}(\lambda-\lambda_i)}$ to $ -\sqrt{\prod_{i=1}^{4}(\lambda-\lambda_i)}$, the solution $\Phi(x,t;\lambda)$ is just taking a column transformation. Thus, the fundamental solutions $\Phi(x,t;\lambda)$ are equivalent to each other whatever we take positive or negative sign of parameter $y$.
\end{remark}

For the different backgrounds of elliptic function solutions, we divide them into the following two conditions to study the limitations of $\Phi(x,t;\lambda)$, as $k\rightarrow 0^+$:
\begin{theorem}\label{theorem:asy-k}
	When $l=0$, as $k\rightarrow 0^+$, the limit of matrix $\Phi(x,t;\lambda)$ is 
	\begin{equation}\label{eq:Phi-lim-0}
			\lim_{k\rightarrow 0^+}\frac{\Phi(x,t;\lambda)\sigma_3}{\sqrt{u^2(0,0)-\beta_{1}}}	
			=\begin{bmatrix}
			\ee^{-\ii \lambda x -4\ii \lambda^3 t} &  0   \\ 0   &
			\ee^{\ii \lambda x+4\ii \lambda^3 t} 
		\end{bmatrix}.
	\end{equation}
When $l=\frac{K'}{2}$, define $ \nu=\sqrt{-\lambda^2-\alpha^2}$, $\lambda\in(-\ii\infty,-\ii\alpha)\cup(\ii\alpha,\ii\infty)$, we have
\begin{equation}\label{eq:Phi-lim-K1}
		\lim_{k\rightarrow 0^+}\frac{\Phi(x,t;\lambda)\sigma_3}{\sqrt{u^2(0,0)-\beta_{1}}}	
		=\Psi(x,t;\lambda),\,\,\,
	\Psi(x,t;\lambda):=
	\begin{bmatrix}
	\ee^{-\nu(x+2(2\lambda^2-\alpha^2)t)}
	& \frac{\alpha\ee^{\nu(x+2(2\lambda^2-\alpha^2)t)}}{\nu+\ii \lambda} \\
	\frac{\alpha\ee^{-\nu(x+2(2\lambda^2-\alpha^2)t)}}{\nu+\ii \lambda}			
	&\ee^{\nu(x+2(2\lambda^2-\alpha^2)t)}
	\end{bmatrix}.
\end{equation}
\end{theorem}
\begin{proof}
	When $l=0$, combining with the Proposition \ref{prop:k-cn-dn} and equation \eqref{eq:cd-app} and letting $y=-\sqrt{\prod_{i=1}^{4}(\lambda-\lambda_i)}$, we get 
	\begin{equation}\label{eq:lim-theta-0}
		\begin{split}
			\lim_{k\rightarrow 0^+}\theta_1=&	\lim_{k\rightarrow 0^+}\int_{0}^{\xi}\frac{2\ii \lambda \beta_1}{u^2(s)-\beta_1}\dd s+\ii \lambda \xi +4\ii \lambda y t
			=\int_{0}^{\xi}-2\ii \lambda\dd s+\ii \lambda \xi +4\ii \lambda y t
			=-\ii \lambda x -4\ii \lambda^3 t,
		\end{split}
	\end{equation}
and 
\begin{equation}\label{eq:lim-u-beta-0}
	\begin{split}
		\lim_{k\rightarrow 0^+} u^2(x,t)-\beta_{2}=&\lim_{k\rightarrow 0^+} (\alpha k\cn(x-st))^2-\left(2\lambda^2+\frac{s}{2}+2y\right)=-2\lambda^2+\frac{\alpha^2}{2}+2\left(\lambda^2-\frac{\alpha^2}{4}\right)=0,\\
		\lim_{k\rightarrow 0^+} u^2(x,t)-\beta_{1}=&\lim_{k\rightarrow 0^+} (\alpha k\cn(x-st))^2-\left(2\lambda^2+\frac{s}{2}-2y\right)=\alpha^2-4\lambda^2.
	\end{split}	
\end{equation}
Furthermore, since $\lim_{k\rightarrow 0^+}\alpha k\cn(x-st)=\lim_{k\rightarrow 0^+}\alpha k\cn(0)=0$, it is easy to obtain that 
\begin{equation}\label{eq:lim-f-0}
	\lim_{k\rightarrow 0^+} u^2(0,0)-\beta_{2}=\lim_{k\rightarrow 0^+} u^2(x,t)-\beta_{2}=0, \qquad \qquad u^2(0,0)-\beta_{1}=\lim_{k\rightarrow 0^+} u^2(x,t)-\beta_{1}=\alpha^2-4\lambda^2.
\end{equation}
Thus, combining with equations \eqref{eq:lim-theta-0}, \eqref{eq:lim-u-beta-0} and \eqref{eq:lim-f-0}, we get equation \eqref{eq:Phi-lim-0}.

	Similarly, we consider the case $l=\frac{K'}{2}$ and let $y=\sqrt{\prod_{i=1}^{4}(\lambda-\lambda_i)}$, which means $y_0=\lim_{k\rightarrow 0^+}y=\ii\lambda\nu$. It is easy to verify that 
	\begin{equation}\label{eq:int-simpli}
			\begin{split}
			\int_{0}^{x}\frac{2\ii \lambda \left(2\lambda^2+\alpha^2\mp 2y_0\right)}{\alpha^2-\left(2\lambda^2+\alpha^2\mp 2y_0\right)}\dd s+\ii \lambda x
			=&-\ii \lambda \frac{\lambda^2\mp y_0+\alpha^2}{\lambda^2\mp y_0}x
			=-\ii \lambda \frac{\lambda^4- y_0^2+\alpha^2(\lambda^2\pm y_0)}{\lambda^4- y_0^2}x
			=\pm \ii \frac{y_0}{\lambda}x.
		\end{split}
	\end{equation}
	Based on the definition of $\beta_{1,2}$ in equation \eqref{eq:beta-theta}, the Proposition \ref{prop:k-cn-dn}, equations \eqref{eq:cd-app} and \eqref{eq:int-simpli}, we obtain
\begin{equation}
	\begin{split}
		\lim_{k\rightarrow 0^+}\theta_{1,2}=&	\lim_{k\rightarrow 0^+}\int_{0}^{\xi}\frac{2\ii \lambda \beta_{1,2}}{u^2(s)-\beta_{1,2}}\dd s+\ii \lambda \xi \pm 4\ii \lambda y t\\
		=&\int_{0}^{(x-2\alpha^2 t)}\frac{2\ii \lambda \left(2\lambda^2+\alpha^2\mp 2y_0\right)}{\alpha^2-\left(2\lambda^2+\alpha^2\mp 2y_0\right)}\dd s+\ii \lambda (x-2\alpha^2 t)\pm 4\ii \lambda y_0 t\\
		=&\pm \ii \frac{y_0}{\lambda}(x-2\alpha^2 t)\pm 4\ii \lambda y_0 t\\
		=&\mp\nu (x+2(2\lambda^2-\alpha^2)t ).
	\end{split}
\end{equation}
Furthermore, we also could get
\begin{equation}\label{eq:u-u-2}
\begin{split}
	\lim_{k\rightarrow 0^+}\frac{u^2(x,t)-\beta_{2}}{u^2(0,0)-\beta_{1}}
	=&\frac{\alpha^2-\left(2\lambda^2+\alpha^2+2y_0\right)}{\alpha^2-\left(2\lambda^2+\alpha^2-2y_0\right)}
	=\frac{\lambda^4-y_0^2}{(\lambda^2-y_0)^2}=\frac{\alpha^2}{2y_0-\alpha^2-2\lambda^2}=\left(\frac{-\alpha}{\ii \lambda +\nu}\right)^2,\\
	\lim_{k\rightarrow 0^+}\frac{u^2(x,t)-\beta_{1}}{u^2(0,0)-\beta_{1}}
	=&\frac{\alpha^2-\left(2\lambda^2+\alpha^2-2y_0\right)}{\alpha^2-\left(2\lambda^2+\alpha^2-2y_0\right)}
	=1.
\end{split}
\end{equation}
Thus, equation \eqref{eq:Phi-lim-K1} holds.
\end{proof}

\begin{remark}
	Expanding the right side of the matrix function \eqref{eq:Phi-lim-K1} in the small neighborhood of $\alpha=0$, we obtain 
	\begin{equation}
		\begin{split}
			\begin{bmatrix}
				\ee^{-\nu(x+2(2\lambda^2-\alpha^2)t)}
				& \frac{\alpha\ee^{\nu(x+2(2\lambda^2-\alpha^2)t)}}{\nu+\ii \lambda} \\
				\frac{\alpha\ee^{-\nu(x+2(2\lambda^2-\alpha^2)t)}}{\nu+\ii \lambda}			
				&\ee^{\nu(x+2(2\lambda^2-\alpha^2)t)}
			\end{bmatrix}
		=\begin{bmatrix}
			\ee^{-\ii \lambda(x+4\lambda^2 t)}+\mathcal{O}(\alpha^2)
			& 
			\frac{ \ee^{\ii\lambda(x+4\lambda^2t)}}{2\ii \lambda}\alpha+\mathcal{O}(\alpha^2)
			 \\ 
			 \frac{ \ee^{-\ii\lambda(x+4\lambda^2t)}}{2\ii \lambda}\alpha+\mathcal{O}(\alpha^2)		
			&\ee^{\ii \lambda(x+4\lambda^2 t)}+\mathcal{O}(\alpha^2)
		\end{bmatrix}.
		\end{split}
	\end{equation}
Comparing with the right side of equation \eqref{eq:Phi-lim-0}, we get that equation \eqref{eq:Phi-lim-0} could be seen as the degeneration of equation \eqref{eq:Phi-lim-K1} as $\alpha \rightarrow 0$. So, in the following analysis, we will just consider equation \eqref{eq:Phi-lim-K1}.
\end{remark}

The limitation of solution $\Phi(x,t;\lambda)$
as $k\rightarrow0^+$ is 
\begin{equation}
	\lim_{k\rightarrow 0^+} \Phi(x,t;\lambda)=\Psi(x,t;\lambda) A(\lambda),	\qquad 	A(\lambda):=\lim_{k\rightarrow 0^+} \frac{\sigma_3}{\sqrt{u^2(0,0)-\beta_{1}}}
	=a(\lambda)\sigma_3,
\end{equation}
and 
\begin{equation}
	\lim_{k\rightarrow 0^+}\Phi_i:=\lim_{k\rightarrow 0^+}\Phi(x,t;\lambda_i)
	\begin{bmatrix}
		1 \\c_i
	\end{bmatrix}
	=\Psi(x,t;\lambda_i)\lim_{k\rightarrow 0^+}A(\lambda_i)
	\begin{bmatrix}
		1 \\c_i
	\end{bmatrix}
	=\Psi(x,t;\lambda_i)a(\lambda_i)
	\begin{bmatrix}
		1 \\ -c_i
	\end{bmatrix}=\Psi_i.
\end{equation}
Based on the Darboux-B\"{a}cklund transformation and collecting the multi soliton solutions \cite{ChanL-1994} and the multi breather solutions \cite{TajiriW-98}, we obtain that as $k\rightarrow 0^+$ the solutions $u^{[N]}(x,t)$ degenerate into the multi soliton solutions, the multi breather solutions and the multi soliton-breather solutions.

Then, we will describe the above solutions degeneration for choosing different spectral parameters $\lambda$. Here, we just consider the degeneration of the elliptic-soliton solution and the elliptic-breather solutions, since all the degeneration cases of multi elliptic-localized solutions are contained in the above two solutions. 

The elliptic-soliton solution $u^{[1]}(x,t)$ is obtained by the parameter $z_1=mK'+l+\ii z_I,m=-1,0,1$, i.e., $\lambda_1=\lambda(z_1)\in \ii \mathbb{R}$ gained from Appendix \ref{appendix:conformal-approximation}. Based on the conformal mapping between $\lambda$ and $z$ provided in \cite{LinglmS-21}, we just consider the parameter $\lambda$ in the following paper. Combining the Theorem \ref{theorem:asy-R-elliptic}, we know that 
\begin{subequations}\label{eq:v-asy-R}
	\begin{align}
		&\lim_{k\rightarrow 0^+}\hat{u}^{[1]}(\xi,\tau;R_{1}^{\pm}) \rightarrow0, \quad \tau\rightarrow \pm \infty, \quad l=0 \quad \text{and}	\quad \lambda_1\in \ii \mathbb{R}, \label{eq:v-asy-R-1}\\
		&\lim_{k\rightarrow 0^+}\hat{u}^{[1]}(\xi,\tau;R_{1}^{\pm}) \rightarrow\alpha , \quad \tau\rightarrow \pm \infty, \quad l=\frac{K'}{2} \quad\text{and}	\quad \lambda_1\in \ii \mathbb{R}, \quad |\Im(\lambda)|<\frac{\alpha(1-k')}{2}, \label{eq:v-asy-R-2}\\
		&\lim_{k\rightarrow 0^+}\hat{u}^{[1]}(\xi,\tau;R_{1}^{\pm}) \rightarrow-\alpha , \quad \tau\rightarrow \pm \infty, \quad l=\frac{K'}{2} \quad\text{and}	\quad \lambda_1\in \ii \mathbb{R}, \quad |\Im(\lambda)|>\frac{\alpha(1+k')}{2}. \label{eq:v-asy-R-3}
	\end{align}
\end{subequations}
Then we study the variation of peaks as $k\to0^+$. In Section \ref{sec:Darboux-transformation}, we have known that the range of $u^{[1]}(x,t)$ is $ \left[\min(u(x,t))-2|\Im(\lambda_1)|,\max(u(x,t))+2|\Im(\lambda_1)|\right]$. Because the limitation $\lim_{k\rightarrow 0^+}u(x,t)$ is obtained in equation \eqref{eq:mKdV-solution-k-0}, we just study the value of $\lambda_1$. From Appendix \ref{appendix:conformal-approximation}, we consider the conformal map $\lambda(z)$ and obtain the following cases:
\begin{itemize}
\item[(i)] When $l=0$, for any fixed $\lambda_1=\lambda(z_1)\in \ii \mathbb{R}$,
	$\lim_{k\rightarrow 0^+}\lambda(z_1)$ always satisfies $\Im(\lambda(z_1))\in \left(0,\frac{\alpha}{2}\right)\cup \left(\frac{\alpha}{2},+\infty\right)$ and $\lambda(z_1)\in \ii \mathbb{R}$ shown in Appendix \ref{appendix:conformal-approximation}, which implies that the peak could not vanish and the maximum value of $\lim_{k\rightarrow 0^+}u^{[1]}(x,t)$ is $2|\Im(\lim_{k\rightarrow 0^+}\lambda_1)|>0$. Furthermore, collecting equation \eqref{eq:v-asy-R-1} and the Theorem \ref{theorem:asy-k}, we obtain that the elliptic-soliton solution under the $\cn$-type background would degenerate into the soliton solution.
	\item[(ii)] When $l=\frac{K'}{2}$, for any fixed $\lambda_1=\lambda(z_1)\in\ii \mathbb{R}$, $|\Im(\lambda_1)|<\frac{\alpha(1-k')}{2}$, we get $\lim_{k\rightarrow 0^+}\lambda(z_1)=0$, since $\Im(\lambda_1)\in (0,\frac{\alpha(1-k')}{2}) $ and $\lim_{k\rightarrow 0^+}\frac{\alpha(1-k')}{2}=0$ proved in Appendix \ref{appendix:conformal-approximation}. Combining with the maximum value analysis, we obtain that the peak would vanish as $k\rightarrow0^+$. Collecting equation \eqref{eq:v-asy-R-2}, we obtain that $\lim_{k\rightarrow 0^+}u^{[1]}(x,t)\rightarrow 0$, which reflects that the elliptic-soliton solution under the $\dn$-type background with $\lambda_1\in \ii \mathbb{R}$, $|\Im(\lambda_1)|<\frac{\alpha(1-k')}{2}$ would degenerate into the constant solution.
	\item[(iii)] When $l=\frac{K'}{2}$, for any fixed $\lambda_1\in \ii \mathbb{R}$, $|\Im(\lambda_1)|>\frac{\alpha(1+k')}{2}$, we get  $\lim_{k\rightarrow 0^+}\lambda(z_1)\in \left(\alpha,\infty\right)$, since $\Im(\lambda(z_1))\in \left(\frac{\alpha(1+k')}{2},\infty\right)$ and $\lim_{k\rightarrow 0^+}\frac{\alpha(1+k')}{2}=\alpha $. Thus, the peak could not vanish and $\Im(\lim_{k\rightarrow 0^+}\lambda_1)|>\alpha$. By equation \eqref{eq:v-asy-R-3} and the Theorem \ref{theorem:asy-k}, we obtain that the elliptic-soliton solution under the $\dn$-type background with $\lambda_1\in \ii \mathbb{R}$, $|\Im(\lambda_1)|>\frac{\alpha(1+k')}{2}$ would degenerate into the soliton solution.
\end{itemize} 

To illustrate the above situation more clearly, we provide the following special examples. Firstly, we consider the $\cn$-type background case with $l=0$. By choosing $z_1=K'+\frac{K}{3}\ii$, it is easy to verify that $z_1$ satisfies the case (i) and equation \eqref{eq:v-asy-R-1}. Plugging $l=0$, $\alpha=1$, $k=\frac{1}{100}$, $c=1$ into equation \eqref{eq:mKdV-solution-xi-n-1}, we obtain a soliton in Figure \ref{fig:cn-1b-k}(a). Comparing Figure \ref{fig:cn-1b-k}(a) with Figure \ref{fig:cn-1b-3d}(a) drawn by the modulus $k=\frac{1}{2}$ with the same parameter $z_1=K'+\frac{K}{3}\ii$, and we know that when the modulus $k$ changes from $\frac{1}{2}$ to $\frac{1}{100}$, the value of $v^{[1]}(x,t)$ in equation \eqref{eq:v-asy-R-1} tends to zero but the peak would not vanish, which provide a vividly description for the above analysis in case (i).

Then, we consider the $\dn$-type background case. Letting $z_1=\frac{K'}{2}+\ii\frac{2K}{5}$, we could verify that it satisfies the case (ii) and equation \eqref{eq:v-asy-R-2}. Plugging $l=\frac{K'}{2}$, $\alpha=1$, $k=\frac{2}{100}$, $c=1$ into equation \eqref{eq:mKdV-solution-xi-n-1}, we draw this solution in Figure \ref{fig:cn-1b-k}(b). Comparing with Figure \ref{fig:dn-1b-3d}(a) drawn by the modulus $k=\frac{9}{10}$, we find that the small peak in
Figure \ref{fig:dn-1b-3d}(a) gradually disappears in Figure \ref{fig:cn-1b-k}(b), when the modulus $k$ changes from $\frac{9}{10}$ to $\frac{2}{100}$. Furthermore, the amplitude is getting smaller and smaller with parameter $\lambda_1$ turning to zero. Thus, as $k$ changes from $\frac{9}{10}$ to $\frac{2}{100}$, the function $u^{[1]}(x,t)$ tends to be constant $\alpha=1$ with  the peak disappearing. This phenomenon confirms the above-mentioned analysis of case (ii).

When $z_1=-\frac{K'}{2}+\ii\frac{K}{5}$ and $l=\frac{K'}{2}$, it falls into the case (iii) and equation \eqref{eq:v-asy-R-3}. Plugging $l=\frac{K'}{2}$, $\alpha=1$, $k=\frac{2}{10}$, $c=1$ into equation \eqref{eq:mKdV-solution-xi-n-1}, we draw this soliton in Figure \ref{fig:cn-1b-k}(c). Comparing with Figure \ref{fig:dn-1b-3d}(b) drawn by the modulus $k=\frac{9}{10}$, we find that as the modulus $k$ changes from $\frac{9}{10}$ to $\frac{2}{10}$, the peak in
Figure \ref{fig:dn-1b-3d}(b) does not disappear in Figure \ref{fig:cn-1b-k}(c) and the function $v^{[1]}(x,t)$ in \eqref{eq:v-asy-R-3} tends to $-\alpha=-1$. This phenomenon confirms the above-mentioned analysis of case (iii).

\begin{figure}[h]
	\centering
	\includegraphics[width=1\linewidth]{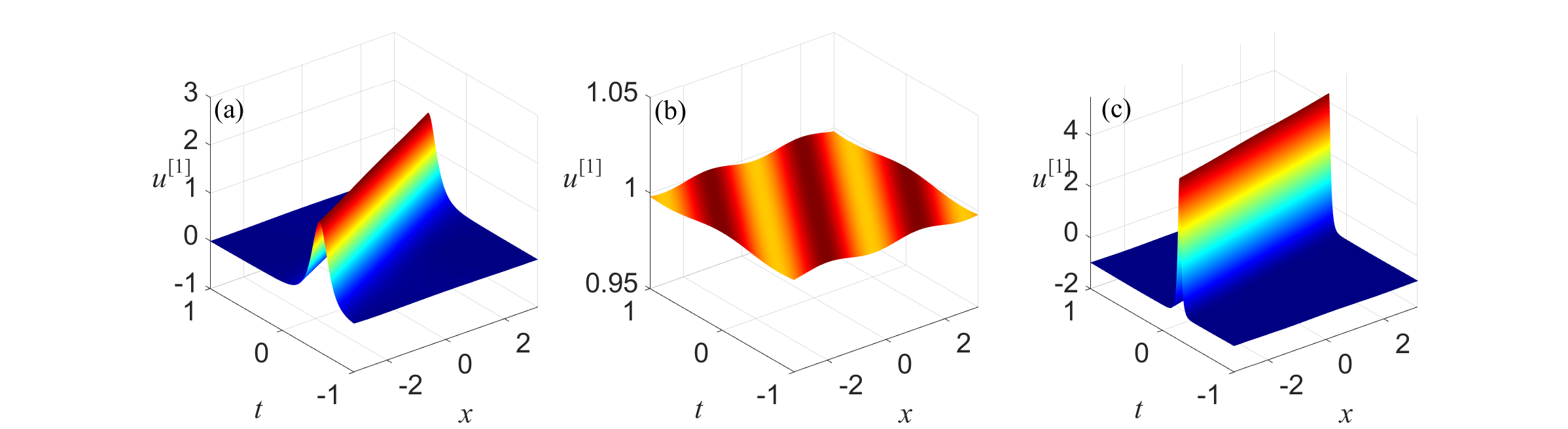}
	\caption{The 3d-plot of soliton solutions $u^{[1]}(x,t)$, as $k\rightarrow 0^+$.
	 (a): The solution is constructed by $k=\frac{1}{100}$, $\lambda_1\in \ii \mathbb{R}$, under $\cn$-type background. (b): The solution is constructed by $k=\frac{2}{100}$, $\lambda_1\in \ii \mathbb{R}$, $|\Im(\lambda_1)|<\frac{(1-k')}{2}$ under $\dn$-type background. (c): The solution is constructed by $k=\frac{2}{10}$, $\lambda_1\in \ii \mathbb{R}$, $|\Im(\lambda_1)|>\frac{(1+k')}{2}$ under the $\dn$-type background.}\label{fig:cn-1b-k}
\end{figure}

Similarly, we obtain that the multi elliptic-breather solutions could degenerate into the multi breather solutions as $k\rightarrow 0^+$ (in Figure \ref{fig:asy-k}). The breather solutions $u^{[1]}(x,t)$ in Figure \ref{fig:asy-k}(a) are obtained by $k=\frac{1}{20}$, $l=0$, $z_1=-\frac{9K'}{10}+\ii \frac{K}{4}$, $\alpha=1$, $c_1=3+4\ii$, $\lambda_1\approx 0.696 - 0.612\ii$. Consider the function $u^{[1]}(x,t)$ in region $R_{1}^{\pm}$. From the Theorem \ref{theorem:asy-R-elliptic}, we know that as $k\rightarrow 0^+$, the function $u^{[1]}(x,t;R_1^{\pm})$ tends to zero by equation \eqref{eq:solution-u-n-infy-R-cn}, if the solution $u^{[1]}(x,t)$ is constructed under the $\cn$-type background. The Figure \ref{fig:asy-k}(b) is obtained by the $\dn$-type background with $k=\frac{1}{10}$. Choosing $z_1=-\frac{K'}{3}+\ii \frac{K}{3}$, $l=\frac{K'}{2}$, $\lambda_1\approx 0.244- 0.502\ii $, $c_1=1$, $\alpha=1$, we get a single breather solution $u^{[1]}(x,t)$ in Figure \ref{fig:asy-k}(b). As $k$ is changing to $\frac{1}{10}$, the function $u^{[1]}(x,t;R_1^{\pm})$ is approximating to the constant $\alpha=1$ by equation \eqref{eq:solution-u-n-infy-R-dn} in Theorem \ref{theorem:asy-R-elliptic}. The breather solution $u^{[2]}(x,t)$ in Figure \ref{fig:asy-k}(c) is obtained by parameters $k=\frac{1}{20}$, $l=0$, $z_1=-\frac{9K'}{10}+\ii \frac{K}{4}$, $z_2=\frac{5K'}{6}+\ii\frac{ 2K}{5}$, $\alpha=1$, $c_1=3+4\ii $, $c_2=1.38 -\ii$, $\lambda_1=0.706 - 0.504\ii $, $\lambda_2=-0.519- 0.243\ii$. The same as the case in Figure \ref{fig:asy-k}(a), the function $u^{[1]}(x,t;R_1^{\pm})\rightarrow 0,$ as $k\rightarrow 0^+$.
 
\begin{figure}[h]
	\centering
	\subfigure {\includegraphics[width=0.33\linewidth]{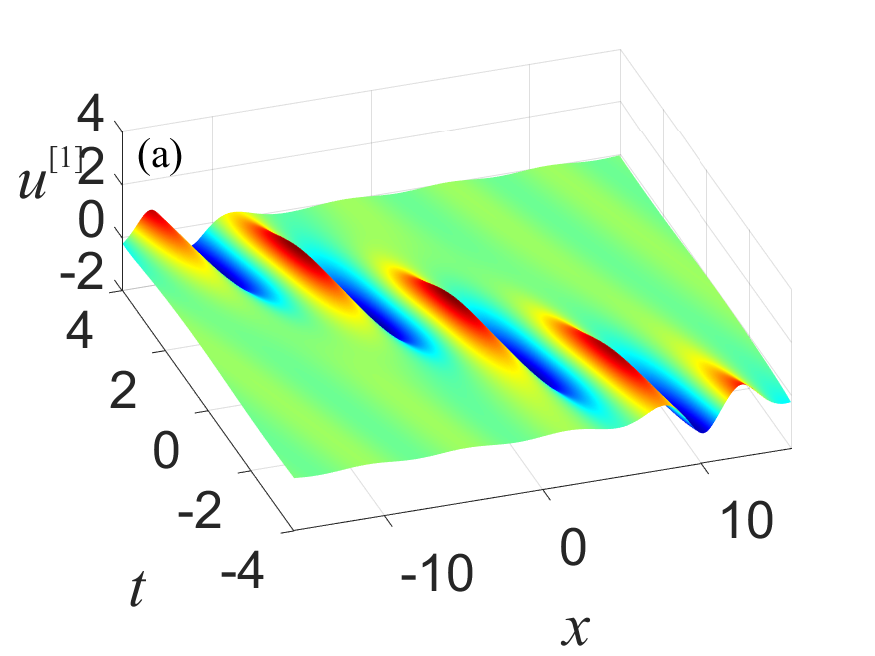}}
	\subfigure {\includegraphics[width=0.33\linewidth]{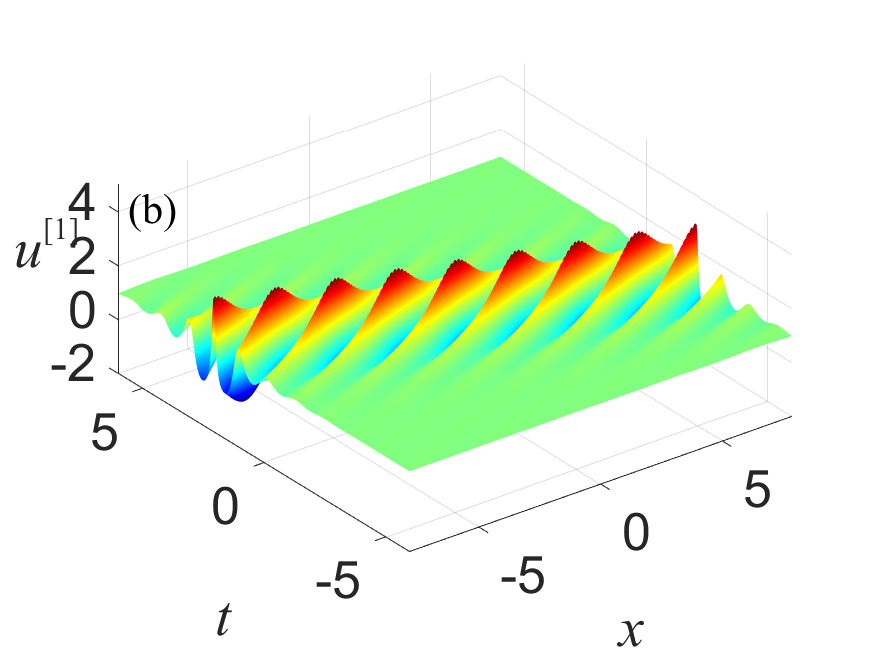}}
	\subfigure {\includegraphics[width=0.33\linewidth]{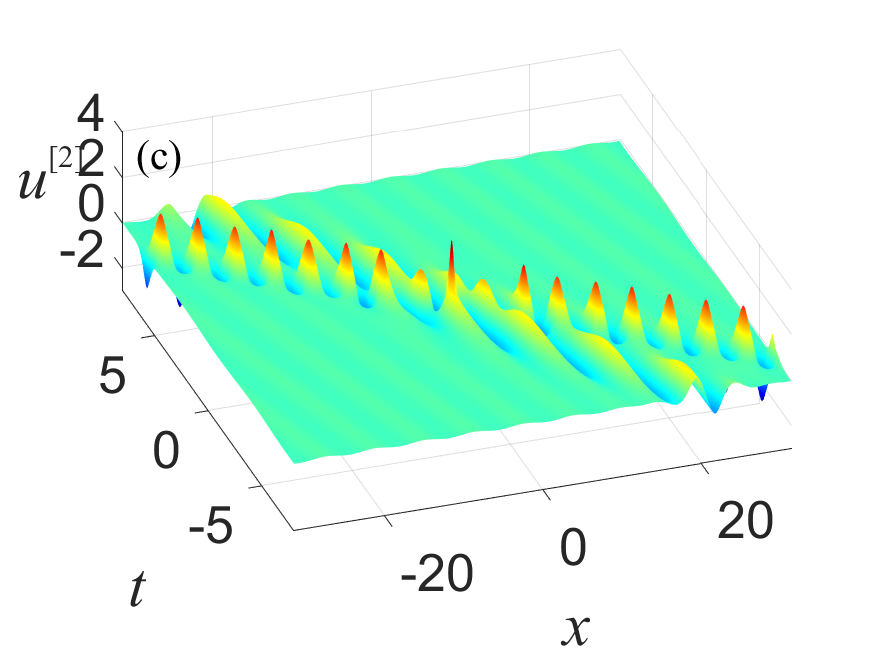}}
	\caption{The 3d-plot of solutions $u^{[1]}(x,t)$ or $u^{[2]}(x,t)$ for the mKdV equation \eqref{eq:mKdV-equation}, as $k\rightarrow 0^+$. (a): The breather solutions $u^{[1]}(x,t)$ with $k=\frac{1}{10}$ under the $\cn$-type background. (b): The breather solutions $u^{[1]}(x,t)$ with $k=\frac{1}{10}$, under the $\dn$-type background. (c): A two-breathers solution $u^{[2]}(x,t)$ with $k=\frac{1}{20}$ under the $\cn$-type background.}
	\label{fig:asy-k}
\end{figure}

Combining the above analysis, we get that as $k\rightarrow 0^+$ when $\lambda_i\in \ii \mathbb{R}$, $i=1,2,\cdots,N$, satisfies $|\Im(\lambda_i)|<\frac{\alpha(1-k')}{2}$, the multi elliptic-soliton solutions under the $\dn$-type background could degenerate into the constant solutions. Excepting the above case, as $k\rightarrow 0^+$, multi elliptic-soliton solutions, multi elliptic-breather solutions, and multi elliptic-soliton-breather solutions could degenerate into multi soliton solutions, multi breather solutions, and multi soliton-breather solutions, respectively. Then, we get Theorem \ref{theorem:Phi-k}.
Moreover, we obtain that the modulus $k$ is an important bridge connecting those two groups of solutions together. Furthermore, we could claim that we have extended the soliton and breather solutions under the constant or vanishing background into the elliptic functions background.

The breather solutions to the vanishing and constant backgrounds are well-known in the previous literature \cite{ChanL-1994,TajiriW-98}. By studying the maximum value of breathers and solitons solutions, it is easy to obtain that the height of the peak for solutions depends on the imaginary part of the spectral parameter $\lambda$. Under the elliptic solution background, as the modulus $k$ increases, the dynamic behavior of the nonlinearly superimposed soliton interacting with the spatial oscillations is more complex than the solitons or breathers under the constant backgrounds. For the peaks of solutions, the periodic background could influence the shape of the peaks, but it does not influence the height of the peaks. Thus, when we choose a small enough modulus $k$, the periodic background could be regarded as the constant background adding a small perturbation. The studies of the variation for these perturbations could be used to reflect the stable or unstable dynamic behaviors. As the evolution of time, if the perturbation is still small enough at any moment, the solution exhibits a stable dynamic behavior; if not, the solution exhibits an unstable dynamic behavior. These solutions could be taken as examples to support the stability analyses of breathers \cite{AlejoM-13,AlejoMP-17}, multi solitons \cite{Le-21} and so on.

\section{Conclusion and discussion}\label{sec:conclusion-discussion}

In this work, we systematically construct multi elliptic-localized solutions for the focusing mKdV equation using the Jacobi theta functions. And then, we provide the exact expressions to the multi elliptic-localized solutions under the different elliptic function solution backgrounds and perform uniform processing. 

In addition, we study the asymptotic behaviors of the multi elliptic-localized solutions along the evolution direction of solitons and breathers as well as on the region between the solitons and breathers. The collision between the breathers and solitons of the multi elliptic-localized solutions is elastic. For the multi elliptic-localized solutions $u^{[N]}(x,t)$, when $c_i=1,i=1,2,\cdots,N$, the collision between the breathers and solitons of solutions is strictly elastic. Moreover, we draw the profile of the multi elliptic-localized solutions before and after the collision to show the dynamic behavior of strictly elastic collision. Furthermore, the multi elliptic-localized solutions could degenerate as the breathers and solitons as $k\rightarrow 0^+$. For the different values of spectral parameter $\lambda$ depending on the uniform parameter $z$, the dynamic behavior of elliptic-localized solutions is shown as $k\rightarrow 0^+$. The above method can be easily extended to the other integrable models, such as the coupled NLS equation and the sine-Gordon equation.

The multi elliptic-localized solutions obtained here could motivate more studies in theories and experiments. Lamb et al. \cite{LambPTPXK-07} revealed that breathers were generated under several initial disturbances in internal wave records in symmetrically stratified fluids. Through experimental observations, Xu et al.\cite{XuGCZK-19} confirmed that the shaped molecular breathing light waves propagate in an almost conservative fiber optic system. In addition, the breathers also appear in the wave field dynamics \cite{Takahashi-16}.

The degenerated analysis among the branch points has not been studied. The rogue wave of the NLS equation in Jacobi theta function form is obtained by Feng et al. \cite{FengLT-19}, but the rogue wave of the mKdV equation in theta function form is not provided. The stability analysis of the elliptic function solutions of the focusing NLS equation and focusing mKdV equation is obtained in \cite{DeconinckyU-20,LinglmS-21}. Naturally, we are interested in the stability of multi elliptic-localized solutions. We will analyze the works mentioned above in the near future.

\section*{Acknowledgement}

This work is supported by the National Natural Science Foundation of China (Grant No. 12122105).

\appendix

\titleformat{\section}[display]
{\centering\LARGE\bfseries}{ }{11pt}{\LARGE}

\renewcommand{\appendixname}{Appendix \ \Alph{section} }

\section{\appendixname. The definitions and properties of elliptic functions}\label{appendix:Elliptic functions}

\setcounter{equation}{0}
\renewcommand\theequation{\Alph{section}.\arabic{equation}}
\setcounter{figure}{0}
\renewcommand\thefigure{\Alph{section}.\arabic{figure}}

\renewcommand\thedefinition{\Alph{section}.\arabic{definition}}

\renewcommand\thelemma{\Alph{section}.\arabic{lemma}}

\subsection*{Jacobi elliptic function}

The functions $K(k)$ and $E(k)$ in the above equations are called the first and second complete elliptic integrals \cite{ByrdF-54}, which are defined as
\begin{equation}\label{eq:J-K-E-int}
	K\equiv K(k)=\int_{0}^{\frac{\pi}{2}}\frac{\dd \theta}{\sqrt{1-k^2\sin^2\theta}},\quad \text{and} \quad
	E\equiv E(k) =\int_{0}^{\frac{\pi}{2}}\sqrt{1-k^2\sin^2\theta}\dd \theta.
\end{equation}
In addition to the above two integrals, we usually use associated complete elliptic integrals $K'=K(k'), k'=\sqrt{1-k^2}$. Then, we list useful formulas in \cite{ByrdF-54} as follows:
\begin{itemize}
	\item Zeros and poles of Jacobi elliptic functions:
	\begin{equation}\label{eq:J 0in}
		\begin{split}
			\sn(2mK+(2n+1)\ii K')=\infty, \qquad & \sn(2mK+2 n\ii K')=0,\\
			\cn(2mK+(2n+1)\ii K')=\infty,\qquad &
			\cn((2m+1)K+2n\ii K')=0,\\
			\dn(2mK+(2n+1)\ii K')=\infty, \qquad &
			\dn((2m+1)K+(2n+1)\ii K')=0,
		\end{split}
	\end{equation}
	where $n$ and $m$ are any integers including zero;
	\item Shift formulas:
	\begin{equation}\label{eq:Jacobi-shift}
		\begin{split}
			&\sn(u+K)=\cd(u),\qquad \sn(u+\ii K')=\ns(u)/k, \qquad \qquad  \sn(u+K+\ii K')=\dc(u)/k, \\
			&\cn(u+K)=-k'\sd(u), \quad \cn(u+\ii K')=-\ii \ds(u)/k, \qquad 
			\cn(u+K+\ii K')=-\ii k'\nc(u)/k, \\ &\dn(u+K)=k'\nd(u), \qquad \dn(u+\ii K')=-\ii \cs(u),\qquad \quad 
			\dn(u+K+\ii K')=\ii k' \tn(u);
		\end{split}
	\end{equation}
	\item Imaginary arguments:
	\begin{equation}\label{eq:Jacobi-I}
		\sn(\ii u,k) =\ii \tn(u,k'),\quad
		\cn(\ii u,k) = \nc(u,k'),\quad
		\dn(\ii u,k) = \dc(u,k'), \quad k^2+k'^2=1;
	\end{equation}
	\item Half arguments:
\begin{equation}\label{eq:Jacobi-half}
	\sn^2\left(\frac{u}{2}\right) =\frac{1-\cn(u)}{1+\dn(u)},\quad
	\cn^2\left(\frac{u}{2}\right) =\frac{\dn(u)+\cn(u)}{1+\dn(u)},\quad
	\dn^2\left(\frac{u}{2}\right) =\frac{\dn(u)+\cn(u)}{1+\cn(u)};
\end{equation}
	\item Approximation formulas:
	\begin{equation}\label{eq:cd-app}
		\cn(u,k)\approx \cos u+k^2\sin u(u-\sin u\cos u)/4,\qquad \dn(u,k)\approx 1-(k^2\sin^2u)/2, \qquad k\ll 1.
	\end{equation}
\end{itemize}

\subsection*{Jacobi theta function}
The definition of the Jacobi theta function in \cite{ArmitageE-06} is as follows:
\begin{definition}\label{define:theta}
The Jacobi theta functions are defined as the summation:
	\begin{equation}\label{eq:theta1234}
		\begin{split}
			\vartheta_1(z,q)=&\ii\sum_{n=-\infty}^{+\infty}(-1)^{n}q^{\left(n-\frac{1}{2}\right)^2}\ee^{(2n-1)\ii\pi z},\qquad \qquad 
			\vartheta_3(z,q)=\sum_{n=-\infty}^{+\infty}q^{n^2}\ee^{2n\ii\pi z}, \\
			\vartheta_2(z,q)=&	\sum_{n=-\infty}^{+\infty}q^{\left(n-\frac{1}{2}\right)^2}\ee^{(2n-1)\ii\pi z}, \qquad \qquad 
			\vartheta_4(z,q)=\sum_{n=-\infty}^{+\infty}(-1)^{n}q^{n^2}\ee^{2n\ii\pi z},
		\end{split}
	\end{equation}
	where $q=e^{\ii \pi \tau}$, $\tau=\frac{\ii K'}{K}$.
\end{definition}
We usually omit the parameter $q$ in the jacobi theta functions: $\vartheta_i(z)\equiv\vartheta_i(z,q)$, $i=1,2,\cdots,4$. Furthermore, if $z=0$, we omit the $0$, i.e., $\vartheta_i=\vartheta_i(0)$.
There are many relationships among the above four theta functions. Here, we just provide common formulas:
\begin{itemize}
	\item Shift formulas among four theta functions in \cite{ArmitageE-06}:
	\begin{equation}\label{eq:Jacobi Theta-K iK}
		\begin{split}
			\vartheta_1(z)&=-\vartheta_2\left(z+\frac{1}{2}\right)=-\ii M \vartheta_3\left(z+\frac{1}{2}+\frac{\tau}{2}\right)=-\ii M\vartheta_4\left(z+\frac{\tau}{2}\right),\\
			\vartheta_2(z)&=\vartheta_1\left(z+\frac{1}{2}\right)=M \vartheta_4\left(z+\frac{1}{2}+\frac{\tau}{2}\right)= M\vartheta_3\left(z+\frac{\tau}{2}\right),\\
			\vartheta_3(z)&=\vartheta_4\left(z+\frac{1}{2}\right)= M \vartheta_1\left(z+\frac{1}{2}+\frac{\tau}{2}\right)= M\vartheta_2\left(z+\frac{\tau}{2}\right),\\
			\vartheta_4(z)&=\vartheta_3\left(z+\frac{1}{2}\right)=\ii M \vartheta_2\left(z+\frac{1}{2}+\frac{\tau}{2}\right)=-\ii M\vartheta_1\left(z+\frac{\tau}{2}\right),
		\end{split}
	\end{equation}
	where $M=q^{1/4}\ee^{\ii \pi z}$.
	\item Conversion formulas between Jacobi theta functions and elliptic functions in \cite{ByrdF-54}:
	\begin{equation}\label{eq:sn vartheta}
		\begin{split}
			\sn(u)=\frac{\vartheta_3\vartheta_1(z)}{\vartheta_2\vartheta_4(z)}, \quad \cn(u)&=\frac{\vartheta_4\vartheta_2(z)}{\vartheta_2\vartheta_4(z)}, \quad
			\dn(u)=\frac{\vartheta_4\vartheta_3(z)}{\vartheta_3\vartheta_4(z)}, 
		\end{split}
	\end{equation}
	where $z=\frac{u}{2K}$ and $ k^{1/2}=\frac{\vartheta_2}{\vartheta_3}$;
	\item Weierstrass addition formulas (or Fay’s identities) are given by \cite{KharchevZ-15}:
	\begin{itemize}
		\item[(1)] Complimentary system:
		\begin{equation}\label{eq:Jacobi Theta}
			\begin{split}
				&\vartheta_k(u+v)\vartheta_k(u-v)\vartheta_l(x+y)\vartheta_l(x-y)\\
				=&\vartheta_i(v+x)\vartheta_i(v-x)\vartheta_j(u+y)\vartheta_j(u-y)
				-\vartheta_i(u+x)\vartheta_i(u-x)\vartheta_j(v+y)\vartheta_j(v-y),
			\end{split}
		\end{equation}
		where the combinations of $[k,l,i,j]$ are $[1,4,2,3]$, $[1,3,2,4]$ and $[1,2,3,4]$;
		\item[(2)] Mixed identities:
		\begin{equation}\label{eq:Jacobi Theta uv}
			\begin{split}
				&\vartheta_1(u+x)\vartheta_2(u-x)\vartheta_3(v+y)\vartheta_4(v-y)
				-\vartheta_1(u-y)\vartheta_2(u+y)\vartheta_3(v-x)\vartheta_4(v+x)\\
				=&\vartheta_1(x+y)\vartheta_2(x-y)\vartheta_3(u+v)\vartheta_4(u-v).
			\end{split}
		\end{equation}
	\end{itemize}
\end{itemize}

\subsection*{Jacobi Zeta function}
The definition of the Jacobi Zeta function in \cite{ByrdF-54} is as follows:
\begin{definition}\label{define:Zeta}
	The Jacobi Zeta function is defined by 
	\begin{equation}\label{eq:zeta}
		Z(u)\equiv \int_0^u \left( \dn^2(v)-\frac{E}{K} \right) \dd v,
	\end{equation}
	where $E\equiv E(k), K\equiv K(k)$ is the complete elliptic integrals defined in \eqref{eq:J-K-E-int}.
\end{definition}

\section{\appendixname. Darboux-B\"acklund transformation for mKdV equation}\label{appendix:Darboux transformation}

\setcounter{equation}{0}
\renewcommand\theequation{\Alph{section}.\arabic{equation}}
\setcounter{theorem}{0}
\renewcommand\thetheorem{\Alph{section}.\arabic{theorem}}
\renewcommand\thelemma{\Alph{section}.\arabic{lemma}}

In this Appendix, we give the fundamental properties of Darboux-B\"acklund transformation. 

\newenvironment{aproof-T}{\emph{Proof of Theorem \ref{theorem:u-N-breather}.}}{\hfill$\Box$\medskip}
\begin{aproof-T}
	Based on the Darboux transformation, the $N$ iterations of the elementary Darboux transformation lead to the $N$-fold Darboux matrix \eqref{eq:multi-T}, and it also could be rewritten as
	\begin{equation}\label{eq:T-n-m}
		\mathbf{T}^{[N]}(\lambda; x,t)=\mathbb{I}-\mathbf{X}_m\mathbf{M}_m( x,t)^{-1}\mathbf{D}_m^{-1}\mathbf{X}_m^{\dagger},
	\end{equation}
where matrices $\mathbf{M}_m( x,t)$ and $\mathbf{X}_m$ are defined in \eqref{eq:define-M-X}, and
\begin{equation}\label{eq:define-D}
	\mathbf{D}_m={\rm diag}(\lambda-\lambda_1^*,\lambda-\lambda_2^*,\cdots ,\lambda-\lambda_m^*).
\end{equation}

Firstly, we verify
\begin{equation}\label{eq:T-n-n}
	\mathbf{T}^{[N]}(\lambda; x,t)=\mathbf{T}^{\mathrm{P}}_{N}(\lambda; x,t)\cdots\mathbf{T}^{\mathrm{P}}_{2}(\lambda; x,t)\mathbf{T}^{\mathrm{P}}_{1}(\lambda; x,t)=\mathbb{I}-\ii\mathbf{X}_N\mathbf{M}_N( x,t)^{-1}\mathbf{D}_N^{-1}\mathbf{X}_N^{\dagger},
\end{equation}
where $\mathbf{T}^{\mathrm{P}}_1(\lambda; x,t)$ and $\mathbf{T}^{\mathrm{P}}_i(\lambda; x,t)$ are defined in \eqref{eq:T-1} and \eqref{eq:T-i-P-C}, and matrices $\mathbf{M}_N( x,t)$ and $\mathbf{X}_N$ are defined in equation \eqref{eq:define-M-X} and
$\mathbf{D}_N(\lambda)$ is defined in equation \eqref{eq:define-D}.
Combing equation \eqref{eq:T-n-n}, we know
\begin{equation}
	\begin{split}
		\mathbf{T}^{[N]}(\lambda_i; x,t)\Phi_i
		=&\mathbf{T}^{\mathrm{P}}_{N}(\lambda_i; x,t)\cdots\mathbf{T}^{\mathrm{P}}_{i+1}(\lambda_i; x,t)\left(\mathbb{I}-\frac{ \Phi_i^{[i-1]}(\Phi_i^{[i-1]})^{\dagger}}{(\Phi_i^{[i-1]})^{\dagger}\Phi_i^{[i-1]}}\right)\mathbf{T}^{\mathrm{P}}_{i-1}(\lambda_i; x,t)\cdots\mathbf{T}^{\mathrm{P}}_{1}(\lambda_i; x,t) \Phi_i\\
		=&\mathbf{T}^{\mathrm{P}}_{N}(\lambda_i; x,t)\cdots\mathbf{T}^{\mathrm{P}}_{i+1}(\lambda_i; x,t)\left(\mathbb{I}-\frac{ \Phi_i^{[i-1]}(\Phi_i^{[i-1]})^{\dagger}}{(\Phi_i^{[i-1]})^{\dagger}\Phi_i^{[i-1]}}\right) \Phi_i^{[i-1]}=0, \qquad  i=1,2,\cdots,N,
	\end{split}
\end{equation}
which means that $\Phi_i^{[i-1]}$ is a kernel of matrix $\mathbf{T}^{[N]}(\lambda_i;x,t)$. Combining the poles of matrix $\mathbf{T}^{[N]}(\lambda;x,t)$ and the Liouville theorem, we could set that the matrix $\mathbf{T}^{[N]}(\lambda; x,t)$ has the ansatz:
\begin{equation}\label{eq:T-ansatz}
	\mathbf{T}^{[N]}(\lambda; x,t)=\mathbb{I}+\sum_{i=1}^{n}\frac{|x_i\rangle\langle y_i|}{\lambda-\lambda_i^*},\qquad \mathbf{T}^{[N]\dagger}(\lambda^*; x,t)=\mathbb{I}+\sum_{i=1}^{n}\frac{|y_i\rangle\langle x_i|}{\lambda-\lambda_i},
\end{equation}
where $| y_i\rangle=(\langle y_i|)^{\dagger}$ and $\langle y_i|$ is a vector with two rows and one column. Since $\mathbf{T}^{[N]}(\lambda; x,t)$ satisfies the symmetry \eqref{eq:T-sym}, 
\begin{equation}
	\mathop{\mathrm{Res}}\limits_{\lambda=\lambda_i}(\mathbf{T}^{[N]}(\lambda_i; x,t)(\mathbf{T}^{[N]}(\lambda_i; x,t))^{-1})=
	\mathop{\mathrm{Res}}\limits_{\lambda=\lambda_i}(\mathbf{T}^{[N]}(\lambda_i; x,t)(\mathbf{T}^{[N]}(\lambda_i^*; x,t))^{\dagger})	=\mathbf{T}^{[N]}(\lambda_i; x,t)|y_i\rangle\langle x_i|=0,
\end{equation}
which means $\mathbf{T}^{[N]}(\lambda_i; x,t)|y_i\rangle=0, i=1,2,\cdots N$. Because $\mathbf{T}^{[N]}(\lambda_i; x,t)$ is a non-zero two-dimensional matrix, without loss of generality, we could set $|y_i\rangle= \Phi_i$. Plugging $|y_i\rangle= \Phi_i$ into equation \eqref{eq:T-ansatz}, we gain
\begin{equation}
	\begin{pmatrix}
		 \Phi_1, &  \Phi_2, & \cdots, &  \Phi_N
	\end{pmatrix}
=-\begin{pmatrix}
	|x_1\rangle, & |x_2\rangle, & \cdots, &|x_N\rangle
\end{pmatrix}\mathbf{M}_N( x,t),
\end{equation}
which means 
\begin{equation}
	\begin{pmatrix}
		|x_1\rangle, & |x_2\rangle, & \cdots,& |x_N\rangle
	\end{pmatrix}
	=-\begin{pmatrix}
		 \Phi_1, &  \Phi_2, & \cdots, &  \Phi_N
	\end{pmatrix}
	(\mathbf{M}_N( x,t))^{-1}.
	\end{equation}
Therefore, we rewritten equation \eqref{eq:T-ansatz} as
\begin{equation}
	\begin{split}
	  \mathbf{T}^{[N]}(\lambda; x,t)=&\mathbb{I}-\mathbf{X}_N
	  \mathbf{M}^{-1}_N( x,t)\mathbf{D}_N^{-1} \mathbf{X}_N^{\dagger},
	\end{split}
\end{equation}
where $\mathbf{M}_N(x,t)$ and $\mathbf{X}_N$ are defined in \eqref{eq:define-M-X} and $\mathbf{D}_N$ is defined in \eqref{eq:define-D}.

 Then, we consider the fixed conditions, including $\mathbf{T}_i^{\mathrm{P}}(\lambda;x,t)$ and $\mathbf{T}_i^{\mathrm{C}}(\lambda;x,t)$. For the Darboux matrix $\mathbf{T}_i^{\mathrm{C}}(\lambda;x,t)$, we could divide it into two transformations as 
 \begin{equation}\label{eq:T-matrix}
 	\mathbf{T}_i^{\mathrm{C}}(\lambda;x,t)=\left(\mathbb{I}-\frac{\lambda_i-\lambda_i^*}{\lambda+\lambda_i}\frac{ \hat{\Phi}_i(\hat{\Phi}_i)^{\dagger}}{(\hat{\Phi}_i)^{\dagger}\hat{\Phi}_i}\right)\left(\mathbb{I}-\frac{\lambda_i-\lambda_i^*}{\lambda-\lambda_i^*}\frac{ \Phi_i\Phi_i^{\dagger}}{\Phi_i^{\dagger}\Phi_i}\right),
 \end{equation}
where $\hat{\Phi}_i=\left.\left(\mathbb{I}-\frac{\lambda_i-\lambda_i^*}{\lambda-\lambda_i^*}\frac{ \Phi_i\Phi_i^{\dagger}}{\Phi_i^{\dagger}\Phi_i}\right)\Phi(x,t;\lambda)\mathbf{c}_{i}^*\right|_{\lambda=-\lambda_i^*}$.
Noticing the matrix $\mathbf{T}_i^{\mathrm{C}}(\lambda;x,t)$ as a two-fold Darboux transformation, we obtain equation \eqref{eq:T-n-m} with the dimension of matrices $\mathbf{M}_m( x,t)$ and $\mathbf{D}_m$ satisfying $N\le m=n_1+2n_2\le 2N$.
Based on the B\"{a}cklund transformation, the solution $u^{[N]}(x,t)$ with multi elliptic-localized solutions could be obtained by \eqref{eq:u-N-breather}.
\end{aproof-T}

\section{\appendixname. The conformal map and the asymptotic analysis}\label{appendix:conformal-approximation}

In this section, we briefly review the conformal mapping studied in \cite{LinglmS-21}, because this helps us immensely with the asymptotic analysis as $k\rightarrow 0^+$. And then, we provide the preliminary results before the asymptotic analysis.

The conformal map $\lambda(z)$ maps a rectangular region $S$ to the entire complex plane $\lambda$-plane with two cuts. More details are given in reference \cite{LinglmS-21}. Here, we consider the specific points $\hat{z}_i=\pm \frac{K'}{2}\pm \ii \frac{K}{2}+l,i=1,2,3,4$, which play an important role in the cuts in this conformal map. The schematic of this conformal map is given in Figure \ref{fig:conformal} when $l=0$ and in Figure \ref{fig:conformal-K} when $l=\frac{K'}{2}$.

\begin{figure}[h]
	\centering
	\includegraphics[width=0.7\linewidth]{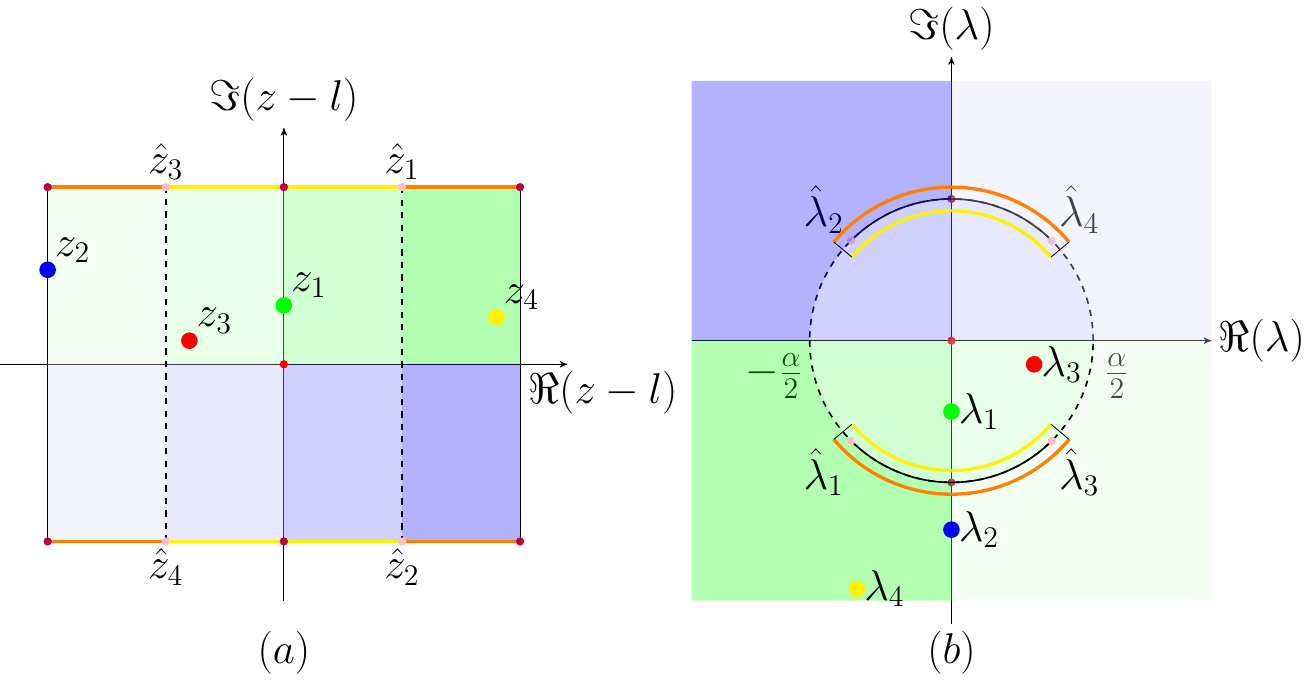}
	\caption{The conformal map $\lambda(z)$ in equation \eqref{eq:lambda-elliptic-0} with $l=0$.}
	\label{fig:conformal}
\end{figure}

When $l=0$, we consider the four points $\hat{z}_i=\pm \frac{K'}{2}\pm \ii \frac{K}{2},i=1,2,3,4$ that determine the endpoint of cuts in $\lambda$-plane. Combining with equation \eqref{eq:lambda-elliptic-0}, we obtain that those four points satisfy the equations 
\begin{equation}\label{eq:lambda-i-0}
	\begin{split}
		\hat{\lambda}_1=&\lambda(\hat{z}_1)=\lambda\left(\frac{K'}{2}+ \ii \frac{K}{2}\right)=\frac{\ii \alpha }{2}\left(-k+\ii k'\right), \qquad \quad \hat{\lambda}_2=\lambda(\hat{z}_2)=\lambda\left(\frac{K'}{2}- \ii \frac{K}{2}\right)=\frac{\ii \alpha }{2}\left(k+\ii k'\right),\\
		\hat{\lambda}_3=&\lambda(\hat{z}_3)=\lambda\left(-\frac{K'}{2}+ \ii \frac{K}{2}\right)=\frac{\ii \alpha }{2}\left(-k-\ii k'\right),  \qquad  \hat{\lambda}_4=\lambda(\hat{z}_4)=\lambda\left(-\frac{K'}{2}- \ii \frac{K}{2}\right)=\frac{\ii \alpha }{2}\left(k-\ii k'\right),
	\end{split}
\end{equation}
where $k'=\sqrt{1-k^2}$. Then, we analyze the changes in cuts, as $k\rightarrow 0^+$. For any point on the cut in $\lambda$-plane, there must exist two points $z_1\neq z_2$ on the lines $z\in \left\{z\in S\left|z=z_R\pm \ii \frac{K}{2}\right.\right\}$ satisfying $\lambda(z_1)=\lambda(z_2)$. Therefore, we mainly consider the above two lines and obtain the Lemma \ref{lemma:cn-K1}. 
\begin{lemma}\label{lemma:cn-K1}
	When $l=0$, for any $k\in (0,1)$ and $z\in \left\{z\in S\left|z=z_R\pm \ii \frac{K}{2}\right.\right\}$, the value of function $\lambda(z)$ is on the circle centered at the origin with radius $\frac{\alpha}{2}$.
\end{lemma}
\begin{proof}
	For any $z\in \left\{z\in S\left|z=z_R\pm \ii \frac{K}{2}\right.\right\}$, utilizing formulas \eqref{eq:Jacobi-shift}, \eqref{eq:Jacobi-I} and \eqref{eq:Jacobi-half}, we know that 
	\begin{equation}
		\begin{split}
			\lambda\left(z_R\pm\ii \frac{K}{2}\right)
			\xlongequal[ ]{\eqref{eq:Jacobi-half}}&\frac{\ii \alpha }{2}\frac{1-\cn(\mp K+2\ii z_R)}{\sn(\mp K+2\ii z_R)}\\
			\xlongequal[ ]{\eqref{eq:Jacobi-shift}}&\frac{\ii \alpha }{2}\frac{\dn(2\ii z_R)\pm k'\sn(2\ii z_R)}{\pm\cn(2\ii z_R)}\\
			\xlongequal[ ]{\eqref{eq:Jacobi-I}}&\pm\frac{\ii \alpha }{2}\left(\dn(2 z_R,k')\pm\ii k' \sn(2z_R,k')\right),
		\end{split}
	\end{equation}
	which implies that for any $k\in (0,1)$ and $z\in \left\{z\in S\left|z=z_R\pm \ii \frac{K}{2}\right.\right\}$, the value of $\lambda(z)$ must satisfy $|\lambda(z)|=\frac{\alpha}{2}$. Thus, $\lambda(z)$ is on the circle centered at the origin with radius $\frac{\alpha}{2}$.
\end{proof}

Then, we consider the case $l=\frac{K'}{2}$. By the definition of $\lambda(z)$ in equation \eqref{eq:lambda-elliptic-K}, it is easy to verify that
\begin{equation}\label{eq:lambda-i-K1}
	\begin{split}
		\hat{\lambda}_1=&\lambda(\hat{z}_1)=\lambda\left(-\frac{K'}{2}+ \ii \frac{K}{2}\right)=-\frac{\ii \alpha }{2}\left(1+k'\right), \qquad \hat{\lambda}_2=\lambda(\hat{z}_2)=\lambda\left(-\frac{K'}{2}- \ii \frac{K}{2}\right)=\frac{\ii \alpha }{2}\left(1+k'\right),\\
		\hat{\lambda}_3=&\lambda(\hat{z}_3)=\lambda\left(\frac{K'}{2}+ \ii \frac{K}{2}\right)=-\frac{\ii \alpha }{2}\left(1-k'\right), \qquad \hat{\lambda}_4=\lambda(\hat{z}_4)=\lambda\left(\frac{K'}{2}- \ii \frac{K}{2}\right)=\frac{\ii \alpha }{2}\left(1-k'\right),\\
	\end{split}
\end{equation}
 where $k'=\sqrt{1-k^2}$.

\begin{figure}[h]
	\centering
	\includegraphics[width=0.7\linewidth]{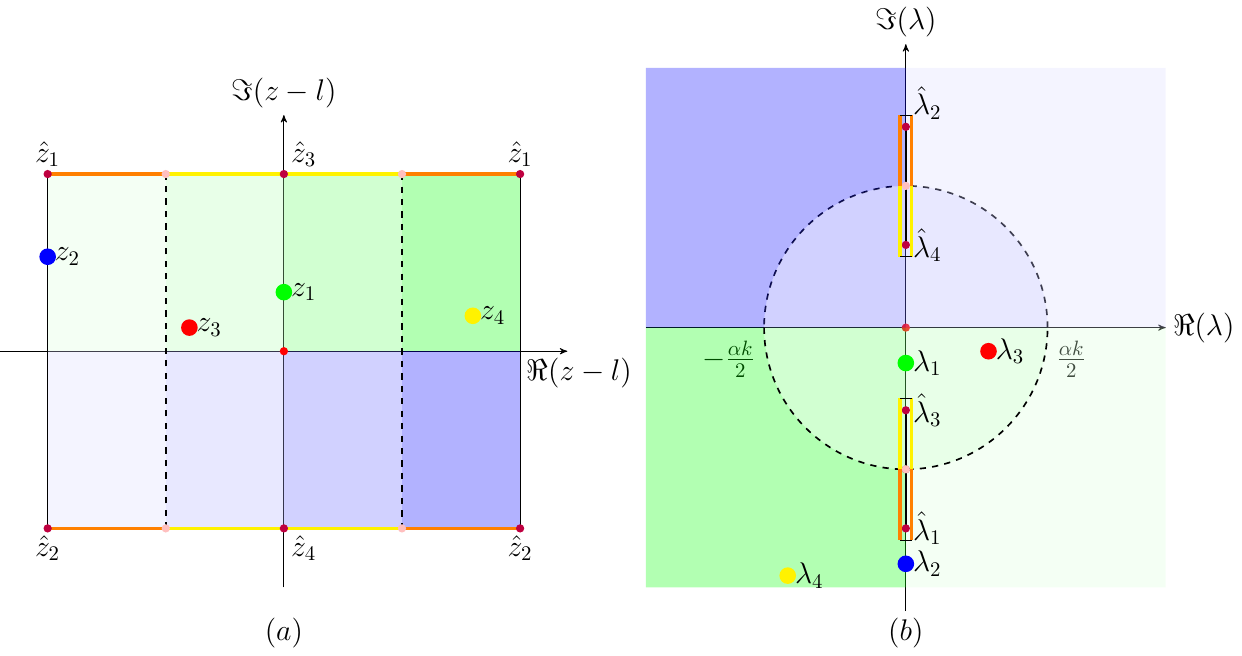}
	\caption{The conformal map $\lambda(z)$ in equation \eqref{eq:lambda-elliptic-K} with $l=\frac{K'}{2}$.}
	\label{fig:conformal-K}
\end{figure}

\begin{lemma}\label{lemma:dn-K1}
	When $l=\frac{K'}{2}$, for any $k\in (0,1)$ and $z\in \left\{z\in S\left|z=z_R\pm \ii \frac{K}{2}\right.\right\}$, the value of function $\lambda(z)$ is on the imaginary axis. 
\end{lemma}
\begin{proof}
	For any $z\in \left\{z\in S\left|z=z_R\pm \ii \frac{K}{2}\right.\right\}$, through utilizing the formulas \eqref{eq:Jacobi-shift}, we obtain
	\begin{equation}
		\begin{split}
			\lambda^*\left(z_R\pm\ii \frac{K}{2}\right)
			=&-\frac{\ii \alpha k^2}{2}\frac{\sn(\mp\frac{K}{2}-\ii z_R)\cn(\mp\frac{K}{2}-\ii z_R)}{\dn(\mp\frac{K}{2}-\ii z_R)}\\
			=&\frac{\ii \alpha k^2}{2}\frac{\sn(\pm\frac{K}{2}+\ii z_R)\cn(\pm\frac{K}{2}+\ii z_R)}{\dn(\pm\frac{K}{2}+\ii z_R)}\\
			=&-\frac{\ii \alpha k^2}{2}\frac{\sn(\mp\frac{K}{2}+\ii z_R)\cn(\mp\frac{K}{2}+\ii z_R)}{\dn(\mp\frac{K}{2}+\ii z_R)}\\
			=&-\lambda\left(z_R\pm\ii \frac{K}{2}\right),
		\end{split}
	\end{equation}
which means that for any $z\in \left\{z\in S\left|z=z_R\pm \ii \frac{K}{2}\right.\right\}$, $\lambda(z)\in \ii \mathbb{R}$. Thus,  the value of $\lambda(z)$ must on the imaginary axis for any $k\in(0,1)$.
\end{proof}

\bibliographystyle{siam}
\bibliography{references}

\end{document}